\documentclass[aps,prx,twocolumn,superscriptaddress,notitlepage,nopacs,nofootinbib,amsmath,amstex,amssymb,citeautoscript,longbibliography,floatfix,a4paper,twocolumn]{revtex4-2}

\usepackage{natbib}
\usepackage[english]{babel}
\usepackage{letltxmacro}
\usepackage{latexsym}
\LetLtxMacro{\ORIGselectlanguage}{\selectlanguage}
\makeatletter
\DeclareRobustCommand{\selectlanguage}[1]{%
  \@ifundefined{alias@\string#1}
    {\ORIGselectlanguage{#1}}
    {\begingroup\edef\x{\endgroup
       \noexpand\ORIGselectlanguage{\@nameuse{alias@#1}}}\x}%
}
\newcommand{\definelanguagealias}[2]{%
  \@namedef{alias@#1}{#2}%
}
\makeatother
\definelanguagealias{en}{english}
\definelanguagealias{English}{english}

\newcommand{\be}{\begin{equation}}
\newcommand{\ee}{\end{equation}}
\newcommand{\bea}{\begin{eqnarray}}
\newcommand{\eea}{\end{eqnarray}}

\newcommand{\tr}{\mathop{\rm tr}}

\newcommand{\im}{{\rm i}}

\usepackage{amsthm}

\newtheorem{theorem}{Theorem}
\newtheorem{corollary}{Corollary}[theorem]

\newtheorem{definition}[theorem]{Definition}

\usepackage{graphicx}
\usepackage{bm}
\usepackage{physics}
\usepackage[dvipsnames]{xcolor}

\usepackage{algorithm}
\usepackage{algpseudocode}

\makeatletter
\newcommand{\printfnsymbol}[1]{%
  \textsuperscript{\@fnsymbol{#1}}%
}
\makeatother

\usepackage{hyperref}
\hypersetup{
    unicode=false,          
    pdftoolbar=false,        
    pdfmenubar=true,        
    pdffitwindow=false,     
    pdfstartview={FitH},    
    pdftitle={},    
    pdfauthor={Stefan H. Sack and Raimel Medina},     
    pdfsubject={},   
    pdfcreator={},   
    pdfproducer={}, 
    pdfkeywords={}, 
    pdfnewwindow=true,      
    colorlinks=true,       
    linkcolor=black,          
    citecolor=magenta,        
    filecolor=magenta,      
    urlcolor=magenta           
}

\begin{document}

\title{Avoiding barren plateaus using classical shadows}

\author{Stefan H. Sack}
\thanks{Equal contribution}
\email{stefan.sack@ist.ac.at}

\author{Raimel A. Medina}
\thanks{Equal contribution}
\email{raimel.medina@ist.ac.at}
\affiliation{IST Austria, Am Campus 1, 3400 Klosterneuburg, Austria}

\author{Alexios A. Michailidis}
\affiliation{IST Austria, Am Campus 1, 3400 Klosterneuburg, Austria}
\affiliation{Department of Theoretical Physics, University of Geneva,
24 quai Ernest-Ansermet, 1211 Geneva, Switzerland}

\author{Richard Kueng}
\affiliation{Institute for Integrated Circuits, Johannes Kepler University Linz, Altenberger Straße 69, Austria}
\author{Maksym Serbyn}
\affiliation{IST Austria, Am Campus 1, 3400 Klosterneuburg, Austria}

\date{\today}
\begin{abstract}
Variational quantum algorithms are promising algorithms for achieving quantum advantage on near-term devices. The quantum hardware is used to implement a variational wave function and measure observables, whereas the classical computer is used to store and update the variational parameters. 
The optimization landscape of expressive variational ans\"atze is however dominated by large regions in parameter space, known as barren plateaus, with vanishing gradients which prevents efficient optimization. In this work we propose a general algorithm to avoid barren plateaus in the initialization and throughout the optimization. To this end we define a notion of \emph{weak barren plateaus} (WBP) based on the entropies of local reduced density matrices. The presence of WBPs can be efficiently quantified using recently introduced shadow tomography of the quantum state with a classical computer. We demonstrate that avoidance of WBPs suffices to ensure sizable gradients in the initialization. In addition, we demonstrate that decreasing the gradient step size, guided by the entropies allows to avoid WBPs during the optimization process. This paves the way for efficient barren plateau free optimization on near-term devices.
\end{abstract}

\maketitle

\section{Introduction}

In recent years the field of quantum computation has seen rapid growth fueled by the arrival of the first generation of quantum computers, dubbed noisy intermediate-scale quantum devices (NISQ)~\cite{preskill2018quantum}. The NISQ era is characterized by quantum computers with a small number of qubits and limited control. The number of coherent operations that can be performed is small and the implementation of famous algorithms with proven quantum speedups, such as Shor's algorithm~\cite{shor}, remains out of reach. To make use of the current generation of quantum computers,  the so-called variational hybrid approach~\cite{nisq} was proposed. The idea is to use the quantum computer in a feedback loop with a classical computer, where it implements a variational wave function that is measured to compute the value of the so-called cost function. This information is then fed into a classical computer where it is processed and the variational wave function is subsequently updated aiming to find a minimum of the cost function, which provides an (approximate) solution to the computationally hard problem. The variational hybrid approach has seen a wide range of proof-of-concept applications on NISQ devices ranging from quantum chemistry~\cite{kandala2017hardware, arute2020hartree} to quantum optimization~\cite{harrigan2021quantum, lacriox2020improving} and quantum machine learning~\cite{havlicek2019supervised, johri2021nearest}.

Despite the large number of suggested applications, the variational approach encountered also  a number of obstacles, that have to be overcome for the future success of the method. In particular, the infamous emergence of \textit{barren plateaus} (BPs) implies that expressive variational ans\"atze tend to be exponentially hard to optimize~\cite{mcclean2018barren}. The main obstacle on the way to optimization lies in the fact that gradients of the cost function are on average zero and deviations vanish exponentially in system size, thus precluding any potential quantum advantage. Moreover, it has been shown that the classical optimization problem is generally NP-hard and plagued with many local minima~\cite{bittel2021training}. 

The problem of BPs attracted significant attention, and numerous approaches were proposed in the literature. In particular, the early research focused on avoidance of BP at the \emph{initialization stage} of variational algorithms~\cite{grant2019initialization, skolik2020layerwise, dborin2021matrix, holmes2021connecting, larocca2021diagnosing}. In a different direction, the relation between occurrence of BPs and the structure of the cost function was studied~\cite{cerezo2020cost, uvarov2020barren}. Also notions of so-called entanglement-induced~\cite{marrero2020entanglement} and noise-induced~\cite{wang2020noise} BPs were introduced. The relation between BPs and entanglement has lead to various proposals that suggest controlling entanglement to mitigate BPs~\cite{kim2021entanglement, kim2021quantum, patti2021entanglement, wiersema2021measurement-induced}. However, measuring entanglement is hard, therefore making these approaches impractical on a real quantum device. 

In this work we introduce the notion of \textit{weak barren plateaus} (WBPs), in order to diagnose and avoid BPs in variational quantum optimization. WBPs emerge when the entanglement of a local subsystem exceeds a certain threshold identified by the entanglement of a fully scrambled state.
In contrast to BPs, WBPs can be efficiently diagnosed using the few-body density matrices and we show that their absence is a sufficient condition for avoiding BPs. Based on the notion of WBPs, we propose an algorithm  that can be readily implemented on available NISQ devices. The algorithm employs \emph{classical shadow} estimation~\cite{huang2020predicting} during the optimization process in order to efficiently estimate the expectation value of the cost function, its  gradients, and the second R\'enyi entropy of small subsystems. 
The tracking of the second R\'enyi entropy enabled by the classical shadows protocol allows for an efficient diagnosis of the WBP both at the initialization step  and during the optimization process of variational parameters. If the algorithm encounters a WBP, as witnessed by a certain subregion having a sufficiently large R\'enyi entropy, the algorithm restarts the optimization process with a decreased value of the update step (controlled by the so-called learning rate). We support the proposed procedure by rigorous results and numerical simulations. The structure of the paper is as follows. 

In Sec.~\ref{sec:2} we introduce the theoretical framework and present our main results. In Sec.~\ref{sec:2A} we introduce the framework of variational quantum eigensolvers (VQEs). Sec.~\ref{sec:2B} introduces the phenomenon of BPs, which dramatically hinders the performance of VQEs. In Sec.~\ref{sec:wbp_algo} we demonstrate WBPs to be a precursor to BPs. We explain why and how WBPs can be efficiently diagnosed in experiments and contrast this with much harder task of detecting BPs. Finally we propose a modification to the VQE algorithm, which allows prevention of the occurance of BPs by avoiding WBPs.

In Sec.~\ref{sec:3} we present a bound for the expectation value of the second R\'enyi entropy in quantum circuits drawn from a unitary ensembles forming a 2-design. This bound allows us to use the second R\'enyi entropy, which is much easier to estimate, instead of the entanglement entropy. In Sec.~\ref{sec:3a} we provide a formal definition of WBPs according to the value of the second R\'enyi entropy of the subsystem and prove that the occurrence of a BP implies the occurrence of a WBP. From this argument it follows that the absence of a WBP precludes the occurrence of a BP. In addition, we provide an upper bound (whose proof is found in Appendix~\ref{appx:shadows}) for the measurement budget require in order to estimate a WBP using classical shadows. Finally, in Sec.~\ref{sec:3b} we demonstrate numerically how the avoidance of WBPs precludes the presence of a BP using the popular BP-free small-angle initialization~\cite{holmes2021connecting, haug2021capacity}.

In Sec.~\ref{sec:4}, we explore how BPs and WBPs emerge at different stages in the VQE optimization and perform a systematic performance analysis.
Next, in Sec.~\ref{sec:4a} we explore the relation of the learning rate and entropy growth for a single update of the VQE algorithm. We analytically and numerically illustrate how a large learning rate 
leads to an uncontrolled growth in subsystem entropies, essentially driving optimization to a WBP region.
In Sec.~\ref{sec:4b} we explore the performance of the WBP-free VQE algorithm in different settings for the Heisenberg model on a chain. Finally, in Sec.~\ref{sec:4c}, we show that our approach allows for the efficient convergence to both, area- and volume-law entangled ground states and compare it to layerwise optimization~\cite{skolik2020layerwise}, which is a popular heuristic for BP avoidance. This is illustrated using the Heisenberg model on a random 3-regular graph, additional results for the Sachdev-Ye-Kitaev (SYK) model can be found in the Appendix~\ref{appx:extra_numerics} which exhibits a nearly maximally entangled ground state.

Finally, in Sec.~\ref{sec:5} we summarize our results, discuss their implications, and outline open questions.

\section{Avoiding barren plateaus in variational quantum optimization}\label{sec:2}
In this section we first introduce  the framework of VQEs, i.e., the unitary ensemble, the cost functions, and the optimization algorithm, and discuss the BP problem.  After this, we present our main result -- a specific modification of the VQE that avoids the issue of BPs.
\subsection{Variational quantum eigensolver \label{sec:2A}}

The aim of the VQE, initially introduced by~\citeauthor{peruzzo2014VQE}, is to approximate the ground state $\ket{\text{GS}}$ of a Hamiltonian $H$ with a variational wave function $\ket{\psi(\bm{\theta})}$. A quantum computer is used to prepare the variational function via the action of a set of unitary gates, $\ket{\psi(\bm{\theta})} = U(\bm{\theta}) \ket{\psi_0}$, where $\ket{\psi_0}$ is the initial state that is typically assumed to be a product state. The variational parameters are then iteratively updated to minimize the expectation value of the Hamiltonian, also-called cost function $E(\bm{\theta})=\bra{\psi(\bm{\theta})} H \ket{\psi(\bm{\theta})}$.

We consider a unitary circuit $U(\bm{\theta})$ of the form of the so-called ``hardware-efficient" ansatz~\cite{kandala2017hardware}
\begin{equation}\label{eq:ruc}
    U(\bm{\theta}) = \prod_{l=1}^{p}  W_l\bigg( \prod_{i=1}^{N} R_l^i(\theta_l^i) \bigg ),
\end{equation}
where $\theta_l^i \in [-\pi, \pi)$ are $p N$ variational angles, concisely denoted as $\bm \theta$. In this expression the product goes over spatial dimension $i=1,\ldots,N$ that labels individual qubits and ``time dimension'', $l = 1,\ldots, p$ with $p$ specifying a number of layers, see Fig.~\ref{fig:1}~(a). We choose the single-qubit gates to be rotations $R_l^i(\theta_l^i)=\exp(-\frac{\im}{2} \theta_l^i G_{l,i})$ with random directions given by $G_{l,i} \in \{\sigma^x, \sigma^y, \sigma^y \}$. $W_l$ is an entangling layer that consists of two-qubit entangling gates represented by nearest-neighbor controlled-Z (CZ) gates with periodic boundary conditions, see Fig.~\ref{fig:1}~(a) for an illustration. 

\begin{figure*}[t]
    \centering
    \includegraphics[width=1.99\columnwidth]{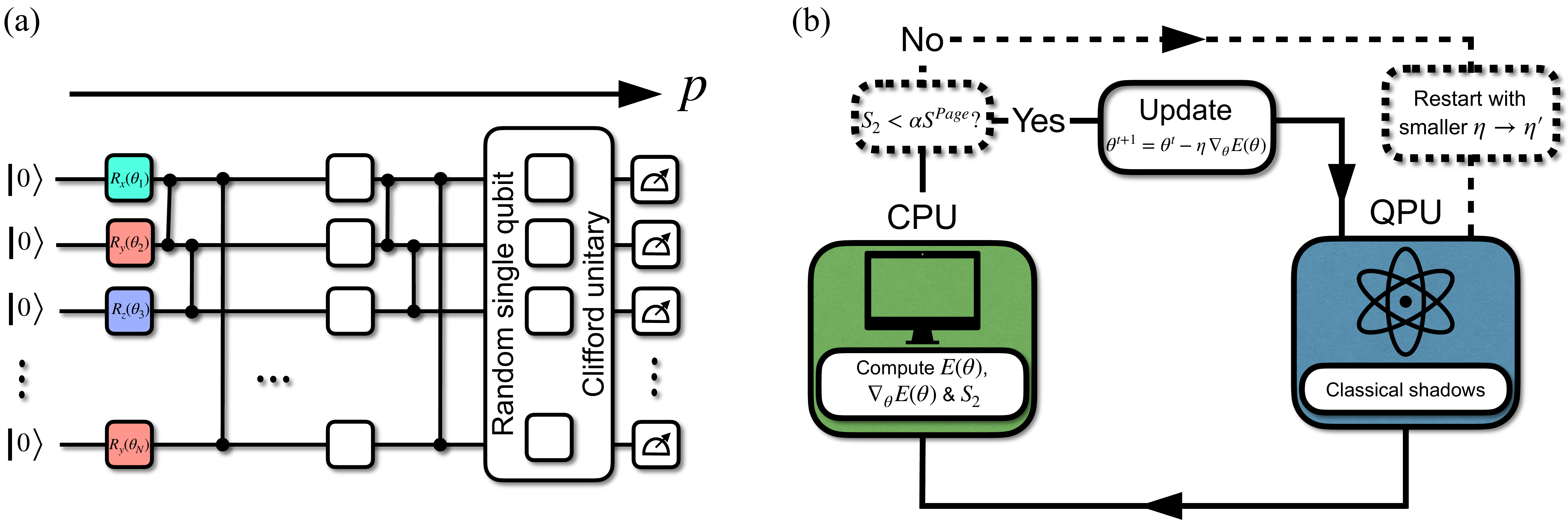}
    \caption{(a) Illustration of the variational quantum circuit $U(\bm{\theta}) \ket{0}$ that is considered in the main text followed by the shadow tomography scheme~\cite{huang2020predicting}. The variational circuit consists of alternating layers of single-qubit rotations represented as boxes and entangling CZ gates shown by lines. The measurements at the end are used to estimate values of the cost function, its gradients, and other quantities. (b) The original hybrid variational quantum algorithm shown by solid boxes can be modified without incurring significant overhead as is shown by the dashed lines and boxes. The modified algorithm tracks entanglement of small subregions and restarts the algorithm if it exceeds the fraction of the Page value that is set by parameter $\alpha$. The full algorithm is efficient; rigorous sample complexity bounds are provided in Appendix~\ref{appx:shadows}.}
    \label{fig:1}
\end{figure*}

We focus our study on $k$-local Hamiltonians $H$, defined as sum of terms each containing at most $k$ Pauli matrices. We take $k$ to be finite and fixed, while the number of qubits $N\gg k$. A particular example of a $2$-local Hamiltonian from the many-body physics is provided by the Heisenberg ($XXX$) model subject to a magnetic field
\begin{equation}\label{eq:heisenberg}
	H_{XXX} = \sum_{i, j \in V_\mathcal{G}} J \big(\sigma_i^z \sigma_{j}^z + \sigma_i^y \sigma_{j}^y + \sigma_i^x \sigma_{j}^x\big)
	+ h_z \sum_{i=1}^N \sigma_i^z,
\end{equation}

where $V_\mathcal{G}$ refers to the vertex set of the graph $\mathcal{G}$ and, couplings are fixed $J=h_z=1$. In our simulations we consider two different graphs: a ring corresponding to a one-dimensional (1D) chain with periodic boundary condition, and a random 3-regular graph. The $U(1)$ symmetry related to the conservation of the $z$ component of the spin under the action of $H$, as well as translational invariance present for chains with periodic boundary condition, can be explored to decrease the space of parameters by using a suitable gate set respecting this symmetry. However, for the sake of generality we focus on the hardware-efficient unitary ansatz defined in  Eq.~\eqref{eq:ruc}. 

Obtaining the energy expectation value $E(\bm{\theta}) =\bra{\psi(\bm{\theta})} H \ket{\psi(\bm{\theta})}$ requires measuring a subset or all qubits in the circuit as we schematically show in Fig.~\ref{fig:1}~(a). For our example of a $2$-local Hamiltonian on the 1D chain, the required measurements include the value of the $\sigma^z$ operator on all sites along with the $\sigma_i^a \sigma_{i+1}^a$ expectation values of all $i=1,\ldots N$ (periodic boundary condition is assumed, so that bits $1$ and $N+1$ are identified) and $a=x,y,z$. Finding the optimal parameters $\bm{\theta}^\star$ requires minimization of the Hamiltonian expectation value $E(\bm{\theta}^\star)=\min_{\bm{\theta}}E(\bm{\theta})$ performed by a classical computer. 

There is a plethora of sophisticated classical optimization algorithms that were applied to this minimization problem~\cite{ostaszewski2021structure, stokes2020quantum, adam, gacon2021simultaneous}. We use the plain gradient-descent (GD) algorithm due to its simplicity, which makes analytical considerations easier. A GD update step is given by
\begin{equation}\label{eq:gd}
	\bm{\theta}^{t+1} = \bm{\theta}^t - \eta \nabla_{\bm{\theta}} E(\bm{\theta}),
\end{equation}
where $\eta$ is the \textit{learning rate}, which controls the update magnitude. This update step is repeated until convergence of $E(\bm{\theta})$, which results from finding a (local) minimum of $E(\bm{\theta})$.

The resulting VQE algorithm is shown schematically in Fig.~\ref{fig:1}~(b) by solid lines. Following the initialization of the variational angles $\bm \theta$, that may be chosen to be real random numbers, the quantum computer is used to prepare the variational state and provide quantum measurement results. The classical computer uses the measurements to estimate the value of the cost-function and its gradient, and performs an update of the variational parameters controlled by the learning rate~$\eta$. 
 
\subsection{Barren plateaus and entanglement\label{sec:2B}}

Whilst the VQE described above is a promising framework for near-term quantum computing due to its modest hardware requirements, its performance may be ruined by the issue of barren plateaus~\cite{mcclean2018barren, cerezo2020cost, holmes2021connecting}. Specifically, the BPs are defined as regions in the space of variational parameters where  the variance of the cost function gradient (and consequently its typical value) vanishes exponentially in the number of qubits~\cite{mcclean2018barren}: 
\begin{equation}\label{eq:Var}
	\mathop{\rm{Var}}[\partial_{i, l} E(\bm{\theta}) ] \sim \mathcal{O}\left(\frac{1}{2^{2N}}\right). 
\end{equation}
 
\citeauthor{mcclean2018barren} were among the first to theoretically investigate BPs. They showed that the appearance of a BP can be related to the circuit matching the Haar random distribution up to the second moment. More precisely, they showed that BPs are a consequence of the unitary ensemble $\mathcal{E}\sim \{U(\bm{\theta})\} _{\bm{\theta}}$ forming a so-called 2-design~\cite{mcclean2018barren} (see Appendix~\ref{appx:t-design} for details and the definition of a $t$-design). To understand the different circuit depth at which BPs are encountered, the authors in Ref.~\cite{cerezo2020cost} introduced the concept of cost-function-dependent BPs. In particular, they argued that the emergence of BP occurs at different circuit depths, depending on the nature of the cost function. 

In contrast, for a so-called global cost function, exemplified by the fidelity, Ref.~\cite{cerezo2020cost} found that BPs already occur at very modest circuit depths $p\sim \mathcal{O}(1)$. The emergence of BP for the fidelity is naturally related to ``orthogonality catastrophe" in many-body physics: even a small global unitary rotation applied to the many-body wave function results in it becoming nearly orthogonal to itself. In terms of fidelity, this implies that it vanishes exponentially in the number of qubits. Moreover, most global state features -- such as expectation values of general operators, fidelities with general states and global purities -- cannot be efficiently accessed on NISQ devices, and are therefore not practical from an algorithmic point of view~\cite{flammia11direct,huang2020predicting,huang2021learning,chen2021exponential}. Therefore, in what follows we do not consider the global cost functions and corresponding BPs.

Local cost functions, that are the focus of the present work are characterized by a later onset of BPs. Specifically, for a $k$-local cost function where $k$ is fixed, the BPs will occur for circuit depth $p \sim \mathcal{O}({\rm poly}(N))$ that increases polynomially in system size~\cite{mcclean2018barren, cerezo2020cost}. In other words, for a large enough $p$ the VQE algorithm will also suffer from a BP already at the very first step of the GD optimization, provided random choice of variational angles $\bm \theta$. We also note that gradient-free optimization strategies do not circumvent the BP problem since the optimization landscape is inherently flat~\cite{arrasmith2021effect}.

The potential emergence of BPs at the initialization stage of the VQE and other algorithms spurred the investigation of different initializations strategies that avoid BPs. Until now, several BP-free initializations were considered in the literature. Ref.~\cite{grant2019initialization} suggests to initialize the circuit with blocks of identities, Ref.~\cite{skolik2020layerwise} suggests to optimize the ansatz layer by layer, and Ref.~\cite{dborin2021matrix} suggests to use a matrix product state ansatz that is optimized by a separate algorithm~\cite{cirac2020matrix} and map that to a quantum circuit. In this work we will focus on small single-qubit rotation as suggested in Ref.~\cite{holmes2021connecting}. 

More recently, it was observed that the entanglement entropy defined as a trace of the reduced density matrix, $S=-\tr \rho_A \ln \rho_A$ (where $\rho_A=\tr_B \rho$ is the reduced density matrix where $A$ is the subset of qubits that are measured and $B$ is the rest of the system) is another source for the occurrence of BPs~\cite{marrero2020entanglement}. The community has subsequently dubbed this kind of BP, \textit{entanglement-induced} BP~\cite{marrero2020entanglement, kim2021entanglement, wiersema2021measurement-induced, patti2021entanglement}. In this work, we will however show that entanglement-induced BPs and BPs for local cost functions, are in fact one and the same. Avoiding entanglement-induced BPs is equivalent to avoiding BPs for local cost functions, the details are presented in Sec.~\ref{sec:3}. 

Experimentally probing a BP is a hard task. 
The estimation of the  exponentially small gradient in Eq.~(\ref{eq:Var}) requires a number of measurements that is exponential in the number of qubits, and therefore invalidates any potential quantum speedup. Moreover, small values of gradient encountered in BP have to be distinguished from the case when gradient vanishes due to convergence to a local minimum. Experimentally diagnosing BPs via entanglement is also impractical. For example, quantum circuits that implement  2-design and thus lead to BPs for local cost functions are characterized by typical volume-law entanglement that approaches nearly maximal values. Checking volume-law entanglement scaling on any device is a formidable challenge.

In the process of variational quantum optimization, the majority of approaches to mitigate BPs apply to the initialization stage~\cite{grant2019initialization, verdon2019learning, volkoff2021largegradients} and not during the optimization. In Sec.~\ref{sec:4}, we illustrate the importance of BP mitigation during the optimization. This motivates the need to devise a BP mitigation strategy for the initialization and during the optimization procedure that is efficient. This procedure is discussed in the algorithm proposed below. 

\subsection{Weak barren plateaus and improved algorithm}\label{sec:wbp_algo}

In order to devise an efficient algorithm for BP-free initialization and optimization of the VQE we introduce the notion of WBPs. Specifically, for a Hamiltonian that is $k$-local, we define the WBP as the point where the second R\'enyi entropy $S_2=-\ln \tr \rho_A^2 $ of any subregion of $k$-qubits satisfies $S_2 \geq \alpha S^{\text{Page}}(k, N)$, where the Page entropy in the limit $k\ll N$ corresponds to the (nearly) maximal possible entanglement of subregion $A$,
\be\label{eq:full_page_asymptotic}
S^{\text{Page}}(k, N) \simeq k\ln 2 - \frac{1}{2^{N-2k+1}},
\ee
where we explicitly used that the Hilbert space dimension of region $A$ is $2^k$ and its complement $B$ has Hilbert space dimension $2^{N-k}$. 
The naive choice for the parameter $\alpha$ is $\alpha=1$. Given some \textit{a priori} knowledge of the entanglement structure of the target state $\ket{\text{GS}}$, the choice can however be more informed to help avoid large entanglement local minima, see Sec.~\ref{sec:3}. 

The notion of WBP is practical since it is defined by $k$-body density matrices, being much easier to access on a real NISQ device. The fact that the prevention of a WBP is sufficient for avoiding the BP may be understood by the intuition from quantum many-body dynamics and the process of thermalization or scrambling of quantum information. In the thermalization process the small subsystems are first to become strongly entangled, as is witnessed by the proximity of their density matrix to the infinite temperature density matrix. This intuition suggests that it is enough to keep in check the density matrices of small subsets of qubits. If their entanglement or other properties are far away from thermal, the system overall is still far away from the BP.

Practically, the WBP can be diagnosed using the shadow tomography scheme~\cite{huang2020predicting}. This scheme enables an efficient way of representing a classical snapshot of a quantum wave function on a classical computer. In essence, the shadow tomography replaces the measurements performed in the computational basis with a more general measurements, that turns out to be sufficient for reconstructing linear and non-linear function of the state, such as expectation values of few-body observables and second R\'enyi entropy of few-body reduced density matrices respectively. 

Relying on the shadow tomography, we propose the following modification of the VQE shown by dashed lines in Figure~\ref{fig:1}~(b). In essence, we suggest to use the tomography to \emph{simultaneously} measure the cost function value and the $k$-body second R\'enyi entropy. For the derivative we require an additional $2pN$ tomographies (two for each parameter) to compute the gradient exactly using the parameter shift rule~\cite{mitarai2018quantum, schuld2019evaluating}, a detailed derivation of the computational cost for each operation is presented in Appendix~\ref{appx:shadows}. Access to the second R\'enyi entropy allows prevention of the appearance of WBPs not only at the initialization step, but throughout the optimization cycle. 
The explicit algorithm works as follows.

\begin{algorithm}[H]
\caption{WBP-free optimization with classical shadows}
\label{algo:wbp}
\begin{algorithmic}[1]
\State Choose $\alpha$, default is $\alpha=1$ \Comment see Sec.~\ref{sec:3a} for details
\State Choose $\bm{\theta}$ such that $S_2< \alpha S^{\text{Page}}(k, N)$
\State Choose learning rate $\eta$
\Repeat \Comment see Appendix~\ref{appx:shadows} for details 
\State Obtain classical shadows $\hat{\rho}^{(t)} (\bm{\theta})$ 
\State Use them to compute $E(\bm{\theta})$, $\nabla_{\bm{\theta}}E(\bm{\theta})$ and $S_2 (\bm{\theta})$
\If{$S_2< \alpha S^{\text{Page}}(k, N)$}
\State $\bm{\theta} \leftarrow \bm{\theta}-\eta \nabla_{\bm{\theta}} E(\bm{\theta})$
\Else 
\State Start again with smaller $\eta \leftarrow \eta'$  
\EndIf
\Until{convergence of $E(\bm{\theta})$}
\end{algorithmic}
\end{algorithm}

If a WBP is diagnosed at the initialization, one may have to take a different initial value of the variational angles or change the initialization ensemble. These aspects are discussed in detail in Sec.~\ref{sec:3}. In addition, the WBP can occur in the optimization loop. This can be mitigated by keeping track of the second R\'enyi entropies in the optimization process. If the WBP condition is fulfilled, one must restart the algorithm with a smaller learning rate. In Sec.~\ref{sec:4} we discuss the optimization process in greater details. In particular, we show how the learning rate is related to the potential change in entanglement entropy, which implies that a smaller learning rate is generally better at avoiding WBPs. 

\section{Weak barren plateaus and initialization of VQE}\label{sec:3}

\subsection{Definition and relation to barren plateaus }\label{sec:3a}

As mentioned in the above, BPs for local cost functions are a consequence of the unitary ensemble $\mathcal{E}\sim \{U(\bm{\theta})\}_{\bm{\theta}}$ forming a $2$-design~\cite{mcclean2018barren, cerezo2020cost}, which leads to an exponentially vanishing gradient variance, i.e.,~a BP. What is important to note is that  the exponential decay is simply a witness of the emergence of a $2$-design. Another, equivalent witness is the second R\'enyi entropy, where we have the following.
\begin{theorem}{(2-design and R\'enyi entropy)} \label{thm:2-design}
 	If the unitary ensemble $\mathcal{E}\sim \{ U(\bm{\theta}) \}$ forms a 2-design, then 
 	for typical instances the second R\'enyi of the state $\rho_A$ concentrates around the Page value $$S^{\rm{Page}}(k, N) - \frac{1}{2^{N-2k+1}} \leq	\mathbb{E}_{\mathcal{E}}\big[ S_2(\rho_A) \big] \leq S^{\rm Page}(k, N),$$ for all subregions $A$ of size $k\ll N$. 
\end{theorem}
\noindent These results are known in the literature, and in the context of random quantum circuits, can be found in Refs.~\cite{popescu2006entanglement, oliveira2007, dahlsten2007typicalentanglement}. However, for completeness we also provide a proof in Appendix~\ref{app:ent}.

The theorem above implies that a large amount of entanglement naturally follows from the similarity between the considered circuit and a random unitary (2-design). Such similarity also gives rise to the vanishing variance of local cost function gradients that define BPs. Therefore, so-called entanglement-induced BPs~\cite{marrero2020entanglement} and BPs for local cost functions are the same. 
In fact, entanglement provides an intuitive picture for the emergence of BPs and its circuit depth dependence. Every entangling layer in the circuit typically increases entanglement of the resulting wave function, until it saturates to its maximal value for any subregion of $k$-qubits at a circuit depth $p\sim \mathcal{O} (\text{poly} (N))$. If the second R\'enyi entropy for half of the subsystem $k=N/2$ has saturated, it has saturated for all smaller subsystem sizes and is thus a sufficient check for a BP. Computing the second R\'enyi is however typically exponentially hard in subsystem size on NISQ devices (for single-copy access this was recently proven in Ref.~\cite{chen2021exponential,huang2021learning}). It is therefore only practical to check a small subregion where $k$ is small and independent of system size.  

The above considerations naturally lead us to introduce the notion of WBPs as a modification of the BP that is computationally efficient to diagnose on NISQ devices. More formally we have as follows.
\begin{definition}{(Weak barren plateaus)}\label{def:1}
 	Let $H$ be an $N$-qubit Hamiltonian, and $A$ is a region containing $k$ qubits. We define a weak barren plateau by the second R\'enyi entropy of the reduced density matrix $\rho_A$ satisfying $S_2 \geq \alpha S^{\text{\rm Page}}(k,N)$ with $\alpha \in [0, 1)$.  
\end{definition}
This definition works for any $k$, however it is reasonable to use $k$ that corresponds to the number of spins involved in interaction terms in the Hamiltonian $H$ since it provides a natural length scale. Moreover, in such a case the reduced density matrix of subregion with $k$ spins contains all necessary information needed to extract the expectation values of Hamiltonian terms localized inside this region. 

While a WBP is a necessary condition for a BP, it is however not sufficient (which motivates the term \textit{weak}). From a practical perspective we are actually interested only in avoiding a BP. For this, WBPs provide a powerful tool, since the following holds.
\begin{corollary}\label{corollary}
    If we find a particular subregion $A$ such that $\rho_A$ does not satisfy the weak barren plateau condition, i.e.\ Definition 2, it is on average also not in a barren plateau where the variance is exponentially small.
\end{corollary}

\begin{proof}
 This assertion immediately follows from negating Theorem~\ref{thm:2-design}.
\end{proof}

The corollary above formalizes the intuition behind the dynamics of entanglement in a circuit: if the state restricted to the smaller subsystem has not scrambled, then neither has the state restricted to a larger subregion. In practice, using classical shadows we can efficiently check one subregion of size $k$ with a total measurement budget
\begin{equation}\label{eq:purity_total_cost}
    T \geq \frac{4^{k+1}  \tr \rho_A^2}{\epsilon^2 \delta},
\end{equation}
where $\epsilon$ is a desired accuracy and $\delta$ is a failure probability (over the randomized measurement process). Parameters $\epsilon$ and $\delta$ do not depend on the number of qubits, whereas the factor $\tr \rho_A^2$ is upper bounded by one for weakly entangled states and can be as small as $2^{-k}$ when entanglement is large.
Moreover, checking all size $k$ subregions incurs an additional overhead of only $k \ln N$. A derivation of this result is presented in Appendix~\ref{appx:shadows}, see Eq.~\eqref{eq:sample-complexity-purity}. Provided that $k$ is small and does not scale with system size, $N$, this can be efficiently implemented on NISQ devices.

If any of these subregions avoids the WBP condition, we are guaranteed to also avoid an actual BP. For simplicity, in the numerical results below we check for the WBP condition for a particular region containing the first $k$ qubits, i.e., \ $A=\{ 1, \cdots, k \}$.

This argument is also intuitive to see by considering a causal cone (blue region) that indicates the extent of the so-called scrambled region (i.e., extend of a subregion with entropy close to the maximal value) in the circuit, see Fig.~\ref{fig:2}~(a). Such a scrambled region grows with every consecutive entangling layer $W_l$ (see Eq.~(\ref{eq:ruc})). When this region extends beyond $k$ qubits, the WBP is reached (left orange dashed line). Later, when the ``scrambling lightcone'' has extended to the full system, the BP is reached (right orange dashed line). Once the BP is reached all smaller regions are also fully entangled and will satisfy the WBP condition on average.

\begin{figure}[tb]
    \centering
    \includegraphics[width=0.98\columnwidth]{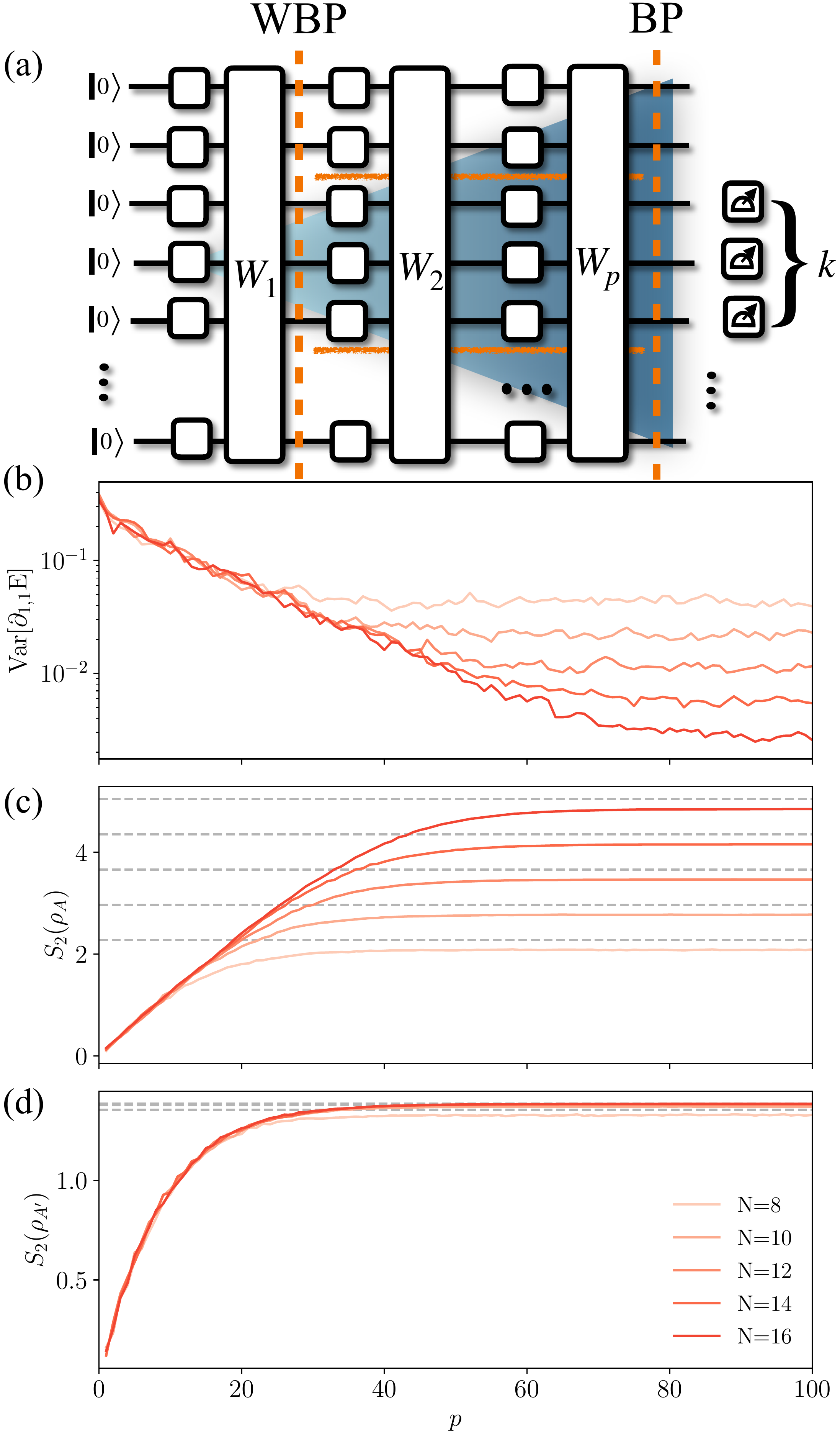}
    \caption{(a) Sketch of the circuit, where the blue color shows the scrambling lightcone. The lightcone first extends over $k$ qubits, where the WBP occurs, and for larger circuit depths extends to the full system size where the BP occurs. (b) The saturation of the gradient variance $\text{Var}[\partial_{1,1} E]$ and (c) saturation of the bipartite second R\'enyi entropy $S_2(\rho_A)$ of the region $A$ consisting of qubits $1,\ldots, N/2$ nearly to the Page value happen at the similar circuit depths $p$, that increases with the system-size $N$. (d) In contrast, the saturation of the second R\'enyi for two qubits ($A'=\{1,2\}$) is system size independent, illustrating that WBP precedes the onset of a BP. Data is averaged over $100$ random initializations. Gradient variance is computed for the local term $\sigma_1^z \sigma^z_2$, typically used in BP illustrations. Gradient variance for the full Heisenberg Hamiltonian, Eq.~(\ref{eq:heisenberg}), looks similar.}
    \label{fig:2}
\end{figure}

Fig.~\ref{fig:2} provides a numerical illustration for the Corollary~\ref{corollary} stated above. We use the hardware-efficient circuit, presented in Eq.~(\ref{eq:ruc}), and compute the gradient variance and second R\'enyi entropy as a function of circuit depth $p$ for different system sizes $N$. We fix $\ket{\psi_0} = \ket{0}$ as the initial state, which is simply all qubits in the zero state. Panel (b) shows the exponential decay of the gradient variance that is usually used to diagnose a BP. Panel (c) shows the corresponding bipartite second R\'enyi entropy. We see that it indeed approaches the Page value (gray dashed line). The Page value is not fully reached since we are considering the second R\'enyi instead of the von Neumann entanglement entropy, this difference however becomes negligible once the subsystem size is decreased. This numerically illustrates that when the $2$-design is reached both the gradient variance and bipartite second R\'enyi entropy have converged. In panel (d) we consider a smaller region of two qubits and see that the second R\'enyi for this region saturates to its maximal value at a significantly lower circuit depth. This illustrates the emergence of the WBP that precedes the onset of the BP after a few more entangling layers. Before the WBP is reached, gradients are well behaved and do not decrease exponentially with the system size. 

Finally, we address the effects of the control parameter $\alpha$, that enters in Definition~\ref{def:1} of the WBP. The naive choice is $\alpha=1$, which means that a WBP is reached if the subregion is maximally entangled with the rest of the system. However, in the case when some a priori knowledge about the entanglement properties of the target state $\ket{\text{GS}}$ is available, it can be used to set a smaller value of $\alpha$. If, for instance, the ground state is only weakly entangled, a choice of $\alpha \ll 1$ may be appropriate. In this way Algorithm~1 in Sec.~\ref{sec:wbp_algo} can also help in avoiding convergence to highly entangled local minima. We discuss this in more detail in Sec.~\ref{sec:4b}.  

\subsection{Illustration of WBP-free initialization}\label{sec:3b}

In order to illustrate the notion of WBP in a more specific setting we apply it to the initialization process of the VQE. Specifically, we focus on the family of initializations that was proposed earlier in order to avoid the issue of BPs~\cite{holmes2021connecting, haug2021capacity}. The one-parametric family of initializations restricts the single-qubit rotation angles from ansatz Eq.~(\ref{eq:ruc}) as $\theta_l^i \in \epsilon_{{\theta}} [-\pi, \pi)$, where  $\epsilon_{{\theta}} \in [0, 1)$ is the control parameter. This strategy allows the onset of the BP to be delayed to arbitrary circuit depths by tuning $\epsilon_{{\theta}}$ accordingly. 

Similarly, it allows the onset of WBPs to be delayed. Depending on the parameter $\epsilon_{{\theta}}$ one can afford a deeper circuit without encountering a WPB in the initialization when compared to the full parameter range ($\epsilon_{{\theta}}=1$). It is straightforward to see that for $\epsilon_{{\theta}}=0$, the ansatz is WBP free for all circuit depths. Indeed, in the absence of the single-qubit rotations, the entangling gates in $W_l$ do not create any entanglement [since the CZ gates used in Eq.~(\ref{eq:ruc}) are diagonal in the computational basis], leaving $\ket{0}$ invariant. Note that, for example, the \textit{identity block} initialization, proposed by \citeauthor{grant2019initialization} works in a similar way in that the unitary is constructed such that it also implements the identity and one is equally left with the zero state.

\begin{figure}[tb]
    \centering
    \includegraphics[width=\columnwidth]{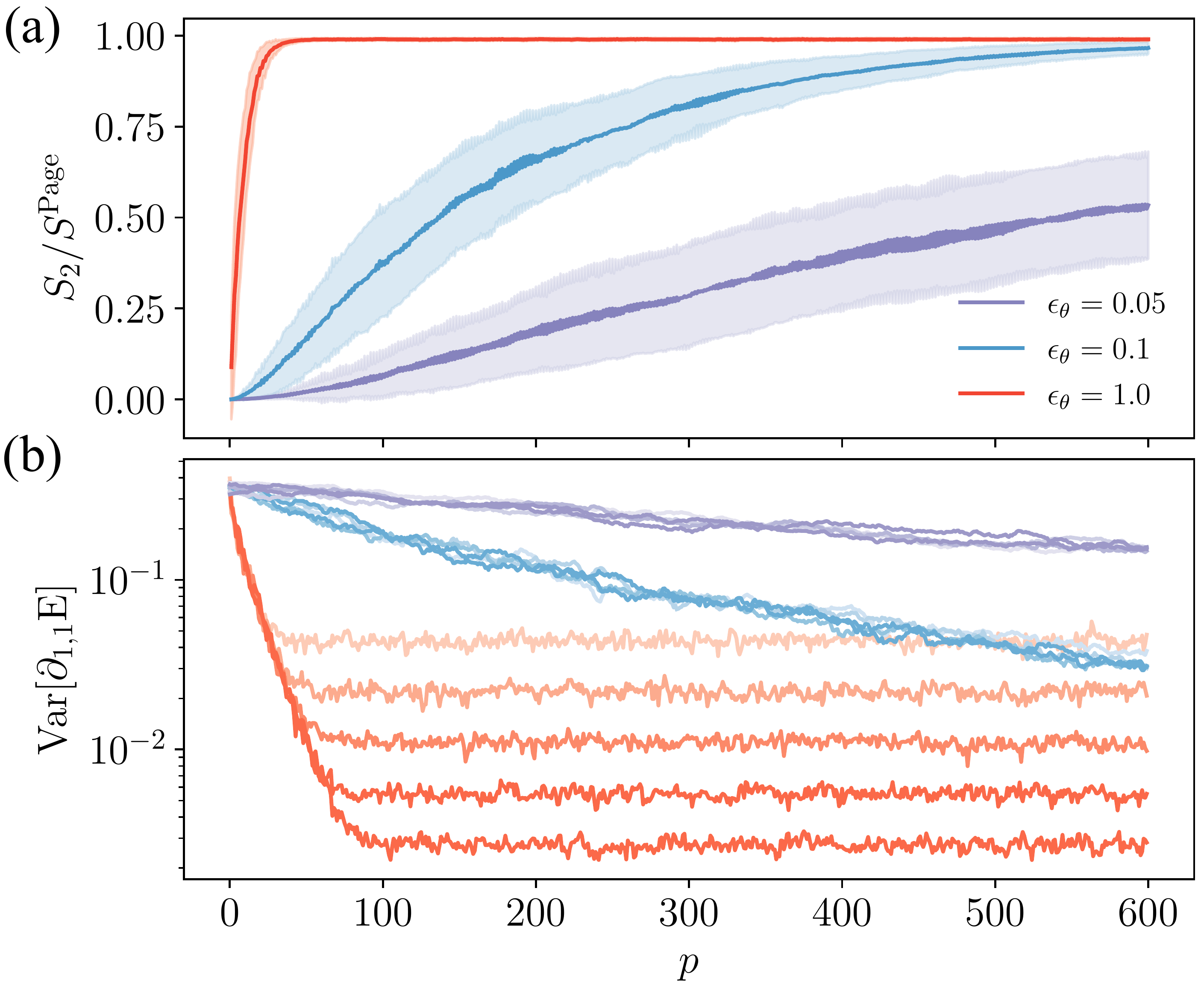}
    \caption{(a) Decreasing parameter $\epsilon_{{\theta}}$ from 1 slows down the growth of the second R\'enyi entropy with the circuit depth $p$. The chosen region contains two qubits.  (b) The encounter of BP in the variance of the gradient of the cost function is visible only for the case $\epsilon_\theta=1$, and it is preceded by the onset of a WBP. We use a system size of $N=16$ for (a) and $N=8, \cdots, 16$ for (b), color intensity corresponds to system size, same as in Fig.~\ref{fig:2}. Data is averaged over $100$ random instances, variance is for the local term $\sigma_1^z \sigma_2^z$.}
    \label{fig:3}
\end{figure}

In Fig.~\ref{fig:3} we numerically illustrate the influence of $\epsilon_{{\theta}}$ on the growth of entanglement and its relation to the gradient variance.
Panel (a) illustrates the growth of the second R\'enyi entropy in the circuit for three different small-angle parameters $\epsilon_{{\theta}}$ and panel (b) shows the corresponding gradient variance. Outside of the WBP the gradient variance vanishes at most polynomially in system size $N$. This illustrates that the avoidance of a WBP is sufficient for avoiding a BP and thus allows for a simple strategy for constructing BP-free initializations. 

\section{Entanglement control during optimization}\label{sec:4}

\subsection{Bounding entanglement increase at a single optimization step}\label{sec:4a}

In Sec.~\ref{sec:2} we presented how the general VQE can be extended with minimal overhead to avoid WBPs in the optimization procedure. The learning rate, as presented in Algorithm 1, hereby plays a crucial role. A smaller learning rate, as observed in Fig.~\ref{fig:1} (c)-(e) is more likely to avoid a WBP. To understand this phenomenological observation on more rigorous grounds, let us consider a sufficiently deep circuit (with a polynomial number of layers in system size), so that the optimization landscape is dominated by WBPs. Careful selection of the parameters allows for an initialization outside of a WBP. However, to remain in the WBP-free region, the optimization has to be performed in a controlled manner, such that the optimizer does not leave the region of low entanglement due to large learning rate and does not end in a WBP. 

Since WBPs are defined in terms of the second R\'enyi entropy $S_2$, we need to bound the change in $S_2$ between iteration steps $t$ and $t+1$. For practical purposes, we instead use the purity ($\tr \rho^2_A = e^{-S_2}$). The change in purity is upper bounded by~\cite{renyis_continuitybound}
\begin{equation}\label{eq:renyis_continuitybound}
  \left|\tr \rho^2_A(t+1) - \tr \rho^2_A(t)\right| \leq 1-(1-T_A(t))^2- \frac{T^2_A(t)}{2^k-1},  
\end{equation}
where $T_A(t) \equiv T(\rho_A(t), \rho_{A}(t+1))$ is the trace distance between the reduced density matrices at iteration steps $t$ and $t+1$, and we assume that region $A$ has $k$ qubits. 

Assuming that the states at consecutive update steps of gradient descent are pertubatively close (see Appendix~\ref{appx:perturb_learning} for details), as measured by the trace distance, one can show that
\begin{equation}\label{eq:trace_learning}
	T(\rho_A(t+1), \rho_A(t)) \lesssim \sqrt{\frac{\eta^2}{4} (\nabla_{\bm{\theta}}E)^T \mathcal{F}(\bm{\theta}) \nabla_{\bm{\theta}}E},
\end{equation}
where $\mathcal{F}_{i,j}(\bm{\theta})=4\mathop{\rm{Re}}[\langle \partial_i \psi| \partial_j \psi \rangle-\langle \partial_i \psi|\psi\rangle \langle \psi| \partial_j \psi \rangle]$ is the quantum Fisher information matrix (QFIM)~\cite{meyer2021fisher} and $\eta$ is the learning rate. Inequalities~\eqref{eq:renyis_continuitybound}-\eqref{eq:trace_learning} imply that the learning rate $\eta$ can be used to limit the maximal possible change of the purity.\footnote{A similar continuity bound which does not require the QFIM can be found in terms of the maximum operator norm of the gate generators. We acknowledge Johannes Jakob Meyer for this remark.} Provided that the change in purity is sufficiently small, the Taylor expansion can be used to argue that the corresponding change in the second R\'enyi entropy $S_2$, related to the purity as $e^{-S_2} = \tr \rho_A^2$, also remains controlled. Therefore, the choice of an appropriately small learning rate can guarantee the avoidance of a WBP at $t+1$, provided the absence of one at $t$.

\begin{figure}[t]
    \centering
    \includegraphics[width=\columnwidth]{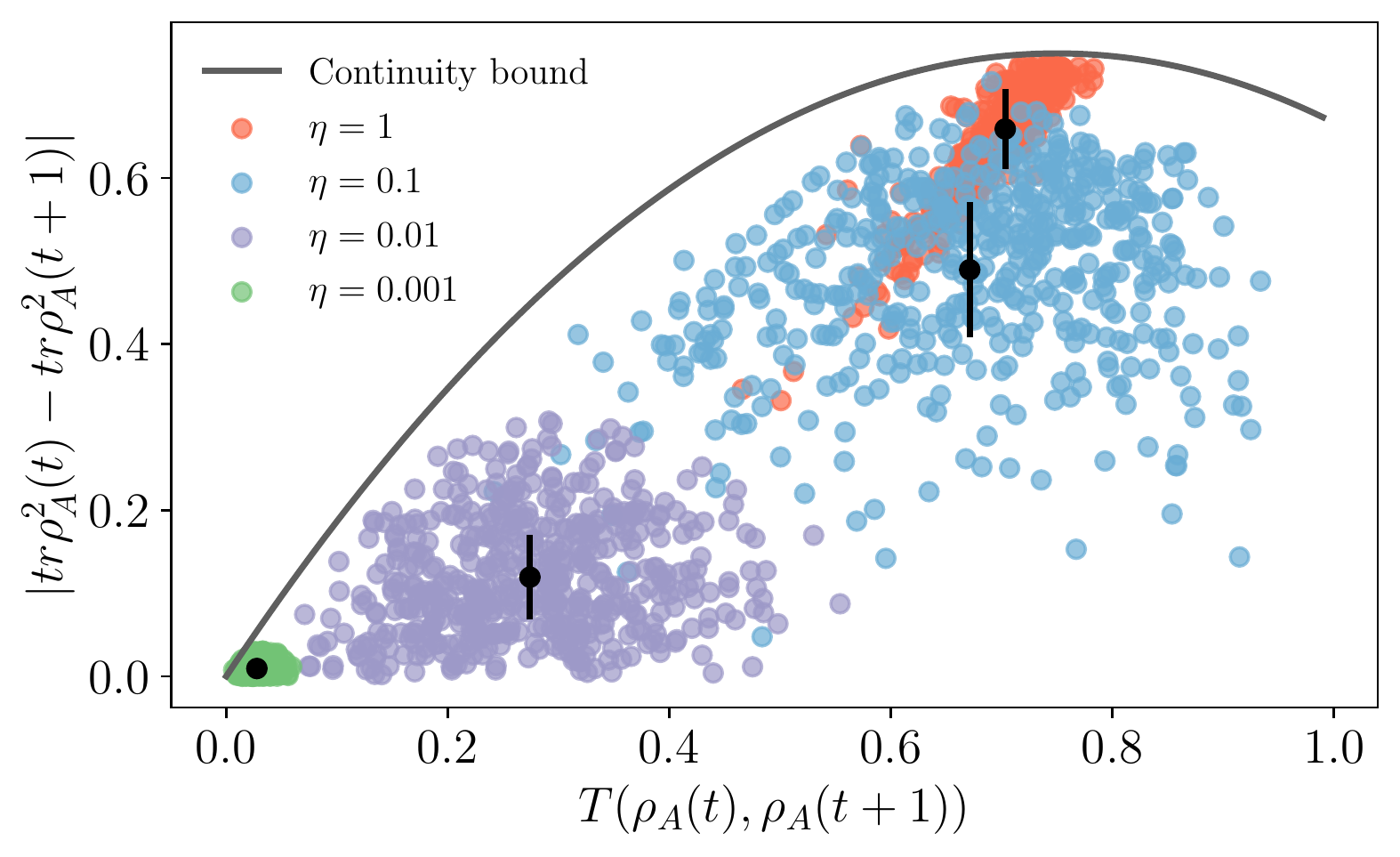}
    \caption{We numerically illustrate the continuity bound Eq.~(\ref{eq:renyis_continuitybound}) and its relation to the learning rate $\eta$ for $t=0$, i.e. at the beginning of the optimization schedule. This shows that one should be careful with the choice of the learning rate since a large learning rate leads to a big change in the trace distance and change in purity. We use a system size of $N=10$ and a random circuit with circuit depth $p=100$ and small qubit rotations ($\epsilon_\theta=0.05$) to generate a BP-free initialization. Data is averaged over $500$ random instances.}
    \label{fig:4}
\end{figure}

To illustrate the bound numerically, we prepare an initialization outside of the  WBP using a small angle parameter $\epsilon_{\bm{\theta}}$ and compute the change in the purity $\tr \rho_A^2$ after one GD update step for different learning rates $\eta$. The results of this procedure for four different learning rates are shown in Fig.~\ref{fig:4}. We see that larger learning rates correspond to a bigger change in purity and are thus more prone to encounter a WBP. At the same time, all data points are below the theoretical bound. While up to the best of our knowledge the bound Eq.~(\ref{eq:renyis_continuitybound}) is not proven to be tight, we observe that points corresponding to the extreme learning rates closely approach the theoretical line.

Using Eq.~\eqref{eq:trace_learning}, the bound can be efficiently approximated on NISQ hardware: the QFIM can be estimated efficiently on a quantum device using techniques suggested in Ref.~\cite{gacon2021simultaneous} or Ref.~\cite{rath2021quantum} using classical shadows. For the computation of the gradient one can use the parameter shift rule~\cite{mitarai2018quantum, schuld2019evaluating} also with shadow tomography. The expression can thus be efficiently evaluated on a real device and used together with the continuity bound to estimate a suitable learning rate $\eta$. However, in practice this might not be needed and simply following Algorithm 1 could be more efficient and easier to implement. 
\subsection{Optimization performance with learning rate}\label{sec:4b}

\begin{figure}[t]
    \centering
    \includegraphics[width=\columnwidth]{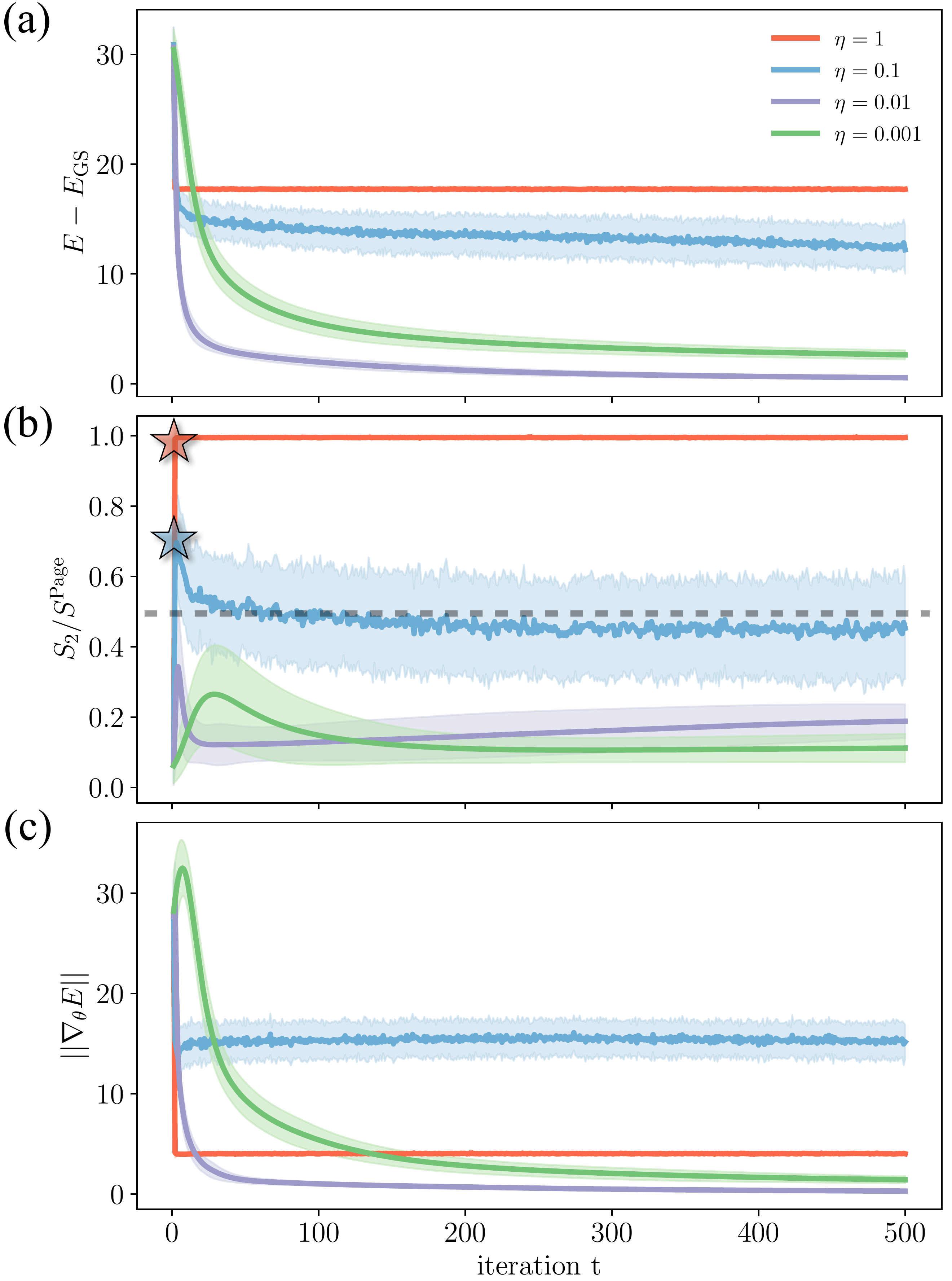}
    \caption{(a-c) The application of the proposed algorithm to the problem of finding the ground state of the Heisenberg model. For large learning rates $\eta=1$ and $0.1$ (red and blue lines) the optimization gets into a large entanglement region as is shown in (b), indicated by colored stars, forcing the restart of the optimization with smaller value of $\eta$. For $\eta=0.01$ the algorithm avoids large entanglement region and gets a good approximation for the ground state. Finally, setting even smaller learning rate (green lines) degrades the performance. The normalized second R\'enyi entropy of the true ground state is $S_2/S^{\text{Page}} (k, N)  \approx 0.246$. (c) Shows the corresponding gradient norm. A small gradient norm equally corresponds to the BP and the good local minima found with $\eta=0.01$ and $0.001$. We use a system size of $N=10$, subsystem size $k=2$, and a random circuit (see Eq.~(\ref{eq:ruc})) with circuit depth $p=100$ and small qubit rotations ($\epsilon_{{\theta}}=0.05$) to generate a BP-free initialization. Here we choose $\alpha=0.5$ indicated by the gray dashed line, see the last paragraph of Sec.~\ref{sec:3a} for a discussion on the choice of $\alpha$. Data is averaged over $100$ random instances.}
    \label{fig:5}
\end{figure}

Finally, we illustrate Algorithm 1 in practice. To this end we first prepare a WBP-free initial state using small qubit rotation angles and compare the performance of GD optimization with different learning rates.  If we start with a large learning rate, $\eta=1$, corresponding to red lines in Fig.~\ref{fig:5}~(a)-(c), we see that the energy expectation value in Fig.~\ref{fig:5}~(a) rapidly (within one or two update steps) converges to a value far away from the target ground state energy $E_{\text{GS}}$. At the same time, panel (b) reveals that this can be attributed to an onset of a WBP, as the second R\'enyi entropy spikes up to the Page value. Finally, panel (c) shows that the gradient norm also is convergent, though at non-zero value. We attribute this to the fact that the system gets trapped in the WBP region. 

As suggested by Algorithm 1, we thus decrease the learning rate to $\eta=0.1$ and start again. This time a WBP is avoided, the algorithm however gets stuck in a local minimum with large entanglement entropy. In this instance a choice of parameter $\alpha$ that defines an onset of a WBP in Def.~\ref{def:1} being smaller than one may be beneficial. For instance, setting $\alpha=0.5$ could help avoiding the suboptimal local minima characterized by large entanglement, see gray dashed line in Fig.~\ref{fig:5}~(b). Note that the large gradient persistent after many iterations for the blue line in Fig.~\ref{fig:5}~(c) may also indicate that the learning rate is chosen too large for the width of the local minima.

Provided that our algorithm uses $\alpha=0.5$, the system would satisfy a WBP condition even for learning rate $\eta=0.1$, forcing us to restart the algorithm with an even smaller learning rate. Setting $\eta=0.01$, we see that the algorithm is now able to converge very close to the true ground state energy (violet line in Fig.~\ref{fig:5}~(a)-(c)). In particular, the norm of the gradient assumes the smallest value among all learning rates. We note, that the further decrease of the learning rate (i.e.,\ to $\eta=0.001$) degrades the performance of GD. While WBPs are not encountered during the optimization process, the GD optimization converges slower within the considered number of iterations and to a larger energy expectation value. This highlights the fact that it is best to choose the highest possible learning rate, that still avoids a WBP. We speculate, that an optimization strategy that adapts the learning rate at each optimization step would give the best performance, though testing this assumption is beyond the scope of the present work. 

\subsection{Classical simulatability and performance comparison}\label{sec:4c}

Now that we have illustrated the procedure outlined in Algorithm 1 in detail, let us comment on the restrictions that our algorithm imposes, its relation to classical simulatability and finally compare our method with other common means for mitigating BPs.

To avoid WBPs and thus BPs we require that the second R\'enyi entropy of a small subregion is less than a fraction $\alpha$ of the Page value, where $\alpha\in(0, 1]$ and the default choice is $\alpha=1$. While this restriction does place a limitation on the entanglement generated by the circuit for a region of k qubits, it does not imply classical simulatabilty of the circuit. Indeed, it is the scaling of the entanglement entropy with system size that  is important for classical simulatability of a quantum system. Only in the special case when the entanglement entropy of the quantum state scales poly-logarithmically with the number of qubits, we can simulate the states on a classical computer in polynomial time~\cite{vidal2003efficientclassical, nest2007classicalsimulation, brandao2013arealaw}. In contrast, the criteria for WBP, Def.~\ref{def:1} is generally consistent with volume-law entanglement as we illustrate below, thus allowing our algorithm to be applied to systems that cannot be efficiently simulated on a classical computer. 

Here we focus on two types of systems: namely systems where the ground state satisfies area law, which implies that the entanglement entropy of an arbitary bipartition of the state scales with the size of the boundary $S(\rho_A)\sim |\partial A|$, as well as volume law, which implies that it scales with the volume, $S(\rho_A)\sim |A|$ (see Ref.~\cite{eisert2010colloquium} for a review on these concepts). 
For area-law states in 1D the entanglement entropy is constant and therefore allows for an efficient classical representation using techniques such as matrix product states~\cite{DMRG_in_the_age_of_MPS}. The 1D Heisenberg model, considered in the previous subsection, is an example for such a system.

The Heisenberg model, however, can be made hard to simulate classically by considering a random-graph geometry illustrated in Fig.~\ref{fig:6}~(a), instead of a 1D chain. This leads to nonlocal interactions and a volume-law entanglement scaling for a typical bipartite cut. Due to the non-local nature of the model we choose $\alpha=1$ since we have no prior knowledge on the entanglement properties of the ground state.  We again use the small-angle initialization~\cite{holmes2021connecting, haug2021capacity} to generate a BP-free initial state. We compare this with layerwise optimization~\cite{skolik2020layerwise}, which is another common heuristic for avoiding BPs. There the circuit is initialized with a single layer, which is optimized, the circuit is then grown by one layer at a time and optimized while keeping the parameters in the previous layers constant.

Fig.~\ref{fig:6}~(b)-(c) reveal that for the Heisenberg model on a graph layerwise optimization ends up in a WBP during the optimization for both learning rates that we considered. The small-angle initialization successfully avoids the WBP for both learning rates, however good convergence is only achieved with $\eta=0.01$. This is similar to the situation encountered in the Heisenberg model in 1D, see Fig.~\ref{fig:7}, where a too large learning rate prevents convergence to the basin of attraction of the local minimum. Likewise to the case of 1D Heisenberg model, the fact that learning rate $\eta=0.1$ does not lead to convergence to a minimum can be revealed through the norm of the gradient which stays large even after 500 iterations.

In addition to the Heisenberg model on the random graph, we also considered the SYK model~\cite{kitaevTalks} that features a volume-law entangled ground state~\cite{yichen2019sykentanglement}. In Appendix~\ref{appx:extra_numerics} we illustrate that our method is also successful in preventing the BP occurrence and results in finding the SYK ground state.

\begin{figure}[t]
    \centering
    \includegraphics[width=\columnwidth]{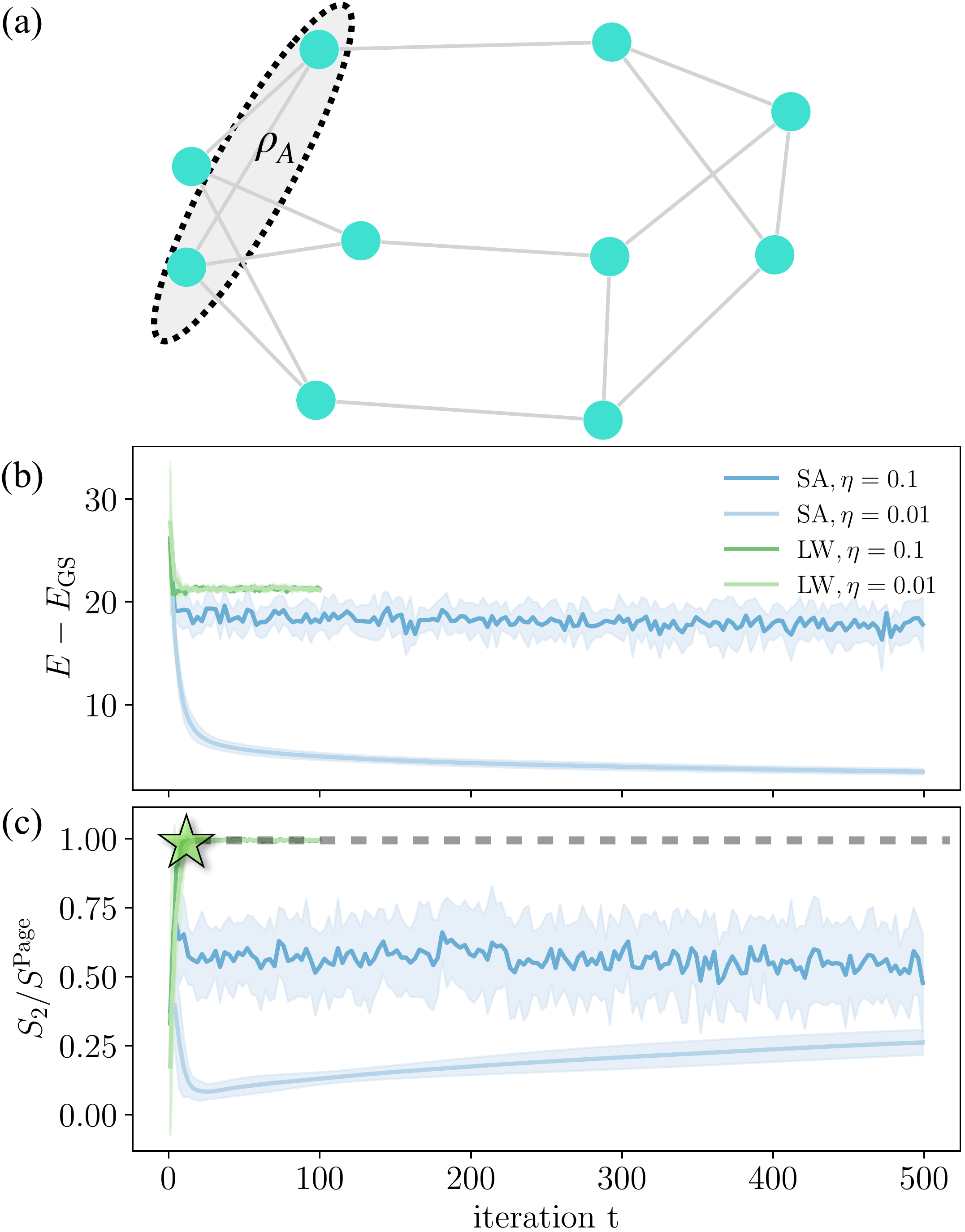}
    \caption{Application of our algorithm to the problem of finding the ground state for the Heisenberg model on a 3-regular random graph depicted in (a). Panel (b) shows the energy as a function of GD iterations $t$ and panel (c) illustrates the second R\'enyi entropy of two-spin region $A$ with $k=2$ shown in panel (a). Since the interactions are now nonlocal and we do not have any prior knowledge on the entanglement properties of the target state we set $\alpha=1$ (gray dashed line). For the initialization we use the small-angle initialization (SA) with $\epsilon_\theta=0.1$ and compare it to layerwise optimization (LW). LW encounters a WBP for both learning rates that we consider (green star). In contrast, SA avoids the WBP for both learning rates. Good performance and further convergence in the local minimum is only achieved through a smaller learning rate of $\eta=0.01$. We use a system size of $N=10$ and a random circuit from  Eq.~\eqref{eq:ruc} with circuit depth $p=100$. Data is averaged over 100 random instances.} \label{fig:6}
\end{figure}

\section{Summary and Discussion}\label{sec:5}

The main result of this work is the introduction of the concept of WBPs, which in essence provides an efficiently detectable version of BPs. In particular, we propose to use the classical shadows protocol to estimate the second R\'enyi entropy of small subregions that are independent of system size. If these subregions avoid nearly maximal entanglement -- a condition sufficient for avoiding WBPs -- the system also avoids conventional BPs. Building on this definition of the WBP, we proposed an algorithm that is capable of avoiding BPs on NISQ devices without requiring a computational overhead that scales exponentially in system size. 

In order to illustrate the notion of WBPs and the proposed algorithm, we studied a particular BP-free initialization of the variational quantum eigensolver. Furthermore, we considered an optimization procedure that uses gradient descent. Phenomenologically, we observed that the encounter of a BP during the optimization crucially depends on the learning rate, which controls the parameter update magnitude between consecutive optimization steps. A smaller learning rate is less likely to lead to the encounter of a BP during the optimization. However, choosing the learning rate to be very small degrades the performance of GD. These results support the feasibility of the proposed algorithm for efficiently avoiding BPs on NISQ devices. While our results and numerical simulations are focused on VQEs, they readily extend to other variational hybrid algorithms, such as quantum machine learning~\cite{benedetti2019parameterized, havlicek2019supervised, schuld2020circuit-centric}, quantum optimization~\cite{farhi2014quantum, sack2021quantumannealing, harrigan2021quantum}, or variational time evolution~\cite{barison2021efficientquantum, lin2021real}. 

Although the issue of avoiding BPs at the circuit initialization is a subject of active research~\cite{grant2019initialization, dborin2021matrix, skolik2020layerwise, holmes2021connecting, larocca2021diagnosing}, the influence and role of BPs in the optimization process has received much less attention~\cite{larocca2021theory}. Our results indicate that entanglement, in addition to playing a crucial role for circumventing BPs at the launch of the VQE, is also important for achieving a good optimization performance. In addition, our heuristic results in Sec.~\ref{sec:4} suggest that postselection based on the entanglement of small subregions may help to avoid low-quality local minima that are characterized by higher entanglement. Algorithm~1 allows for such postselection by appropriately tuning the value of $\alpha$. Doing so, however, requires some prior knowledge about the entanglement structure of the target state. This may be inferred from the structure of the Hamiltonian (for instance, for a Hamiltonian that is diagonal in the computational basis, the eigenstates are product states with no entanglement), or by targeting small instances of the computational problem using exact diagonalization.

Beyond that, one could imagine an algorithm where the learning rate is not only adapted when a WBP is encountered, but dynamically adjusted at every step of the optimization process. This may allow for efficiently maneuvering complicated optimization landscapes by staying clear of highly entangled local minima. VQE, for instance, is known to have many local minima~\cite{bittel2021training}, but a systematic study of their entanglement structure, required for devising such dynamic entanglement post selection procedure, has yet to be done.

Another important question concerns the effect of noise, which has been suggested to be an additional source for the emergence of BPs~\cite{wang2020noise}. Noise cannot be avoided on NISQ machines and has a profound impact on any near-term quantum algorithm, which is difficult to analyze analytically. Fortunately, none of the tools we propose are especially susceptible to noise corruption.
In fact, both the classical shadow protocol and the estimation of observables and purities are stable with respect to the addition of a small but finite amount of noise, and there have even been some proposals for noise mitigation techniques~\cite{chen2021robust,koh2020classical}. 

Finally, we comment on the possibility of testing Algorithm~1 on a real NISQ device. While the shadows protocol can readily be implemented on near-term devices to diagnose WBPs, whether a variational circuit with enough entangling layers that lead to a BP can be realized on a NISQ device is not entirely clear at this stage. Nevertheless recent results of Ref.~\cite{mi2021information} observed convergence of the out-of-time correlators to zero, indicating that a 2-design might already have been reached. This implies that large entanglement, as present in a BP, could be realizable on available NISQ devices, and opens the door to experimental studies of the effect of entanglement on the optimization performance on current NISQ machines using the proposed shadows protocol.  

\section*{Acknowledgments}
We thank Marco Cerezo, Zoe Holmes, and Nicholas Hunter-Jones for fruitful discussion and valuable feedback. We also acknowledge Adam Smith, Johannes Jakob Meyer, and Victor V. Albert for comments on the paper. The simulations were performed in the Julia programming language~\cite{julia} using the Yao module~\cite{yao}. S.H.S., R.A.M., A.A.M., and M.S. acknowledge support by the European Research Council (ERC) under the European Union's Horizon 2020 research and innovation program (Grant Agreement No.~850899).

\section*{Appendix}
\appendix
\section{Classical shadows and implementation details}\label{appx:shadows}

\emph{Shadow tomography} attempts to directly estimate interesting properties of an unknown state without performing full state tomography as an intermediate step. \citeauthor{aaronson2018shadow} and \citeauthor{aaronson2019gentle} showcased that such a direct estimation protocol can be exponentially more efficient, both in terms of Hilbert space dimension ($2^N$ in our case) and in the number of target properties (we use $L$ to denote this cardinality). These techniques do, however, require copies of the underlying quantum state to be stored in parallel within a quantum memory and highly entangled gates to be performed on all copies simultaneously. This is too demanding for current and near-term quantum devices. 

\citeauthor{huang2020predicting} developed a more near-term friendly variant of this general idea known as prediction with \emph{classical shadows}. Similar ideas have been independently proposed by  \citeauthor{paini2019approximate} and \citeauthor{morris2019selective}, respectively. As explained in detail below, the key idea is to sequentially generate state copies and perform randomly selected single-qubit Pauli measurements. Such measurements can be routinely implemented in current quantum hardware and 
enable the prediction of many (linear and polynomial) properties of the underlying quantum state. Importantly, the measurement budget (number of required measurements) still scales logarithmically in the number of target properties $L$, but it may scale exponentially in the support size $k$ of these properties. This is not a problem for local features, like subsystem purities or terms in a quantum many-body Hamiltonian, but does prevent us from directly estimating global state features like fidelity estimation. 

The general measurement budget that is required to simultaneously estimate $L$ local observables using classical shadows, necessary for the energy expectation value estimation, is provided in Theorem~\ref{thm:linear-shadows}. Typically the estimation of $L$ observables would scale linearly in $L$ (essentially every term is estimated individually). This is traded with a $\ln L$ dependence instead and an exponential dependence on the support $k$ of the operators. The cost for estimating the subsystem purities and thus second R\'enyi entanglement entropies is provided in Eq.~(\ref{eq:sample-complexity-purity}) and is exponential in $k$~(this dependence was recently proven to be unavoidable~\cite{chen2021exponential}). However since for the WBP check outlined in the main text $k$ is small, this is generally an efficient operation. Lastly, the cost for estimating the gradients is given in Eq.~(\ref{eq:gradient-cost}). The efficiency of using classical shadows to estimate the energy expectation value and gradients is system dependent (see Ref.~\cite{huang2020predicting} for the application of classical shadow tomography to the lattice Schwinger model). For the estimation of the purities, the shadow protocol, however, generally provides the most efficient technique currently available~\cite{elben2020mixed-state}. One possibility to circumvent these restrictions is to use a hybrid scheme where the energy and gradients are estimated with either classical shadows or the usual approach dependent on the structure of the Hamiltonian while the second R\'enyi entropies for the WBP check are always estimated using classical shadows.

\subsection{Data acquisition via classical shadows}

We use randomized single-qubit measurements to extract information about a variational $N$-qubit state represented by a density matrix
\begin{equation*}
\rho (\bm{\theta}) = |\psi (\bm{\theta}) \rangle \! \langle \psi (\bm{\theta})| \quad \text{with $\bm{\theta} \in \mathbb{R}^m$.}
\end{equation*}
To this end, we repeat the following procedure a total of $T$ times. For $1 \leq t \leq T$ we carry out the following.
\begin{enumerate}
    \item Prepare quantum state $\rho (\bm{\theta})$ on the NISQ device.
    \item Select $N$ single-qubit Pauli observables independently and uniformly at random.
    \item Perform the associated $N$-qubit Pauli measurement (single shot) to obtain $N$ classical bits ($0$ if we measure ''spin down" and $1$ if we measure ``spin up").
    \item Store $N$ single-qubit ``postmeasurement" states, $|s_i^{(t)}\rangle$, where an $i$ th qubit measurement outcome, $s_i$, can take six possible values denoted as $|0 \rangle$, $|1 \rangle$ if qubit is measured in $z$ basis, $|+ \rangle$ and $|- \rangle$ for $x$ basis, and, finally, $| +\mathrm{i}\rangle$ and $|-\mathrm{i}\rangle$ for $y$ basis. Here, $|\pm \rangle = \left( |0 \rangle \pm |1 \rangle \right)/\sqrt{2}$ denote Pauli-$x$ matrix eigenstates and $|\pm \mathrm{i} \rangle = \left(|0 \rangle \pm \mathrm{i}|1\rangle \right)/\sqrt{2}$ are two Pauli-$y$ eigenstates. In practice, this is achieved by applying random single-qubit Clifford gates that effectively implement a change of basis such that the usual $z$-basis measurement can be used, see Fig.~\ref{fig:1}~(a) for a visualization.
    \item (Implicitly) Construct the $N$-qubit \emph{classical shadow}
    \begin{equation}
    \hat{\rho}^{(t)} (\bm{\theta}) = 
     \bigotimes_{i=1}^N \left( 3 |s_i^{(t)} \rangle \! \langle s_i^{(t)}|- \mathbb{I} \right). \label{eq:shadow}
    \end{equation}
\end{enumerate}
Repeating this procedure a total of $T$ times provides us with $T$ classical shadows $\rho^{(1)} (\bm{\theta}),\ldots,\rho^{(T)} (\bm{\theta})$.
These are random matrices that are statistically independent (because they are constructed from independent quantum measurements).
By construction, each classical shadow reproduces the true underlying state in expectation (over both the choice of Pauli observable and the observed spin direction):
\begin{equation}
\mathbb{E} \left[ \hat{\rho}^{(t)} (\bm{\theta}) \right]= \rho (\bm{\theta}) = |\psi (\bm{\theta}) \rangle \! \langle \psi (\bm{\theta})|, \label{eq:shadow-average}
\end{equation}
see e.g.\ Ref.~\cite[Proposition~S.2]{huang2020predicting}.
We can now approximate this ideal expectation value by empirical averaging over all samples:
\begin{equation*}
\rho (\bm{\theta}) \approx \frac{1}{T}\sum_{t=1}^T \hat{\rho}^{(t)} (\bm{\theta}).
\end{equation*}
This approximation becomes exact in the limit $T \to \infty$ of infinitely many measurement repetitions. But the main results in Refs.~\cite{huang2020predicting,paini2019approximate} highlight that convergence actually happens much more rapidly. 

This is, in particular, true for subsystem density matrices. The tensor product structure of classical shadows, Eq.~\eqref{eq:shadow}, plays nicely with taking partial traces.
Let $A \subseteq \left\{1,\ldots,N\right\}$ be a collection of $|A|=k$ qubits. Then,
\begin{equation}
\hat{\rho}_A^{(t)}(\bm{\theta}) = \mathrm{tr}_{\neg A} \left( \hat{\rho}_A^{(t)}\right)  \label{eq:blueuced-shadow}
\end{equation}
is a $k$ qubit shadow that can be used to approximate the associated subsystem density matrix. More precisely, Eq.~\eqref{eq:shadow-average} asserts
\begin{equation}
\mathbb{E} \left[ \rho_A^{(t)}(\bm{\theta})\right] = \mathrm{tr}_{\neg A}\left( \mathbb{E} \left[ \hat{\rho}^{(t)}(\bm{\theta}) \right]\right) = \mathrm{tr}_{\neg A}(\rho(\bm{\theta})) = \rho_A (\bm{\theta})
\label{eq:blueuced-average}
\end{equation}
which can (and should) form the basis of empirical averaging directly for the subsystem in question. Here is a mathematically rigorous result in this direction. In what follows, the range (or weight) of an observable is the number of qubits on which it acts nontrivially. For example coupling terms in the Heisenberg Hamiltonian~\eqref{eq:heisenberg} have range $k=2$, while the external field terms have range $k=1$.

\begin{theorem} \label{thm:linear-shadows}
Fix a collection of $L$ range-$k$ observables $O_l$, as well as parameters $\epsilon,\delta >0$. Then, with probability (at least) $1-\delta$, classical shadows of size
\begin{equation*}
T \geq \frac{4^{k+1} \ln (2L/\delta)}{\epsilon^2}
\end{equation*}
suffice to jointly estimate \emph{all} $L$ expectation values up to additive accuracy $\epsilon$. In other words, 
\begin{equation*}
\hat{\rho}(\bm{\theta}) = \frac{1}{T}\sum_{t=1}^T \hat{\rho}^{(t)} (\bm{\theta}) \;\text{obeys} \; \left| \mathrm{tr} \left( O_l \hat{\rho}(\bm{\theta})\right) - \mathrm{tr} \left( O_l \rho (\bm{\theta})\right) \right| \leq \epsilon,
\end{equation*}
for all $1 \leq l \leq L$.
\end{theorem}

We emphasize that it is not necessary to form global shadow approximations. If $O_l$ only acts nontrivially on subsystem $A_l \subseteq \left\{1,\ldots,N\right\}$ ($O_l = \tilde{O}_l \otimes \mathbb{I}_{\neg A_l}$), then $\mathrm{tr} \left( O_l \hat{\rho}(\bm{\theta})\right) = \mathrm{tr} \left( \hat{O}_l \hat{\rho}_{A_l} \right)$.
Theorem~\ref{thm:linear-shadows} is slightly stronger than a related result in Ref.~\cite{huang2020predicting} (it does not require median-of-means estimation).
Conceptually similar results have been established in Refs.~\cite{huang2021provably} and \cite{evans2019scalable,huang2021derandomization}. Notably, the authors of Ref.~\cite{acharya2021informationally} pointed out to us that they provided a similar statement as in Theorem~\ref{thm:linear-shadows} in their work. We present a formal proof in Appendix~\ref{Sec:proof} below.

\subsection{Estimating subsystem purities}

Suppose we are interested of estimating a collection of multiple subsystem purities
\begin{equation}
p_A (\bm{\theta}) = \mathrm{tr} \left( \rho_A (\bm{\theta})^2 \right) = \mathrm{tr} \left( \rho_A (\bm{\theta}) \rho_A (\bm{\theta}) \right), 
\end{equation}
where $A \subseteq \left\{1,\ldots,N\right\}$ labels different subsystems of size $|A|=k$ each.
Then, we can use the corresponding subsystem shadows, Eq.~\eqref{eq:blueuced-shadow}, to approximate each $p_A$ by empirical averaging:
\begin{align}
\hat{p}_A (\bm{\theta}) = \frac{1}{T(T-1)} \sum_{t \neq t'} \mathrm{tr} \left( \hat{\rho}_A^t \hat{\rho}_A^{t'} \right). \label{eq:purity-estimate}
\end{align}
It is important that we restrict our averaging operation to distinct pairs of classical shadows ($t \neq t'$).
This guarantees that the expectation values factorize, i.e.\ 
\begin{equation*}
\mathbb{E} \left[ \hat{\rho}_A^t \hat{\rho}_A^{t'} \right]= \mathbb{E} \left[ \hat{\rho}_A^t \right] \mathbb{E} \left[ \hat{\rho}_A^{t'} \right] = \rho_A^2,
\end{equation*}
where the last equality is due to Eq.~\eqref{eq:blueuced-shadow}. 
Formula~\eqref{eq:purity-estimate} is an empirical average over all distinct shadow pairs contained in the data set. It converges to the true average $p_A (\bm{\theta})=\mathbb{E} \left[ \hat{p}_A (\bm{\theta}) \right]$, and the speed of convergence is governed by the variance. As data size $T$ increases, this variance decays as
\begin{equation*}
\mathrm{Var} \left[ \hat{p}_A (\bm{\theta}) \right]
\leq \frac{2}{T}\left( 2 \times 4^{k} p_2 (\bm{\theta}) + \frac{1}{T-1} 2^{4k}\right),
\end{equation*}
see, e.g.,\ Ref.~\cite[SM Eq.~(12)]{neven2021symmetry}.
In the large-$T$ limit, this expression is dominated by the first term in parentheses, $4 \times 2^{k}p_2 (\bm{\theta})/T$, and Chebyshev's inequality allows us to bound the probability of a large approximation error. For $\epsilon >0$,
\begin{equation*}
\mathrm{Pr} \left[ \left| \hat{p}_A (\bm{\theta})-\mathrm{tr} \left( \rho_A (\bm{\theta})^2 \right)\right| \geq \epsilon \right] \lesssim  \frac{4^{k+1} \mathrm{tr} \left( \rho_A^2 \right)}{T \epsilon^2},
\end{equation*}
provided that the total number of measurements $T$ is large enough to suppress the higher-order contribution in the variance bound (this is why we write $\lesssim$). In this regime, a measurement budget that scales as
\begin{equation}
T \geq  \frac{4^{k+1} \mathrm{tr} \left( \rho_A^2 \right)}{\epsilon^2 \delta} \label{eq:sample-complexity-purity}
\end{equation}
suppresses the probability of a sizable approximation error ($\geq \epsilon$) below $\delta$.
It is worthwhile to point out that this bound depends on the subsystem purity under consideration. Smaller purities are cheaper to estimate than large ones. It is also important to note that the accuracy parameter $\epsilon$ has to be small enough in order to accurately capture the purity in the WBP regime, which decays exponentially fast, but only with the subsystem size $k$.  

The $\delta$-dependence in Eq.~\eqref{eq:sample-complexity-purity} can be further improved to $\ln (1/\delta)$ by replacing simple empirical averaging in Eq.~\eqref{eq:purity-estimate} by median-of-means estimation~\cite{huang2020predicting}. Doing so would allow us to estimate all possible $L=\binom{N}{k}\leq N^k$ size-$k$ subsystem purities with only a $k \ln N$-overhead. Median-of-means estimation does, however, worsen the dependence on $\epsilon$ by a constant amount. 
Empirical studies conducted in Ref.~\cite{elben2020mixed} showcase that such a trade-off only becomes viable if one wishes to approximate polynomially many subsystem purities.

\subsection{Estimating gradients}

To perform the GD update step suggested in Algorithm 1 we require the knowledge of gradient $\nabla_{\bm{\theta}}E(\bm{\theta})$, which consists of $p N$ derivatives $\partial_{i, l}E(\bm{\theta})$. The derivative can naively be approximated using finite difference, though for variational single-qubit rotation gates, as used in the main text [see Eq.~(\ref{eq:ruc})], we can use the parameter-shift rule to compute the gradients exactly (up to finite sampling errors)~\cite{mitarai2018quantum, schuld2019evaluating}. The parameter-shift rule is given by
\begin{align*}
\partial_{i,l} E \left(\bm{\theta} \right) = \frac{1}{2} \left( E \left(\bm{\theta}+ (\pi/2)\bm{e}_{i,l} \right) - E \left( \bm{\theta}-(\pi/2) \bm{e}_{i,l} \right) \right),
\end{align*}
where $i$ labels the qubits and $l$ cycles through all circuit layers, and $\bm{e}_{i,l}$ is the unit vector.
In order to approximate a single gradient, we need to estimate the difference of two energy expectation values $E(\bm{\theta}_+)=\langle \psi (\bm{\theta}_+) | H |\psi (\bm{\theta}_+)\rangle$ with $\bm{\theta}_+ = \bm{\theta} + (\pi/2) \bm{e}_{i,l}$ and $E(\bm{\theta}_-)=\langle \psi (\bm{\theta}_-) | H |\psi (\bm{\theta}_-)\rangle$ with $\bm{\theta}_- = \bm{\theta} - (\pi/2) \bm{e}_{i,l}$ (we suppress $i$ and $l$ indices in $\bm \theta_\pm$ for the sake of brevity). Typically, the Hamiltonian itself can be decomposed into a sum of $L$ `simple' terms: $H=\sum_{l=1}^L h_l$, where often $L$ can be proportional to the number of qubits, $N$. This allows expression of the gradient as a linear combination of $2L$  expectation values,
\begin{equation}
\partial_{i,l} E \left(\bm{\theta} \right) =  \frac{1}{2} \sum_{l=1}^L \left( \langle \psi (\bm{\theta}_+) | h_l | \psi (\bm{\theta}_+) \rangle - \langle \psi (\bm{\theta}_-) | h_l | \psi (\bm{\theta}_-) \rangle \right), \label{eq:gradient-reformulation}
\end{equation}
each of which can be estimated by performing a collection of single-qubit Pauli measurements. 
If each term $h_l$ is supported on (at most) $k$-qubits, then Theorem~\ref{thm:linear-shadows} applies. 
Performing
$
T \approx 4^{k} \ln (L/\delta)/\epsilon^2
$
randomized Pauli measurements on state $\rho (\bm{\theta}_+)$ and $\rho (\bm{\theta}_-)$ each allows us to $\epsilon$-approximate all $2L$ simple terms in Eq.~\eqref{eq:gradient-reformulation}.

Unfortunately, approximation errors may accumulate when taking the sum over all $2L$ terms. Suppose that we obtain $\epsilon$-accurate estimators $\hat{E}_l (\bm{\theta}_\pm)$ of contribution of the local Hamiltonian term to the energy $E_l (\bm{\theta}_\pm)=\langle \psi (\bm{\theta}_\pm)|h_l | \psi (\bm{\theta}_{\pm}) \rangle$.
A triangle inequality over all approximation errors then produces only
\begin{align*}
& \left| \partial_{i,l}E(\bm{\theta})- \hat{\partial}_{i,l} E(\bm{\theta}) \right| \\
=&  \frac{1}{2} \left|\sum_{l=1}^L \left( \hat{E}_l (\bm{\theta}_+) - E_l (\bm{\theta}_+) - \hat{E}_l (\bm{\theta}_-) + E_l (\bm{\theta}_-\right) \right| \\
\leq &\frac{1}{2}\sum_{l=1}^L \left| \hat{E}_l (\bm{\theta}_+) - E_l (\bm{\theta}_+)\right| + \frac{1}{2}\sum_{l=1}^L \left| \hat{E}_l (\bm{\theta}_-) - E_l (\bm{\theta}_-) \right| = L \epsilon.
\end{align*}
This upper bound equals only $\epsilon$ if we rescale the accuracy of original approximation  to $\epsilon / L$.
Inserting this rescaled accuracy into Theorem~\ref{thm:linear-shadows} produces an overall measurement cost of
\begin{equation}
T \geq  \frac{4^{k+1} L^2 \ln (2L/\delta)}{\epsilon^2}. \label{eq:gradient-cost}
\end{equation}
The number $L$ of terms in the Hamiltonian typically scales (at least) linearly in the number of qubits $N$. 
This implies that the measurement budget, Eq.~\eqref{eq:gradient-cost}, required to (conservatively) estimate gradients 
scales quadratically in the system size and thus is parametrically larger than the (conservative) cost of estimating purities of size-$k$ subsystems, Eq.~\eqref{eq:sample-complexity-purity}. To obtain the full gradient $\nabla_{\bm{\theta}} E(\bm{\theta})$ the procedure has to be repeated $p N$ times since the parameter-shift rule has to implemented for every variational parameter. It should be noted though, that in principle this can be computed in parallel, provided large enough (quantum) computational resources. For example, different NISQ computers could be used to estimate different gradient components at the same time.

\subsection{Example of error accumulation in an Ising model}

The extra scaling with $L^2$ in Eq.~\eqref{eq:gradient-cost} is a consequence of error accumulation. If we use the same measurement data to jointly estimate many Hamiltonian terms, then all these estimators become highly correlated. And the effect of outlier corruption -- which occurs naturally in empirical estimation -- becomes amplified. 

Here, we illustrate this subtlety by means of a simple example. Let $H=-J\sum_{i=1}^{N-1} \sigma_i^z \sigma_{i+1}^z$  be the Ising Hamiltonian on a 1D chain comprised of $N$ qubits ($L=N-1$). Let us also assume that $N$ is even. This Hamiltonian is diagonal in the $z$ basis $|i_1,\ldots,i_N \rangle = |i_1 \rangle \otimes \cdots \otimes |i_N \rangle$ with $i_1,\ldots,i_N \in \left\{0,1\right\}$. So, in order to estimate $H$, it suffices to perform measurements solely in this basis. Born's rule asserts, that we observe bitstring $\hat{s}_1,\ldots,\hat{s}_N$ with probability
\begin{equation*}
\mathrm{Pr} \left[ \hat{s}_1,\ldots,\hat{s}_N \right] = \langle \hat{s}_1,\ldots,\hat{s}_N | \rho | \hat{s}_1,\ldots,\hat{s}_N \rangle,
\end{equation*}
where $\rho$ denotes the underlying $N$-qubit state. And, we can use these outcomes to directly estimate the total energy. It is easy to check that
\begin{align*}
\hat{E} =& \langle \hat{s}_1,\ldots,\hat{s}_N | H | \hat{s}_1,\ldots,\hat{s}_N \rangle \\
=& - J \sum_{i=1}^N \langle \hat{s}_i |\sigma_i^z| \hat{s}_i \rangle \langle \hat{s}_{i+1}|\sigma_{i+1}^z | \hat{s}_{i+1}\rangle
\end{align*}
obeys $\mathbb{E} \left[ \hat{E} \right] = \mathrm{tr} \left( H \rho \right)$, regardless of the quantum state $\rho$ in question. Also, estimating individual terms in this sum is both cheap and easy. 
Convergence of the sum, however, does depend on the underlying quantum state and the correlations within. 
To illustrate this, we choose $\lambda \in (0,1)$ and set
\begin{equation*}
\rho (\lambda) = (1-\lambda) |\psi \rangle \! \langle \psi| + \lambda |\phi \rangle \! \langle \phi|,
\end{equation*}
where $|\psi \rangle = |00 \cdots 00 \rangle$ is the Ising ground state and $|\phi \rangle = |01\cdots 01\rangle$ is a N\'eel state. These states obey 
$\langle \psi| H |\psi \rangle = - J(N-1)$ (ground state) and $\langle \phi |H| \phi \rangle = +J (N-1)$ (highest excited state), so
\begin{equation*}
\mathrm{tr}(H \rho(\lambda)) 
= -J(n-1) \left(1-2\lambda\right).
\end{equation*}
The task is to approximate this expectation value based on computational basis measurements. 
For each measurement, we either obtain outcome $0 \cdots 0$ (with probability $1-p$) or outcome $01 \cdots 01$ (with probability $p$). This dichotomy extends to our estimator
\begin{equation*}
\hat{E} =
\begin{cases}
\langle \psi | H |\psi \rangle = - J (N-1) &\text{with prob.\ $1-\lambda$}, \\
\langle \phi|H|\phi \rangle = +J (N-1) & \text{with prob.\ $\lambda$}.
\end{cases}
\end{equation*}
and we are effectively faced with estimating the (rescaled) expectation value of a biased coin. The associated variance of such a coin toss can be easily computed and amounts to
\begin{align*}
\mathrm{Var} \left[ \hat{E}\right] =& \mathbb{E} \left[ \hat{E}^2 \right] - \left( \mathbb{E} \left[ \hat{E}\right]\right)^2 
= 4J^2 (N-1)^2 \lambda (1-\lambda).
\end{align*}
Unless $\lambda \neq 0,1$ (where the variance vanishes), this variance it is proportional to $L^2=(N-1)^2$ and controls the rate of convergence. Asymptotically, a total number of 
\begin{equation*}
T \geq \mathrm{Var} \left[ \hat{E}\right] /\epsilon^2= 4J^2 L^2\lambda(1-\lambda)/\epsilon^2=\Omega (L^2/\epsilon^2)
\end{equation*}
independent coin tosses are necessary (and sufficient) to $\epsilon$-approximate the true expectation value $\mathbb{E} \left[ \hat{E}\right] = \mathrm{tr} \left( \rho (\lambda) H \right)$. This is a consequence of the central limit theorem and showcases that a measurement budget scaling with the number $L$ of Hamiltonian terms is unavoidable in general. 

We emphasize that this is a contrived worst-case argument that showcases how correlated measurements can affect the approximation quality of a sum of many simple terms, while each term individually is cheap and easy to evaluate. A generalization to the Heisenberg Hamiltonian considered in the main text, see Eq.~\eqref{eq:heisenberg}, is straightforward.  

\subsection{Proof of Theorem~\ref{thm:linear-shadows} \label{Sec:proof}}

Theorem~\ref{thm:linear-shadows} is a consequence of the following concentration inequality. Let $\|O\|_\infty$ denote the operator and spectral norm of an observable. We also use $\| \cdot \|_1$ to denote the trace norm.

\begin{theorem} \label{thm:linear-shadows-restatement}
Fix a collection of $L$ range-$k$ observables $O_l$ with $\|O_l\|_\infty \leq 1$, a quantum state $\rho$ and let $\hat{\rho} = \frac{1}{T}\sum_{t=1}^T \hat{\rho}^{(t)} $ be a classical shadow estimate thereof. Then, for $\epsilon \in (0,1)$,
\begin{equation*}
\mathrm{Pr}\left[ \max_{1 \leq l \leq L}\left| \mathrm{tr} \left( O_l \hat{\rho})\right) - \mathrm{tr} \left( O_l \rho \right) \right| \geq \epsilon \right] \leq 2L \exp \left( - \frac{\epsilon^2 T}{4^{k+1}} \right).
\end{equation*}
\end{theorem}

This large deviation bound is a consequence of another well-known tail bound, see, e.g.,\ Ref.~\cite[Theorem 7.30]{rauhut2013book}.

\begin{theorem}[Bernstein inequality]
Let $X^{(1)},\ldots,X^{(T)}$ be independent, centered (i.e., $\mathbb{E}\left[X_t \right]=0$) random variables that obey $|X^{(t)}| \leq R$ almost surely.
Then, for $\epsilon >0$
\begin{equation*}
\mathrm{Pr} \left[ \left| \frac{1}{T}\sum_{t=1}^T X^{(t)} \right| \geq \epsilon \right] \leq 2  \exp \left( - \frac{\epsilon^2T^2/2}{\sigma^2+RT\epsilon} \right),
\end{equation*}
where $\sigma^2 = \sum_{t=1}^T \mathbb{E} \left[\left(X^{(t)}\right)^2\right]$.
\end{theorem}

\begin{proof}[Proof of Theorem~\ref{thm:linear-shadows-restatement}]

Fix an observable $O=O_l$ with $1 \leq l \leq L$ and define $X^{(t)}=\mathrm{tr} \left( O \hat{\rho}^{(t)}\right)-\mathrm{tr} \left( O \rho \right)$. Then, by construction of classical shadows, each $X^{(t)}$ is an independent random variable that also obeys $\mathbb{E}\left[X^{(t)}\right]=0$, courtesy of Eq.~\eqref{eq:shadow-average}. Next, let $A \subseteq \left\{1,\ldots,N\right\}$ with $|A|=k$ be the subsystem on which the range-$k$ observable acts nontrivially, i.e., $O = O_A \otimes \mathbb{I}_{\neg A}$ and $\|O \|_\infty = \| O_A\|_\infty \leq 1$.
Then, Hoelder's inequality ($\left|\mathrm{tr} \left( O_A \rho_A \right)\right| \leq \|O_A \|_\infty \|\rho_A \|_1$) asserts
\begin{align*}
\left|X^{(t)}\right| =& \left| \mathrm{tr} \left( O_A \hat{\rho}^{(t)}_A\right) - \mathrm{tr} \left( O_A \rho_A \right) \right| \\
\leq & \left\| O_A \right\|_\infty \left( \left\| \rho_A \right\|_1 + \left\| \hat{\rho}_A^{(t)} \right\|_1 \right) \\
=& \left\| O_A \right\|_\infty \left( 1 + \prod_{a \in A} \left\| 3 |s_a^{(t)}\rangle \! \langle s_a^{(t)}|-\mathbb{I} \right\|_1 \right) \\
\leq & \left(1 + 2^{|A|}\right) = 1+2^k =R,
\end{align*}
where we also use $\|\rho_A \|_1 = \mathrm{tr}(\rho_A)=1$ and the specific form of subsystem classical shadows Eq.~\eqref{eq:blueuced-shadow}, that factorizes nicely into tensor products. Estimating the variance is more difficult by comparison. However, Ref.~\cite[Proposition~S3]{huang2020predicting} asserts
\begin{equation*}
\mathbb{E} \left[ \left(X^{(t)}\right)^2 \right] \leq \|O \|_{\mathrm{shadow}}^2 \leq 4^k \|O \|_\infty = 4^k.
\end{equation*}
In turn, $\sigma^2 \leq T 4^k$ and we conclude
\begin{align*}
& \mathrm{Pr} \left[ \left| \mathrm{tr} \left( O \hat{\rho}\right) - \mathrm{tr} \left( O \rho \right) \right| \geq \epsilon \right] \\
=& \mathrm{Pr} \left[ \left| \frac{1}{T} \sum_{t=1}^T X^{(t)}\right| \geq \epsilon \right]  \\
\leq & 2 \exp \left( - \frac{\epsilon^2 T^2/2}{T4^k + (1+2^k)T \epsilon} \right)  \\
\leq & 2 \exp \left( - \frac{\epsilon^2 T}{4^{k+1}} \right),
\end{align*}
where the last line is a rough simplification of the exponent. Such a tail bound is valid for any $O=O_l$ and the advertised statement follows from taking a union bound (also known as Boole's inequality) over all possible deviations:
\begin{align*}
& \mathrm{Pr}\left[ \max_{1 \leq l \leq L}\left| \mathrm{tr} \left( O_l \hat{\rho})\right) - \mathrm{tr} \left( O_l \rho \right) \right| \geq \epsilon \right] \\
\leq & \sum_{l=1}^L \mathrm{Pr}\left[\left| \mathrm{tr} \left( O_l \hat{\rho})\right) - \mathrm{tr} \left( O_l \rho \right) \right| \geq \epsilon \right] \\
\leq & 2L \exp \left( - \frac{\epsilon^2 T}{4^{k+1}} \right).
\end{align*}
\end{proof}

\section{Unitary \texorpdfstring{$t$}--designs}\label{appx:t-design}

Here, we briefly review the notion of unitary $t$-designs.
The Haar measure is the unique left and right invariant measure on the unitary group $U(d)$, where $d$ here stands for the dimension of the full Hilbert space, $d=2^N$. 
Unitary $t$-designs are ensembles of unitaries that approximate moments of the Haar measure. More precisely, let $\mathcal{E}$ be an ensemble of unitaries, i.e.,\ a subset of $U(d)$ equipped with a probability measure. For an operator $O$ acting on the $t$-fold Hilbert space $\mathcal{H}^{\otimes t}$, the $t$-fold channel with respect to $\mathcal{E}$ is defined as
\begin{equation}
\Phi^t_{\mathcal{E}}(O) = \int_{\mathcal{E}} \mathrm{d}U U^{\otimes t} (O) U^{\dagger \otimes t}.
\end{equation}
Essentially, we are asking when the average of an operator $O$ over the ensemble $\mathcal{E}$ equals an average over the full unitary group. A unitary $t$-design~\cite{dankert09unitary,gross07evenly} is an ensemble $\mathcal{E}$ for which the $t$-fold channels are equal for all operators~$O$,
$$
\Phi^t_{\mathcal{E}} (O) = \Phi^t_{\mathop{\rm{Haar}}} (O).
$$
Being a $t$-design means we exactly capture the first $t$ moments of the Haar measure with larger $t$ better approximating the full unitary group. There are known constructions of $t$-designs for $t=2$ and $t=3$~\cite{dankert2009exact, cleve2016nearlinear, kueng2015qubit, webb2016clifford, zhu17clifford}. For $t=1$, it is known that any basis for the algebra of operators of $\mathcal{H}$, including the Pauli group, is a $1$-design. In practice, one is more interested in when the ensemble of unitaries is close to forming a $t$-design. With this, given a tolerance $\epsilon_t >0$ one refers to the ensemble $\mathcal{E}$ as being an approximate $t$-design if 
$$
\norm{\Phi^t_{\mathcal{E}} - \Phi^t_{\mathop{\rm{Haar}}} }_\diamond \leq \epsilon_t,
$$
where $\norm{\cdot}_\diamond$ is the diamond norm -- a worst-case distance measure that is very popular in quantum information theory, see, e.g.,\ \cite{watrous_2018}. In the quantum-machine-learning literature the distance between the two $t$-fold channels is known as the expressibility of the ensemble $\mathcal{E}$~\cite{holmes2021connecting}, the smaller the distance the more expressive the ensemble is.

\section{Entanglement and unitary \texorpdfstring{$2$}--designs}\label{app:ent}
Random unitary operators have often been used to approximate late-time quantum dynamics. In the crudest approximation, it is assumed that the unitary matrix is directly drawn from the Haar measure. Although modeling quantum dynamics by random unitaries is an approximation, it has led to new insights into black hole physics~\cite{page1993, hayden2007blackholes, sekino2008fastscramblers} and produced computable models of information spreading and entanglement dynamics~\cite{nahum_2017_RUC, nahum_2018_RUC, hosur2016, pollman2018otocs}. 

In what follows, we consider a weaker situation where the random unitary operator is drawn from an ensemble $\mathcal{E}$ forming a $2$-design, and focus on the entanglement properties of $N$-qubits random pure states \be\label{eq:random_state}
|\psi \rangle = U|\psi_0\rangle,
\ee
with $U\sim \mathcal{E}$. These results have been previously obtained, for example, Refs.~\cite{popescu2006entanglement, oliveira2007,dahlsten2007typicalentanglement} and references therein. 

Given a bipartition $(A, \neg{A})$ of the system, we begin by studying the distance of the reduced density matrix $\rho_A$ to the maximally entangled state $\rho^{\infty}_{A} = \mathbb{I}_A/d_A$, where $d_A$ is the dimension of the Hilbert space $\mathcal{H}_A$ associated with region $A$. The full Hilbert space dimension is denoted by $d=2^N$.

\subsection{Bounding the expected trace distance}
Let us recall the following inequality relating the 1-norm (trace distance) $\norm{M}_1 = \tr \sqrt{ M^\dagger M}$, and the 2-norm (Frobenius norm) $\norm{M}_2 = \sqrt{\tr (M^\dagger M)}$
\begin{equation}
\norm{M}_2 \leq  \norm{M}_1  \leq \sqrt{d} \norm{M}_2. \label{eq:norm_1and2} 
\end{equation} 
We are interested in bounding $\mathbb{E}_\mathcal{E}\big(\norm{\rho_A - \mathbb{I}_A/d_A }_1\big)^2$. To do so we first use Jensen's inequality and afterwards employ the inequality~\eqref{eq:norm_1and2},
\be
\begin{split}
\mathbb{E}_\mathcal{E}\big(\norm{\rho_A - \mathbb{I}_A/d_A }_1\big) ^2 & \leq \mathbb{E}_\mathcal{E}\big(\norm{\rho_A - \mathbb{I}_A/d_A}^2_1 \big) \\
& \leq d_A \mathbb{E}_\mathcal{E}(\norm{\rho_A - \mathbb{I}_A/d_A}^2_2).
\end{split}
\ee
The last term on the right-hand side is related to the purity:
\be\label{eq:bound_rhoA}
\mathbb{E}_\mathcal{E}(\norm{\rho_A - \mathbb{I}_A/d_A}^2_2)= \mathbb{E}_\mathcal{E}(\tr \rho^2_A) - 1/d_A .
\ee
As we see, the only nontrivial dependence on $U$ comes from the purity of the reduced density matrix. Let $\{ |I\rangle = |i_A, j_{\neg{A}}\rangle \}_{i,j}$ be the computational basis for the Hilbert space $\mathcal{H} = \mathcal{H}_{A} \otimes \mathcal{H}_{\neg{A}}$ (such that it respects the bipartition).

Let us now proceed with the calculation of the average purity. We first compute the reduced density matrix $\rho_A$ and write it as a sum over products of matrix elements of the unitary operator $U$:
\begin{align*}
    \rho_A &= \sum^{d_{\neg{A}}}_{j_{\neg{A}}} \langle j_{\neg{A}}| \rho | j_{\neg{A}} \rangle = \sum^{d_{\neg{A}}}_{j_{\neg{A}}} \sum^{d}_{J,I} \rho_{I,K}  \langle j_{\neg{A}}|  I \rangle \langle K | j_{\neg{A}} \rangle, \\
    &= \sum_{i_A, k_A} \sum_{j_{\neg{A}}} \rho_{(i_A, j_{\neg{A}}), (k_A, j_{\neg{A}})} | i_A \rangle \langle k_A |, \\
    &= \sum_{i_A, k_A} \sum_{j_{\neg{A}}} U_{(i_A, j_{\neg{A}}),(0,0)}U^*_{(k_A, j_{\neg{A}}),(0,0)} | i_A \rangle \langle k_A |,
\end{align*}
where the last line follows from Eq.~\eqref{eq:random_state}.
 
Afterwards, it can be easily verified that $\tr (\rho^2_A)$ reads
\begin{widetext}
\be
\begin{split}
\tr(\rho^2_A) &= \sum_{i_A, k_A} \sum_{j_{\neg{A}}, p_{\neg{A}}} U_{(i_A, j_{\neg{A}}), (0, 0)} U_{(k_A, p_{\neg{A}}), (0,0)} U^*_{(k_A, j_{\neg{A}}), (0,0)} U^*_{(i_A, p_{\neg{A}}),(0,0)}.
\end{split}
\ee
\end{widetext}
Using the following identities for the first and second moment of the unitary group endowed with the Haar measure
\begin{equation}\label{eq:weingarten}
\begin{split}
&\int_{U(n)} dU_H U_{i,j} U^*_{i_1,j_1} = \delta_{i,i_1} \delta_{j,j_1}/d, \\
& \int_{U(n)} dU_H U_{i,j} U_{l,m}U^*_{i_1,j_1} U^*_{l_1,m_1} = \\ &\frac{1}{d^2-1}(\delta_{i,i_1}\delta_{l,l_1}\delta_{j,j_1}\delta_{m,m_1}+  \delta_{i,l_1}\delta_{l,i_1}\delta_{j,j_1}\delta_{m,m_1}) - \\
& \frac{1}{d(d^2-1)} (\delta_{i,i_1}\delta_{l,l_1}\delta_{j,m_1}\delta_{m,j_1} + \delta_{i,l_1}\delta_{l,i_1}\delta_{j,j_1}\delta_{m,m_1}),
\end{split}
\end{equation}
we get that the following simple expression for the expected purity
\be\label{eq:average_purity}
\mathbb{E}_\mathcal{E}(\tr \rho^2_A) = \frac{d_A+d_{\neg{A}}}{1+d_A d_{\neg{A}}}.
\ee
Finally, substituting Eq.~\eqref{eq:average_purity} into Eq.~\eqref{eq:bound_rhoA} we obtain
\be
\mathbb{E}_\mathcal{E}\big(\norm{\rho_A - \mathbb{I}_A/d_A }_1\big) \leq \sqrt{\frac{d_A^2 -1}{d_Ad_{\neg{A}} + 1}} \sim \mathcal{O}(\sqrt{d_A/d_{\neg{A}}})
\ee
Note that the above result implies that when the complementary subsystem $\neg{A}$ is (significantly) larger than $A$, the expected deviation of $\rho_A$ from the maximally mixed state is exponentially small.

\subsection{Bounding the expected second R\'enyi entropy}
Let us now explore the average value of the second R\'enyi entropy, which, as mentioned in the main text, can be easily estimated using the classical shadows protocol by~\citeauthor{huang2020predicting}. 

Computing the exact average value of the second R\'enyi is a complicated task. Hence, we instead provide a lower and an upper bound for it. On one hand, via Jensen's inequality, we have that
\be
-\ln \mathbb{E}_\mathcal{E}(\tr \rho^2_A) \leq \mathbb{E}_\mathcal{E}(S_2(\rho_A)), 
\ee
which changes the focus of our attention to the expectation value of the purity of the reduced density matrix $\mathbb{E}_\mathcal{E}(\tr \rho^2_A)$. Using the result from the previous subsection Eq.~\eqref{eq:average_purity} and taking the logarithm, we get the following lower bound:
\be
-\ln \mathbb{E}_\mathcal{E}(\tr \rho^2_A)=-\ln  \frac{d_A+d_{\neg{A}}}{1+d_A d_{\neg{A}}}.
\ee
Taking the large $d$ limit and writing everything in terms of $d_A/d_{\neg{A}}$ we find
\be
\label{eq:S2_final}
-\ln \mathbb{E}_\mathcal{E}(\tr \rho^2_A) \approx \ln d_A - \frac{d_A}{d_{\neg{A}}} + \mathcal{O}\left(\frac{d^2_A}{d^2_{\neg{A}}}\right).
\ee
On the other hand, we have that for any state $\rho_A$ the following inequality holds: $$
S_2(\rho_A) \leq S(\rho_A) = -\ln \rho_A \tr \rho_A,
$$
where $S(\rho_A)$ is the von Neumann entropy of $\rho_A$. Taking averages does not change this relation and we conclude $\mathbb{E}_{\mathcal{E}}(S_2(\rho_A)) \leq \mathbb{E}_{\mathcal{E}}(S(\rho_A))$. The expectation value of the von Neumann entropy is upper bounded by the \emph{Page entropy}:
\be
S^{\text{Page}}(d_A, d) = \frac{1}{\ln 2}\Big( -\frac{d_A - 1}{2 }\frac{d_A}{d} + \sum_{j=d/d_A+1}^{d} \frac{1}{j} \Big).
\ee
\citeauthor{page1993} conjectured that this analytical formula accurately captures the von Neumann entropy of a Haar random state. This conjecture was subsequently proven in Ref.~\cite{proofPage_1994}. 
Putting everything together, we obtain
\be
-\ln  \frac{d_A+d_{\neg{A}}}{1+d_A d_{\neg{A}}} \leq \mathbb{E}_\mathcal{E}(S_2(\rho_A)) \leq S^{\text{Page}}(d_A, d).
\ee
Considering now that the number of qubits inside region $A$ is equal to $k$ and assuming that $d_A/d_{\neg{A}} = 1/2^{N-2k} \ll 1$ we arrive
at the expression in Theorem~\ref{thm:2-design}, that is
\be
k \ln 2 - \frac{1}{2^{N-2k}} \leq \mathbb{E}_{\mathcal{E}}(S_2) \leq k \ln 2 - \frac{1}{2}\frac{1}{2^{N-2k}}.
\ee
We see that whenever the unitary ensemble $\mathcal{E}$ forms a $2$-design, the expected value of the second R\'enyi entropy is close to the Page entropy. 

\section{Entanglement growth and learning rate}\label{appx:perturb_learning}

Here we detail the derivation of Eq.~(\ref{eq:trace_learning}). We first upper bound the trace distance via
\begin{align}\label{eq:trace-distance}
	T(\rho_A, \sigma_A) \leq T(\ket{\psi}, \ket{\phi}) &= \sqrt{1-f(\ket{\psi}, \ket{\phi})}, 
\end{align}
where $f$ stands for the pure state fidelity $f(\ket{\psi(\bm{\theta})}, \ket{\psi(\bm{\theta} +\bm{\delta})})=|\bra{\psi(\bm{\theta})} \ket{\psi(\bm{\theta} + \bm{\delta})}|^2$. Taylor expanding the pure state fidelity around $\bm{\theta}$ we get
\begin{equation}\label{eq:fidelity-fisher}
	f(\ket{\psi(\bm{\theta})}, \ket{\psi(\bm{\theta} +\bm{\delta})}) = 1-\frac{1}{4}\bm{\delta}^T \mathcal{F}(\bm{\theta}) \bm{\delta} + \mathcal{O}(\bm{\delta}^4),
\end{equation}
where $\mathcal{F}(\bm{\theta})$ is the QFIM given by
\begin{equation}
\label{eq:qfim_def}
\mathcal{F}_{i j}(\bm{\theta}) = 4 \Re{\bra{\partial_i \psi} \ket{\partial_j \psi} - \bra{\partial_i \psi} \ket{\psi} \bra{\psi} \ket{\partial_j \psi}}.
\end{equation}
Assuming $\bm{\delta} \ll 1$ we can neglect higher-order terms in $\bm{\delta}$ and so
\be	
	T(\rho_A, \sigma_A) \lesssim \sqrt{\frac{1}{4}\bm{\delta}^T \mathcal{F}(\bm{\theta}) \bm{\delta}} =\sqrt{\frac{\eta^2}{4} (\nabla_{\bm{\theta}}E)^T \mathcal{F}(\bm{\theta}) \nabla_{\bm{\theta}}E}, 
\ee	
where in the last equality we plug in the parameter change under GD (Eq.~\eqref{eq:gd}), $\bm{\delta}=-\eta \nabla_{\bm{\theta}} E$.

\section{Algorithm performance for SYK model}\label{appx:extra_numerics}
\begin{figure}[t]
    \centering
    \includegraphics[width=\columnwidth]{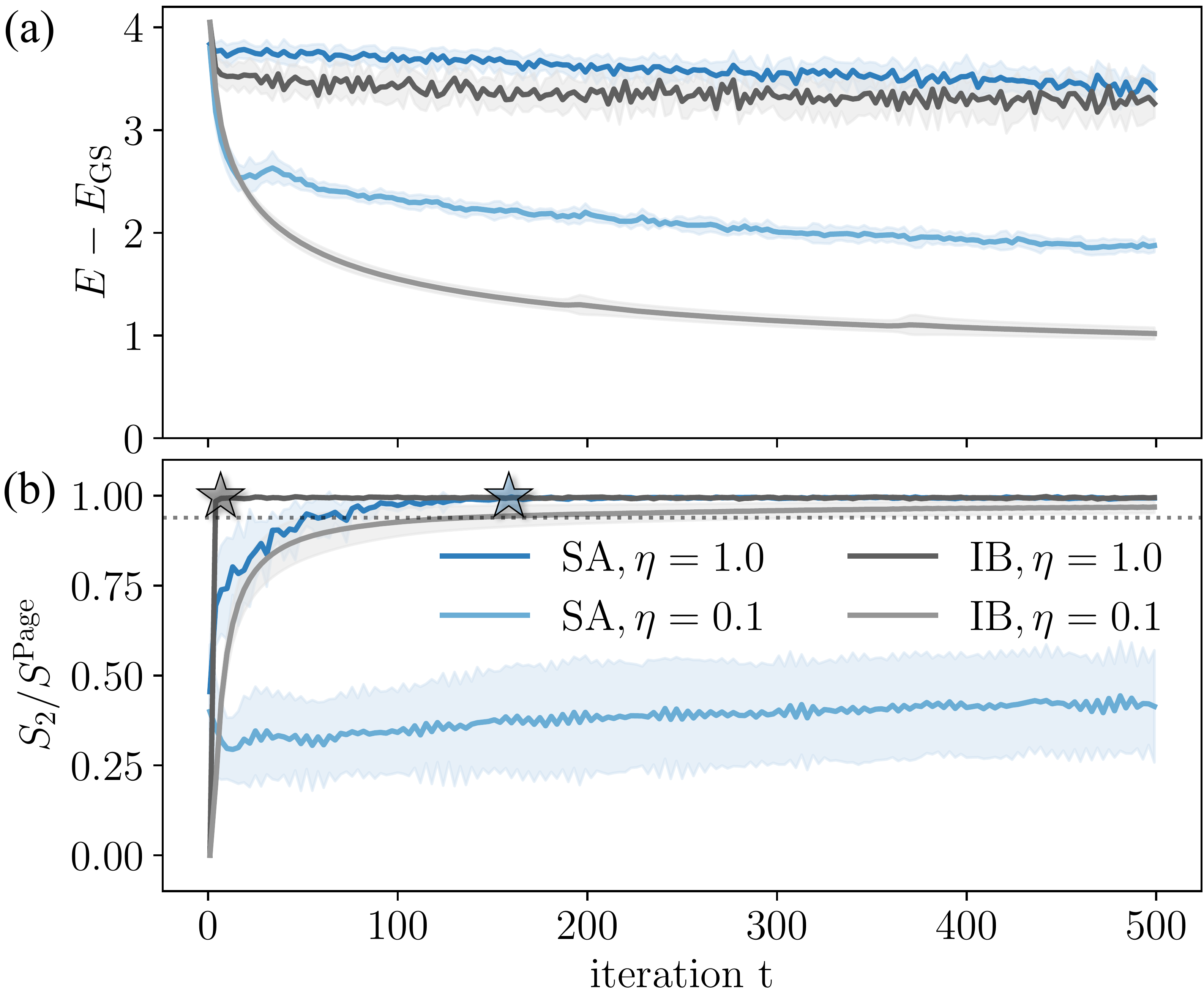}
    \caption{(a-b) The application of our algorithm to the problem of finding the ground state of the SYK model. For the initialization we consider the small-angle (SA) ($\epsilon_\theta=0.1$) and identity block (IB) initialization~\cite{grant2019initialization} (using one block). We can see that only through the reset of the learning rate $\eta$, as suggested by Algorithm 1, WBPs are avoided during the optimization. The entanglement entropy of the target state is nearly maximal (indicated by the dotted line), we omit the WBP line for $\alpha=1$ for improved visibility. We measure energy in units of $J$ and use a system size of $N=10$, subsystem size $k=2$ and a random circuit from Eq.~\eqref{eq:ruc} with circuit depth $p=100$. Data is averaged over 100 random instances.} \label{fig:7}
\end{figure}

In this section we show the numerical results for the VQE applied to the ground state search of the SYK model~\cite{kitaevTalks}. The SYK model provides a canonical example for a volume-law model where the ground state is nearly maximally entangled~\cite{yichen2019sykentanglement}. The nonlocal nature of the Hamiltonian does not allow for an efficient estimation of the energy expectation value of this model using classical shadows. Thus, this model may be viewed as a theoretical example that shows that application of our algorithm is not limited to area-law entangled states. We use a small-angle initialization as well as the identity-block initialization~\cite{grant2019initialization} to illustrate our method.

The SYK model is a quantum-mechanical model of $2N$ spinless Majorana fermions $\chi_i$ satisfying the following anticommutation relations $\{ \chi_i, \chi_j\} = \delta_{ij}$. The SYK model was introduced by Kitaev~\cite{kitaevTalks} as a simplified variant of a model introduced by Sachdev and Ye~\cite{sachdevandye_1993}. The Hamiltonian of the model is 
\begin{equation}    
H_{\text{SYK}} = \sum_{i,j,k,l}^{2N} J_{i,j,l} \chi_i \chi_j \chi_k \chi_l,
\end{equation}
where the couplings $J_{i,j,k,l}$ are taken randomly from a Gaussian distribution with zero mean and variance
$$
\mathop{\rm{var}}[ J_{i,j,k,l}] = \frac{3!}{(N-3)(N-2)(N-1)} J^2.
$$
We can study Majorana fermions using spin-chain variables by a nonlocal change of basis known as the Jordan-Wigner transformation:
\begin{equation}
\chi_{2i} = \frac{1}{\sqrt{2}} \sigma^x_1\cdots \sigma^x_{i-1} \sigma^y_i, \quad \chi_{2i-1} = \frac{1}{\sqrt{2}} \sigma^x_1\cdots \sigma^x_{i-1} \sigma^z_i,    
\end{equation}
such that $\{\chi_i, \chi_j \} = \delta_{i,j}$. With this representation, encoding $2N$ Majorana fermions requires $N$ qubits. For our studies, we set $J=1$ and consider a system of $N=10$ qubits.

We study performance of VQE for SYK model using two different initializations. 
Fig.~\ref{fig:7}~(a)-(b) show that the WBP is avoided during optimization for if the learning rate is chosen appropriately. For a large learning rate ($\eta=1$) both initializations encounter a WBP during the optimization (indicated by the gray and blue star). Once the learning rate is decreased ($\eta=0.1)$ the entanglement entropy slowly grows to the nearly maximal value associated with the ground state of the SYK model (dotted line) instead of uncontrollably reaching the Page value. For this model, it is important to use $\alpha=1$ (the default value) such that the entanglement entropy can grow during the optimization. Only if there is some \textit{a priori} knowledge of the properties of the ground state, $\alpha$ can be chosen to be smaller. 

The identity block initialization~\cite{grant2019initialization} here leads to the best optimization performance. We attribute this to the fact that the identity block initialization allows for a faster growth in entanglement since the parameter values are highly fine tuned. Our results suggest that sensitivity of the initialization ansatz to small perturbations may be beneficial for the cases when the ground state is nearly maximally entangled. These results highlight the advantage of using our algorithm. The tracking of the second R\'enyi entanglement entropy during the optimization reveals that the larger learning rates encounter a WBP while the smaller learning rates successfully avoid it.


\begin{thebibliography}{94}%
\makeatletter
\providecommand \@ifxundefined [1]{%
 \@ifx{#1\undefined}
}%
\providecommand \@ifnum [1]{%
 \ifnum #1\expandafter \@firstoftwo
 \else \expandafter \@secondoftwo
 \fi
}%
\providecommand \@ifx [1]{%
 \ifx #1\expandafter \@firstoftwo
 \else \expandafter \@secondoftwo
 \fi
}%
\providecommand \natexlab [1]{#1}%
\providecommand \enquote  [1]{``#1''}%
\providecommand \bibnamefont  [1]{#1}%
\providecommand \bibfnamefont [1]{#1}%
\providecommand \citenamefont [1]{#1}%
\providecommand \href@noop [0]{\@secondoftwo}%
\providecommand \href [0]{\begingroup \@sanitize@url \@href}%
\providecommand \@href[1]{\@@startlink{#1}\@@href}%
\providecommand \@@href[1]{\endgroup#1\@@endlink}%
\providecommand \@sanitize@url [0]{\catcode `\\12\catcode `\$12\catcode
  `\&12\catcode `\#12\catcode `\^12\catcode `\_12\catcode `\%12\relax}%
\providecommand \@@startlink[1]{}%
\providecommand \@@endlink[0]{}%
\providecommand \url  [0]{\begingroup\@sanitize@url \@url }%
\providecommand \@url [1]{\endgroup\@href {#1}{\urlprefix }}%
\providecommand \urlprefix  [0]{URL }%
\providecommand \Eprint [0]{\href }%
\providecommand \doibase [0]{https://doi.org/}%
\providecommand \selectlanguage [0]{\@gobble}%
\providecommand \bibinfo  [0]{\@secondoftwo}%
\providecommand \bibfield  [0]{\@secondoftwo}%
\providecommand \translation [1]{[#1]}%
\providecommand \BibitemOpen [0]{}%
\providecommand \bibitemStop [0]{}%
\providecommand \bibitemNoStop [0]{.\EOS\space}%
\providecommand \EOS [0]{\spacefactor3000\relax}%
\providecommand \BibitemShut  [1]{\csname bibitem#1\endcsname}%
\let\auto@bib@innerbib\@empty
\bibitem [{\citenamefont {{Preskill}}(2018)}]{preskill2018quantum}%
  \BibitemOpen
  \bibfield  {author} {\bibinfo {author} {\bibfnamefont {J.}~\bibnamefont
  {{Preskill}}},\ }\bibfield  {title} {\bibinfo {title} {{Quantum Computing in
  the NISQ era and beyond}},\ }\href@noop {} {\bibfield  {journal} {\bibinfo
  {journal} {arXiv e-prints}\ ,\ \bibinfo {eid} {arXiv:1801.00862}} (\bibinfo
  {year} {2018})},\ \Eprint {https://arxiv.org/abs/1801.00862}
  {arXiv:1801.00862 [quant-ph]} \BibitemShut {NoStop}%
\bibitem [{\citenamefont {{Shor}}(1995)}]{shor}%
  \BibitemOpen
  \bibfield  {author} {\bibinfo {author} {\bibfnamefont {P.~W.}\ \bibnamefont
  {{Shor}}},\ }\bibfield  {title} {\bibinfo {title} {{Polynomial-Time
  Algorithms for Prime Factorization and Discrete Logarithms on a Quantum
  Computer}},\ }\href@noop {} {\bibfield  {journal} {\bibinfo  {journal} {arXiv
  e-prints}\ ,\ \bibinfo {eid} {quant-ph/9508027}} (\bibinfo {year} {1995})},\
  \Eprint {https://arxiv.org/abs/quant-ph/9508027} {arXiv:quant-ph/9508027
  [quant-ph]} \BibitemShut {NoStop}%
\bibitem [{\citenamefont {{Bharti}}\ \emph {et~al.}(2021)\citenamefont
  {{Bharti}}, \citenamefont {{Cervera-Lierta}}, \citenamefont {{Kyaw}},
  \citenamefont {{Haug}}, \citenamefont {{Alperin-Lea}}, \citenamefont
  {{Anand}}, \citenamefont {{Degroote}}, \citenamefont {{Heimonen}},
  \citenamefont {{Kottmann}}, \citenamefont {{Menke}}, \citenamefont {{Mok}},
  \citenamefont {{Sim}}, \citenamefont {{Kwek}},\ and\ \citenamefont
  {{Aspuru-Guzik}}}]{nisq}%
  \BibitemOpen
  \bibfield  {author} {\bibinfo {author} {\bibfnamefont {K.}~\bibnamefont
  {{Bharti}}}, \bibinfo {author} {\bibfnamefont {A.}~\bibnamefont
  {{Cervera-Lierta}}}, \bibinfo {author} {\bibfnamefont {T.~H.}\ \bibnamefont
  {{Kyaw}}}, \bibinfo {author} {\bibfnamefont {T.}~\bibnamefont {{Haug}}},
  \bibinfo {author} {\bibfnamefont {S.}~\bibnamefont {{Alperin-Lea}}}, \bibinfo
  {author} {\bibfnamefont {A.}~\bibnamefont {{Anand}}}, \bibinfo {author}
  {\bibfnamefont {M.}~\bibnamefont {{Degroote}}}, \bibinfo {author}
  {\bibfnamefont {H.}~\bibnamefont {{Heimonen}}}, \bibinfo {author}
  {\bibfnamefont {J.~S.}\ \bibnamefont {{Kottmann}}}, \bibinfo {author}
  {\bibfnamefont {T.}~\bibnamefont {{Menke}}}, \bibinfo {author} {\bibfnamefont
  {W.-K.}\ \bibnamefont {{Mok}}}, \bibinfo {author} {\bibfnamefont
  {S.}~\bibnamefont {{Sim}}}, \bibinfo {author} {\bibfnamefont {L.-C.}\
  \bibnamefont {{Kwek}}},\ and\ \bibinfo {author} {\bibfnamefont
  {A.}~\bibnamefont {{Aspuru-Guzik}}},\ }\bibfield  {title} {\bibinfo {title}
  {{Noisy intermediate-scale quantum (NISQ) algorithms}},\ }\href@noop {}
  {\bibfield  {journal} {\bibinfo  {journal} {arXiv e-prints}\ ,\ \bibinfo
  {eid} {arXiv:2101.08448}} (\bibinfo {year} {2021})},\ \Eprint
  {https://arxiv.org/abs/2101.08448} {arXiv:2101.08448 [quant-ph]} \BibitemShut
  {NoStop}%
\bibitem [{\citenamefont {{Kandala}}\ \emph {et~al.}(2017)\citenamefont
  {{Kandala}}, \citenamefont {{Mezzacapo}}, \citenamefont {{Temme}},
  \citenamefont {{Takita}}, \citenamefont {{Brink}}, \citenamefont {{Chow}},\
  and\ \citenamefont {{Gambetta}}}]{kandala2017hardware}%
  \BibitemOpen
  \bibfield  {author} {\bibinfo {author} {\bibfnamefont {A.}~\bibnamefont
  {{Kandala}}}, \bibinfo {author} {\bibfnamefont {A.}~\bibnamefont
  {{Mezzacapo}}}, \bibinfo {author} {\bibfnamefont {K.}~\bibnamefont
  {{Temme}}}, \bibinfo {author} {\bibfnamefont {M.}~\bibnamefont {{Takita}}},
  \bibinfo {author} {\bibfnamefont {M.}~\bibnamefont {{Brink}}}, \bibinfo
  {author} {\bibfnamefont {J.~M.}\ \bibnamefont {{Chow}}},\ and\ \bibinfo
  {author} {\bibfnamefont {J.~M.}\ \bibnamefont {{Gambetta}}},\ }\bibfield
  {title} {\bibinfo {title} {{Hardware-efficient variational quantum
  eigensolver for small molecules and quantum magnets}},\ }\href
  {https://doi.org/10.1038/nature23879} {\bibfield  {journal} {\bibinfo
  {journal} {\nat}\ }\textbf {\bibinfo {volume} {549}},\ \bibinfo {pages} {242}
  (\bibinfo {year} {2017})},\ \Eprint {https://arxiv.org/abs/1704.05018}
  {arXiv:1704.05018 [quant-ph]} \BibitemShut {NoStop}%
\bibitem [{\citenamefont {{Arute \emph{et al.}}}(2020)}]{arute2020hartree}%
  \BibitemOpen
  \bibfield  {author} {\bibinfo {author} {\bibfnamefont {F.}~\bibnamefont
  {{Arute \emph{et al.}}}},\ }\bibfield  {title} {\bibinfo {title}
  {{Hartree-Fock on a superconducting qubit quantum computer}},\ }\href
  {https://doi.org/10.1126/science.abb9811} {\bibfield  {journal} {\bibinfo
  {journal} {Science}\ }\textbf {\bibinfo {volume} {369}},\ \bibinfo {pages}
  {1084} (\bibinfo {year} {2020})},\ \Eprint {https://arxiv.org/abs/2004.04174}
  {arXiv:2004.04174 [quant-ph]} \BibitemShut {NoStop}%
\bibitem [{\citenamefont {{Harrigan \emph{et
  al.}}}(2021)}]{harrigan2021quantum}%
  \BibitemOpen
  \bibfield  {author} {\bibinfo {author} {\bibfnamefont {M.~P.}\ \bibnamefont
  {{Harrigan \emph{et al.}}}},\ }\bibfield  {title} {\bibinfo {title} {{Quantum
  approximate optimization of non-planar graph problems on a planar
  superconducting processor}},\ }\href
  {https://doi.org/10.1038/s41567-020-01105-y} {\bibfield  {journal} {\bibinfo
  {journal} {Nature Physics}\ }\textbf {\bibinfo {volume} {17}},\ \bibinfo
  {pages} {332} (\bibinfo {year} {2021})},\ \Eprint
  {https://arxiv.org/abs/2004.04197} {arXiv:2004.04197 [quant-ph]} \BibitemShut
  {NoStop}%
\bibitem [{\citenamefont {Lacroix}\ \emph {et~al.}(2020)\citenamefont
  {Lacroix}, \citenamefont {Hellings}, \citenamefont {Andersen}, \citenamefont
  {Di~Paolo}, \citenamefont {Remm}, \citenamefont {Lazar}, \citenamefont
  {Krinner}, \citenamefont {Norris}, \citenamefont {Gabureac}, \citenamefont
  {Heinsoo}, \citenamefont {Blais}, \citenamefont {Eichler},\ and\
  \citenamefont {Wallraff}}]{lacriox2020improving}%
  \BibitemOpen
  \bibfield  {author} {\bibinfo {author} {\bibfnamefont {N.}~\bibnamefont
  {Lacroix}}, \bibinfo {author} {\bibfnamefont {C.}~\bibnamefont {Hellings}},
  \bibinfo {author} {\bibfnamefont {C.~K.}\ \bibnamefont {Andersen}}, \bibinfo
  {author} {\bibfnamefont {A.}~\bibnamefont {Di~Paolo}}, \bibinfo {author}
  {\bibfnamefont {A.}~\bibnamefont {Remm}}, \bibinfo {author} {\bibfnamefont
  {S.}~\bibnamefont {Lazar}}, \bibinfo {author} {\bibfnamefont
  {S.}~\bibnamefont {Krinner}}, \bibinfo {author} {\bibfnamefont {G.~J.}\
  \bibnamefont {Norris}}, \bibinfo {author} {\bibfnamefont {M.}~\bibnamefont
  {Gabureac}}, \bibinfo {author} {\bibfnamefont {J.}~\bibnamefont {Heinsoo}},
  \bibinfo {author} {\bibfnamefont {A.}~\bibnamefont {Blais}}, \bibinfo
  {author} {\bibfnamefont {C.}~\bibnamefont {Eichler}},\ and\ \bibinfo {author}
  {\bibfnamefont {A.}~\bibnamefont {Wallraff}},\ }\bibfield  {title} {\bibinfo
  {title} {Improving the performance of deep quantum optimization algorithms
  with continuous gate sets},\ }\href
  {https://doi.org/10.1103/PRXQuantum.1.020304} {\bibfield  {journal} {\bibinfo
   {journal} {PRX Quantum}\ }\textbf {\bibinfo {volume} {1}},\ \bibinfo {pages}
  {110304} (\bibinfo {year} {2020})}\BibitemShut {NoStop}%
\bibitem [{\citenamefont {{Havl{\'\i}{\v{c}}ek}}\ \emph
  {et~al.}(2019)\citenamefont {{Havl{\'\i}{\v{c}}ek}}, \citenamefont
  {{C{\'o}rcoles}}, \citenamefont {{Temme}}, \citenamefont {{Harrow}},
  \citenamefont {{Kandala}}, \citenamefont {{Chow}},\ and\ \citenamefont
  {{Gambetta}}}]{havlicek2019supervised}%
  \BibitemOpen
  \bibfield  {author} {\bibinfo {author} {\bibfnamefont {V.}~\bibnamefont
  {{Havl{\'\i}{\v{c}}ek}}}, \bibinfo {author} {\bibfnamefont {A.~D.}\
  \bibnamefont {{C{\'o}rcoles}}}, \bibinfo {author} {\bibfnamefont
  {K.}~\bibnamefont {{Temme}}}, \bibinfo {author} {\bibfnamefont {A.~W.}\
  \bibnamefont {{Harrow}}}, \bibinfo {author} {\bibfnamefont {A.}~\bibnamefont
  {{Kandala}}}, \bibinfo {author} {\bibfnamefont {J.~M.}\ \bibnamefont
  {{Chow}}},\ and\ \bibinfo {author} {\bibfnamefont {J.~M.}\ \bibnamefont
  {{Gambetta}}},\ }\bibfield  {title} {\bibinfo {title} {{Supervised learning
  with quantum-enhanced feature spaces}},\ }\href
  {https://doi.org/10.1038/s41586-019-0980-2} {\bibfield  {journal} {\bibinfo
  {journal} {\nat}\ }\textbf {\bibinfo {volume} {567}},\ \bibinfo {pages} {209}
  (\bibinfo {year} {2019})},\ \Eprint {https://arxiv.org/abs/1804.11326}
  {arXiv:1804.11326 [quant-ph]} \BibitemShut {NoStop}%
\bibitem [{\citenamefont {{Johri}}\ \emph {et~al.}(2021)\citenamefont
  {{Johri}}, \citenamefont {{Debnath}}, \citenamefont {{Mocherla}},
  \citenamefont {{Singk}}, \citenamefont {{Prakash}}, \citenamefont {{Kim}},\
  and\ \citenamefont {{Kerenidis}}}]{johri2021nearest}%
  \BibitemOpen
  \bibfield  {author} {\bibinfo {author} {\bibfnamefont {S.}~\bibnamefont
  {{Johri}}}, \bibinfo {author} {\bibfnamefont {S.}~\bibnamefont {{Debnath}}},
  \bibinfo {author} {\bibfnamefont {A.}~\bibnamefont {{Mocherla}}}, \bibinfo
  {author} {\bibfnamefont {A.}~\bibnamefont {{Singk}}}, \bibinfo {author}
  {\bibfnamefont {A.}~\bibnamefont {{Prakash}}}, \bibinfo {author}
  {\bibfnamefont {J.}~\bibnamefont {{Kim}}},\ and\ \bibinfo {author}
  {\bibfnamefont {I.}~\bibnamefont {{Kerenidis}}},\ }\bibfield  {title}
  {\bibinfo {title} {{Nearest centroid classification on a trapped ion quantum
  computer}},\ }\href {https://doi.org/10.1038/s41534-021-00456-5} {\bibfield
  {journal} {\bibinfo  {journal} {npj Quantum Information}\ }\textbf {\bibinfo
  {volume} {7}},\ \bibinfo {eid} {122} (\bibinfo {year} {2021})},\ \Eprint
  {https://arxiv.org/abs/2012.04145} {arXiv:2012.04145 [quant-ph]} \BibitemShut
  {NoStop}%
\bibitem [{\citenamefont {{McClean}}\ \emph {et~al.}(2018)\citenamefont
  {{McClean}}, \citenamefont {{Boixo}}, \citenamefont {{Smelyanskiy}},
  \citenamefont {{Babbush}},\ and\ \citenamefont
  {{Neven}}}]{mcclean2018barren}%
  \BibitemOpen
  \bibfield  {author} {\bibinfo {author} {\bibfnamefont {J.~R.}\ \bibnamefont
  {{McClean}}}, \bibinfo {author} {\bibfnamefont {S.}~\bibnamefont {{Boixo}}},
  \bibinfo {author} {\bibfnamefont {V.~N.}\ \bibnamefont {{Smelyanskiy}}},
  \bibinfo {author} {\bibfnamefont {R.}~\bibnamefont {{Babbush}}},\ and\
  \bibinfo {author} {\bibfnamefont {H.}~\bibnamefont {{Neven}}},\ }\bibfield
  {title} {\bibinfo {title} {{Barren plateaus in quantum neural network
  training landscapes}},\ }\href {https://doi.org/10.1038/s41467-018-07090-4}
  {\bibfield  {journal} {\bibinfo  {journal} {Nature Communications}\ }\textbf
  {\bibinfo {volume} {9}},\ \bibinfo {eid} {4812} (\bibinfo {year} {2018})},\
  \Eprint {https://arxiv.org/abs/1803.11173} {arXiv:1803.11173 [quant-ph]}
  \BibitemShut {NoStop}%
\bibitem [{\citenamefont {{Bittel}}\ and\ \citenamefont
  {{Kliesch}}(2021)}]{bittel2021training}%
  \BibitemOpen
  \bibfield  {author} {\bibinfo {author} {\bibfnamefont {L.}~\bibnamefont
  {{Bittel}}}\ and\ \bibinfo {author} {\bibfnamefont {M.}~\bibnamefont
  {{Kliesch}}},\ }\bibfield  {title} {\bibinfo {title} {{Training Variational
  Quantum Algorithms Is NP-Hard}},\ }\href
  {https://doi.org/10.1103/PhysRevLett.127.120502} {\bibfield  {journal}
  {\bibinfo  {journal} {\prl}\ }\textbf {\bibinfo {volume} {127}},\ \bibinfo
  {eid} {120502} (\bibinfo {year} {2021})}\BibitemShut {NoStop}%
\bibitem [{\citenamefont {Grant}\ \emph {et~al.}(2019)\citenamefont {Grant},
  \citenamefont {Wossnig}, \citenamefont {Ostaszewski},\ and\ \citenamefont
  {Benedetti}}]{grant2019initialization}%
  \BibitemOpen
  \bibfield  {author} {\bibinfo {author} {\bibfnamefont {E.}~\bibnamefont
  {Grant}}, \bibinfo {author} {\bibfnamefont {L.}~\bibnamefont {Wossnig}},
  \bibinfo {author} {\bibfnamefont {M.}~\bibnamefont {Ostaszewski}},\ and\
  \bibinfo {author} {\bibfnamefont {M.}~\bibnamefont {Benedetti}},\ }\bibfield
  {title} {\bibinfo {title} {An initialization strategy for addressing barren
  plateaus in parametrized quantum circuits},\ }\href
  {https://doi.org/10.22331/q-2019-12-09-214} {\bibfield  {journal} {\bibinfo
  {journal} {{Quantum}}\ }\textbf {\bibinfo {volume} {3}},\ \bibinfo {pages}
  {214} (\bibinfo {year} {2019})}\BibitemShut {NoStop}%
\bibitem [{\citenamefont {{Skolik}}\ \emph {et~al.}(2020)\citenamefont
  {{Skolik}}, \citenamefont {{McClean}}, \citenamefont {{Mohseni}},
  \citenamefont {{van der Smagt}},\ and\ \citenamefont
  {{Leib}}}]{skolik2020layerwise}%
  \BibitemOpen
  \bibfield  {author} {\bibinfo {author} {\bibfnamefont {A.}~\bibnamefont
  {{Skolik}}}, \bibinfo {author} {\bibfnamefont {J.~R.}\ \bibnamefont
  {{McClean}}}, \bibinfo {author} {\bibfnamefont {M.}~\bibnamefont
  {{Mohseni}}}, \bibinfo {author} {\bibfnamefont {P.}~\bibnamefont {{van der
  Smagt}}},\ and\ \bibinfo {author} {\bibfnamefont {M.}~\bibnamefont
  {{Leib}}},\ }\bibfield  {title} {\bibinfo {title} {{Layerwise learning for
  quantum neural networks}},\ }\href@noop {} {\bibfield  {journal} {\bibinfo
  {journal} {arXiv e-prints}\ ,\ \bibinfo {eid} {arXiv:2006.14904}} (\bibinfo
  {year} {2020})},\ \Eprint {https://arxiv.org/abs/2006.14904}
  {arXiv:2006.14904 [quant-ph]} \BibitemShut {NoStop}%
\bibitem [{\citenamefont {{Dborin}}\ \emph {et~al.}(2021)\citenamefont
  {{Dborin}}, \citenamefont {{Barratt}}, \citenamefont {{Wimalaweera}},
  \citenamefont {{Wright}},\ and\ \citenamefont {{Green}}}]{dborin2021matrix}%
  \BibitemOpen
  \bibfield  {author} {\bibinfo {author} {\bibfnamefont {J.}~\bibnamefont
  {{Dborin}}}, \bibinfo {author} {\bibfnamefont {F.}~\bibnamefont {{Barratt}}},
  \bibinfo {author} {\bibfnamefont {V.}~\bibnamefont {{Wimalaweera}}}, \bibinfo
  {author} {\bibfnamefont {L.}~\bibnamefont {{Wright}}},\ and\ \bibinfo
  {author} {\bibfnamefont {A.~G.}\ \bibnamefont {{Green}}},\ }\bibfield
  {title} {\bibinfo {title} {{Matrix Product State Pre-Training for Quantum
  Machine Learning}},\ }\href@noop {} {\bibfield  {journal} {\bibinfo
  {journal} {arXiv e-prints}\ ,\ \bibinfo {eid} {arXiv:2106.05742}} (\bibinfo
  {year} {2021})},\ \Eprint {https://arxiv.org/abs/2106.05742}
  {arXiv:2106.05742 [quant-ph]} \BibitemShut {NoStop}%
\bibitem [{\citenamefont {{Holmes}}\ \emph {et~al.}(2021)\citenamefont
  {{Holmes}}, \citenamefont {{Sharma}}, \citenamefont {{Cerezo}},\ and\
  \citenamefont {{Coles}}}]{holmes2021connecting}%
  \BibitemOpen
  \bibfield  {author} {\bibinfo {author} {\bibfnamefont {Z.}~\bibnamefont
  {{Holmes}}}, \bibinfo {author} {\bibfnamefont {K.}~\bibnamefont {{Sharma}}},
  \bibinfo {author} {\bibfnamefont {M.}~\bibnamefont {{Cerezo}}},\ and\
  \bibinfo {author} {\bibfnamefont {P.~J.}\ \bibnamefont {{Coles}}},\
  }\bibfield  {title} {\bibinfo {title} {{Connecting ansatz expressibility to
  gradient magnitudes and barren plateaus}},\ }\href@noop {} {\bibfield
  {journal} {\bibinfo  {journal} {arXiv e-prints}\ ,\ \bibinfo {eid}
  {arXiv:2101.02138}} (\bibinfo {year} {2021})},\ \Eprint
  {https://arxiv.org/abs/2101.02138} {arXiv:2101.02138 [quant-ph]} \BibitemShut
  {NoStop}%
\bibitem [{\citenamefont {Larocca}\ \emph
  {et~al.}(2021{\natexlab{a}})\citenamefont {Larocca}, \citenamefont {Czarnik},
  \citenamefont {Sharma}, \citenamefont {Muraleedharan}, \citenamefont
  {Coles},\ and\ \citenamefont {Cerezo}}]{larocca2021diagnosing}%
  \BibitemOpen
  \bibfield  {author} {\bibinfo {author} {\bibfnamefont {M.}~\bibnamefont
  {Larocca}}, \bibinfo {author} {\bibfnamefont {P.}~\bibnamefont {Czarnik}},
  \bibinfo {author} {\bibfnamefont {K.}~\bibnamefont {Sharma}}, \bibinfo
  {author} {\bibfnamefont {G.}~\bibnamefont {Muraleedharan}}, \bibinfo {author}
  {\bibfnamefont {P.~J.}\ \bibnamefont {Coles}},\ and\ \bibinfo {author}
  {\bibfnamefont {M.}~\bibnamefont {Cerezo}},\ }\href@noop {} {\bibinfo {title}
  {Diagnosing barren plateaus with tools from quantum optimal control}}
  (\bibinfo {year} {2021}{\natexlab{a}}),\ \Eprint
  {https://arxiv.org/abs/2105.14377} {arXiv:2105.14377 [quant-ph]} \BibitemShut
  {NoStop}%
\bibitem [{\citenamefont {{Cerezo}}\ \emph {et~al.}(2021)\citenamefont
  {{Cerezo}}, \citenamefont {{Sone}}, \citenamefont {{Volkoff}}, \citenamefont
  {{Cincio}},\ and\ \citenamefont {{Coles}}}]{cerezo2020cost}%
  \BibitemOpen
  \bibfield  {author} {\bibinfo {author} {\bibfnamefont {M.}~\bibnamefont
  {{Cerezo}}}, \bibinfo {author} {\bibfnamefont {A.}~\bibnamefont {{Sone}}},
  \bibinfo {author} {\bibfnamefont {T.}~\bibnamefont {{Volkoff}}}, \bibinfo
  {author} {\bibfnamefont {L.}~\bibnamefont {{Cincio}}},\ and\ \bibinfo
  {author} {\bibfnamefont {P.~J.}\ \bibnamefont {{Coles}}},\ }\bibfield
  {title} {\bibinfo {title} {{Cost function dependent barren plateaus in
  shallow parametrized quantum circuits}},\ }\href
  {https://doi.org/10.1038/s41467-021-21728-w} {\bibfield  {journal} {\bibinfo
  {journal} {Nature Communications}\ }\textbf {\bibinfo {volume} {12}},\
  \bibinfo {eid} {1791} (\bibinfo {year} {2021})},\ \Eprint
  {https://arxiv.org/abs/2001.00550} {arXiv:2001.00550 [quant-ph]} \BibitemShut
  {NoStop}%
\bibitem [{\citenamefont {{Uvarov}}\ and\ \citenamefont
  {{Biamonte}}(2020)}]{uvarov2020barren}%
  \BibitemOpen
  \bibfield  {author} {\bibinfo {author} {\bibfnamefont {A.}~\bibnamefont
  {{Uvarov}}}\ and\ \bibinfo {author} {\bibfnamefont {J.}~\bibnamefont
  {{Biamonte}}},\ }\bibfield  {title} {\bibinfo {title} {{On barren plateaus
  and cost function locality in variational quantum algorithms}},\ }\href@noop
  {} {\bibfield  {journal} {\bibinfo  {journal} {arXiv e-prints}\ ,\ \bibinfo
  {eid} {arXiv:2011.10530}} (\bibinfo {year} {2020})},\ \Eprint
  {https://arxiv.org/abs/2011.10530} {arXiv:2011.10530 [quant-ph]} \BibitemShut
  {NoStop}%
\bibitem [{\citenamefont {{Ortiz Marrero}}\ \emph {et~al.}(2020)\citenamefont
  {{Ortiz Marrero}}, \citenamefont {{Kieferov{\'a}}},\ and\ \citenamefont
  {{Wiebe}}}]{marrero2020entanglement}%
  \BibitemOpen
  \bibfield  {author} {\bibinfo {author} {\bibfnamefont {C.}~\bibnamefont
  {{Ortiz Marrero}}}, \bibinfo {author} {\bibfnamefont {M.}~\bibnamefont
  {{Kieferov{\'a}}}},\ and\ \bibinfo {author} {\bibfnamefont {N.}~\bibnamefont
  {{Wiebe}}},\ }\bibfield  {title} {\bibinfo {title} {{Entanglement Induced
  Barren Plateaus}},\ }\href@noop {} {\bibfield  {journal} {\bibinfo  {journal}
  {arXiv e-prints}\ ,\ \bibinfo {eid} {arXiv:2010.15968}} (\bibinfo {year}
  {2020})},\ \Eprint {https://arxiv.org/abs/2010.15968} {arXiv:2010.15968
  [quant-ph]} \BibitemShut {NoStop}%
\bibitem [{\citenamefont {{Wang}}\ \emph {et~al.}(2020)\citenamefont {{Wang}},
  \citenamefont {{Fontana}}, \citenamefont {{Cerezo}}, \citenamefont
  {{Sharma}}, \citenamefont {{Sone}}, \citenamefont {{Cincio}},\ and\
  \citenamefont {{Coles}}}]{wang2020noise}%
  \BibitemOpen
  \bibfield  {author} {\bibinfo {author} {\bibfnamefont {S.}~\bibnamefont
  {{Wang}}}, \bibinfo {author} {\bibfnamefont {E.}~\bibnamefont {{Fontana}}},
  \bibinfo {author} {\bibfnamefont {M.}~\bibnamefont {{Cerezo}}}, \bibinfo
  {author} {\bibfnamefont {K.}~\bibnamefont {{Sharma}}}, \bibinfo {author}
  {\bibfnamefont {A.}~\bibnamefont {{Sone}}}, \bibinfo {author} {\bibfnamefont
  {L.}~\bibnamefont {{Cincio}}},\ and\ \bibinfo {author} {\bibfnamefont
  {P.~J.}\ \bibnamefont {{Coles}}},\ }\bibfield  {title} {\bibinfo {title}
  {{Noise-Induced Barren Plateaus in Variational Quantum Algorithms}},\
  }\href@noop {} {\bibfield  {journal} {\bibinfo  {journal} {arXiv e-prints}\
  ,\ \bibinfo {eid} {arXiv:2007.14384}} (\bibinfo {year} {2020})},\ \Eprint
  {https://arxiv.org/abs/2007.14384} {arXiv:2007.14384 [quant-ph]} \BibitemShut
  {NoStop}%
\bibitem [{\citenamefont {{Kim}}\ and\ \citenamefont
  {{Oz}}(2021{\natexlab{a}})}]{kim2021entanglement}%
  \BibitemOpen
  \bibfield  {author} {\bibinfo {author} {\bibfnamefont {J.}~\bibnamefont
  {{Kim}}}\ and\ \bibinfo {author} {\bibfnamefont {Y.}~\bibnamefont {{Oz}}},\
  }\bibfield  {title} {\bibinfo {title} {{Entanglement Diagnostics for
  Efficient Quantum Computation}},\ }\href@noop {} {\bibfield  {journal}
  {\bibinfo  {journal} {arXiv e-prints}\ ,\ \bibinfo {eid} {arXiv:2102.12534}}
  (\bibinfo {year} {2021}{\natexlab{a}})},\ \Eprint
  {https://arxiv.org/abs/2102.12534} {arXiv:2102.12534 [quant-ph]} \BibitemShut
  {NoStop}%
\bibitem [{\citenamefont {{Kim}}\ and\ \citenamefont
  {{Oz}}(2021{\natexlab{b}})}]{kim2021quantum}%
  \BibitemOpen
  \bibfield  {author} {\bibinfo {author} {\bibfnamefont {J.}~\bibnamefont
  {{Kim}}}\ and\ \bibinfo {author} {\bibfnamefont {Y.}~\bibnamefont {{Oz}}},\
  }\bibfield  {title} {\bibinfo {title} {{Quantum Energy Landscape and VQA
  Optimization}},\ }\href@noop {} {\bibfield  {journal} {\bibinfo  {journal}
  {arXiv e-prints}\ ,\ \bibinfo {eid} {arXiv:2107.10166}} (\bibinfo {year}
  {2021}{\natexlab{b}})},\ \Eprint {https://arxiv.org/abs/2107.10166}
  {arXiv:2107.10166 [quant-ph]} \BibitemShut {NoStop}%
\bibitem [{\citenamefont {{Patti}}\ \emph {et~al.}(2021)\citenamefont
  {{Patti}}, \citenamefont {{Najafi}}, \citenamefont {{Gao}},\ and\
  \citenamefont {{Yelin}}}]{patti2021entanglement}%
  \BibitemOpen
  \bibfield  {author} {\bibinfo {author} {\bibfnamefont {T.~L.}\ \bibnamefont
  {{Patti}}}, \bibinfo {author} {\bibfnamefont {K.}~\bibnamefont {{Najafi}}},
  \bibinfo {author} {\bibfnamefont {X.}~\bibnamefont {{Gao}}},\ and\ \bibinfo
  {author} {\bibfnamefont {S.~F.}\ \bibnamefont {{Yelin}}},\ }\bibfield
  {title} {\bibinfo {title} {{Entanglement devised barren plateau
  mitigation}},\ }\href {https://doi.org/10.1103/PhysRevResearch.3.033090}
  {\bibfield  {journal} {\bibinfo  {journal} {Physical Review Research}\
  }\textbf {\bibinfo {volume} {3}},\ \bibinfo {eid} {033090} (\bibinfo {year}
  {2021})},\ \Eprint {https://arxiv.org/abs/2012.12658} {arXiv:2012.12658
  [quant-ph]} \BibitemShut {NoStop}%
\bibitem [{\citenamefont {{Wiersema}}\ \emph {et~al.}(2021)\citenamefont
  {{Wiersema}}, \citenamefont {{Zhou}}, \citenamefont {{Carrasquilla}},\ and\
  \citenamefont {{Kim}}}]{wiersema2021measurement-induced}%
  \BibitemOpen
  \bibfield  {author} {\bibinfo {author} {\bibfnamefont {R.}~\bibnamefont
  {{Wiersema}}}, \bibinfo {author} {\bibfnamefont {C.}~\bibnamefont {{Zhou}}},
  \bibinfo {author} {\bibfnamefont {J.~F.}\ \bibnamefont {{Carrasquilla}}},\
  and\ \bibinfo {author} {\bibfnamefont {Y.~B.}\ \bibnamefont {{Kim}}},\
  }\bibfield  {title} {\bibinfo {title} {{Measurement-induced entanglement
  phase transitions in variational quantum circuits}},\ }\href@noop {}
  {\bibfield  {journal} {\bibinfo  {journal} {arXiv e-prints}\ ,\ \bibinfo
  {eid} {arXiv:2111.08035}} (\bibinfo {year} {2021})},\ \Eprint
  {https://arxiv.org/abs/2111.08035} {arXiv:2111.08035 [quant-ph]} \BibitemShut
  {NoStop}%
\bibitem [{\citenamefont {{Huang}}\ \emph {et~al.}(2020)\citenamefont
  {{Huang}}, \citenamefont {{Kueng}},\ and\ \citenamefont
  {{Preskill}}}]{huang2020predicting}%
  \BibitemOpen
  \bibfield  {author} {\bibinfo {author} {\bibfnamefont {H.-Y.}\ \bibnamefont
  {{Huang}}}, \bibinfo {author} {\bibfnamefont {R.}~\bibnamefont {{Kueng}}},\
  and\ \bibinfo {author} {\bibfnamefont {J.}~\bibnamefont {{Preskill}}},\
  }\bibfield  {title} {\bibinfo {title} {{Predicting many properties of a
  quantum system from very few measurements}},\ }\href
  {https://doi.org/10.1038/s41567-020-0932-7} {\bibfield  {journal} {\bibinfo
  {journal} {Nature Physics}\ }\textbf {\bibinfo {volume} {16}},\ \bibinfo
  {pages} {1050} (\bibinfo {year} {2020})},\ \Eprint
  {https://arxiv.org/abs/2002.08953} {arXiv:2002.08953 [quant-ph]} \BibitemShut
  {NoStop}%
\bibitem [{\citenamefont {{Haug}}\ \emph {et~al.}(2021)\citenamefont {{Haug}},
  \citenamefont {{Bharti}},\ and\ \citenamefont {{Kim}}}]{haug2021capacity}%
  \BibitemOpen
  \bibfield  {author} {\bibinfo {author} {\bibfnamefont {T.}~\bibnamefont
  {{Haug}}}, \bibinfo {author} {\bibfnamefont {K.}~\bibnamefont {{Bharti}}},\
  and\ \bibinfo {author} {\bibfnamefont {M.~S.}\ \bibnamefont {{Kim}}},\
  }\bibfield  {title} {\bibinfo {title} {{Capacity and Quantum Geometry of
  Parametrized Quantum Circuits}},\ }\href
  {https://doi.org/10.1103/PRXQuantum.2.040309} {\bibfield  {journal} {\bibinfo
   {journal} {PRX Quantum}\ }\textbf {\bibinfo {volume} {2}},\ \bibinfo {eid}
  {040309} (\bibinfo {year} {2021})},\ \Eprint
  {https://arxiv.org/abs/2102.01659} {arXiv:2102.01659 [quant-ph]} \BibitemShut
  {NoStop}%
\bibitem [{\citenamefont {Peruzzo}\ \emph {et~al.}(2014)\citenamefont
  {Peruzzo}, \citenamefont {McClean}, \citenamefont {Shadbolt}, \citenamefont
  {Yung}, \citenamefont {Zhou}, \citenamefont {Love}, \citenamefont
  {Aspuru-Guzik},\ and\ \citenamefont {O'Brien}}]{peruzzo2014VQE}%
  \BibitemOpen
  \bibfield  {author} {\bibinfo {author} {\bibfnamefont {A.}~\bibnamefont
  {Peruzzo}}, \bibinfo {author} {\bibfnamefont {J.}~\bibnamefont {McClean}},
  \bibinfo {author} {\bibfnamefont {P.}~\bibnamefont {Shadbolt}}, \bibinfo
  {author} {\bibfnamefont {M.-H.}\ \bibnamefont {Yung}}, \bibinfo {author}
  {\bibfnamefont {X.-Q.}\ \bibnamefont {Zhou}}, \bibinfo {author}
  {\bibfnamefont {P.~J.}\ \bibnamefont {Love}}, \bibinfo {author}
  {\bibfnamefont {A.}~\bibnamefont {Aspuru-Guzik}},\ and\ \bibinfo {author}
  {\bibfnamefont {J.~L.}\ \bibnamefont {O'Brien}},\ }\bibfield  {title}
  {\bibinfo {title} {A variational eigenvalue solver on a photonic quantum
  processor},\ }\bibfield  {journal} {\bibinfo  {journal} {Nature
  Communications}\ }\textbf {\bibinfo {volume} {5}},\ \href
  {https://doi.org/10.1038/ncomms5213} {10.1038/ncomms5213} (\bibinfo {year}
  {2014})\BibitemShut {NoStop}%
\bibitem [{\citenamefont {Ostaszewski}\ \emph {et~al.}(2021)\citenamefont
  {Ostaszewski}, \citenamefont {Grant},\ and\ \citenamefont
  {Benedetti}}]{ostaszewski2021structure}%
  \BibitemOpen
  \bibfield  {author} {\bibinfo {author} {\bibfnamefont {M.}~\bibnamefont
  {Ostaszewski}}, \bibinfo {author} {\bibfnamefont {E.}~\bibnamefont {Grant}},\
  and\ \bibinfo {author} {\bibfnamefont {M.}~\bibnamefont {Benedetti}},\
  }\bibfield  {title} {\bibinfo {title} {Structure optimization for
  parameterized quantum circuits},\ }\href
  {https://doi.org/10.22331/q-2021-01-28-391} {\bibfield  {journal} {\bibinfo
  {journal} {{Quantum}}\ }\textbf {\bibinfo {volume} {5}},\ \bibinfo {pages}
  {391} (\bibinfo {year} {2021})}\BibitemShut {NoStop}%
\bibitem [{\citenamefont {Stokes}\ \emph {et~al.}(2020)\citenamefont {Stokes},
  \citenamefont {Izaac}, \citenamefont {Killoran},\ and\ \citenamefont
  {Carleo}}]{stokes2020quantum}%
  \BibitemOpen
  \bibfield  {author} {\bibinfo {author} {\bibfnamefont {J.}~\bibnamefont
  {Stokes}}, \bibinfo {author} {\bibfnamefont {J.}~\bibnamefont {Izaac}},
  \bibinfo {author} {\bibfnamefont {N.}~\bibnamefont {Killoran}},\ and\
  \bibinfo {author} {\bibfnamefont {G.}~\bibnamefont {Carleo}},\ }\bibfield
  {title} {\bibinfo {title} {Quantum {N}atural {G}radient},\ }\href
  {https://doi.org/10.22331/q-2020-05-25-269} {\bibfield  {journal} {\bibinfo
  {journal} {{Quantum}}\ }\textbf {\bibinfo {volume} {4}},\ \bibinfo {pages}
  {269} (\bibinfo {year} {2020})}\BibitemShut {NoStop}%
\bibitem [{\citenamefont {{Kingma}}\ and\ \citenamefont {{Ba}}(2014)}]{adam}%
  \BibitemOpen
  \bibfield  {author} {\bibinfo {author} {\bibfnamefont {D.~P.}\ \bibnamefont
  {{Kingma}}}\ and\ \bibinfo {author} {\bibfnamefont {J.}~\bibnamefont
  {{Ba}}},\ }\bibfield  {title} {\bibinfo {title} {{Adam: A Method for
  Stochastic Optimization}},\ }\href@noop {} {\bibfield  {journal} {\bibinfo
  {journal} {arXiv e-prints}\ ,\ \bibinfo {eid} {arXiv:1412.6980}} (\bibinfo
  {year} {2014})},\ \Eprint {https://arxiv.org/abs/1412.6980} {arXiv:1412.6980
  [cs.LG]} \BibitemShut {NoStop}%
\bibitem [{\citenamefont {Gacon}\ \emph {et~al.}(2021)\citenamefont {Gacon},
  \citenamefont {Zoufal}, \citenamefont {Carleo},\ and\ \citenamefont
  {Woerner}}]{gacon2021simultaneous}%
  \BibitemOpen
  \bibfield  {author} {\bibinfo {author} {\bibfnamefont {J.}~\bibnamefont
  {Gacon}}, \bibinfo {author} {\bibfnamefont {C.}~\bibnamefont {Zoufal}},
  \bibinfo {author} {\bibfnamefont {G.}~\bibnamefont {Carleo}},\ and\ \bibinfo
  {author} {\bibfnamefont {S.}~\bibnamefont {Woerner}},\ }\bibfield  {title}
  {\bibinfo {title} {Simultaneous {P}erturbation {S}tochastic {A}pproximation
  of the {Q}uantum {F}isher {I}nformation},\ }\href
  {https://doi.org/10.22331/q-2021-10-20-567} {\bibfield  {journal} {\bibinfo
  {journal} {{Quantum}}\ }\textbf {\bibinfo {volume} {5}},\ \bibinfo {pages}
  {567} (\bibinfo {year} {2021})}\BibitemShut {NoStop}%
\bibitem [{\citenamefont {Flammia}\ and\ \citenamefont
  {Liu}(2011)}]{flammia11direct}%
  \BibitemOpen
  \bibfield  {author} {\bibinfo {author} {\bibfnamefont {S.~T.}\ \bibnamefont
  {Flammia}}\ and\ \bibinfo {author} {\bibfnamefont {Y.-K.}\ \bibnamefont
  {Liu}},\ }\bibfield  {title} {\bibinfo {title} {Direct fidelity estimation
  from few pauli measurements},\ }\href
  {https://doi.org/10.1103/PhysRevLett.106.230501} {\bibfield  {journal}
  {\bibinfo  {journal} {Phys. Rev. Lett.}\ }\textbf {\bibinfo {volume} {106}},\
  \bibinfo {pages} {230501} (\bibinfo {year} {2011})}\BibitemShut {NoStop}%
\bibitem [{\citenamefont {{Huang}}\ \emph
  {et~al.}(2021{\natexlab{a}})\citenamefont {{Huang}}, \citenamefont
  {{Broughton}}, \citenamefont {{Cotler}}, \citenamefont {{Chen}},
  \citenamefont {{Li}}, \citenamefont {{Mohseni}}, \citenamefont {{Neven}},
  \citenamefont {{Babbush}}, \citenamefont {{Kueng}}, \citenamefont
  {{Preskill}},\ and\ \citenamefont {{McClean}}}]{huang2021learning}%
  \BibitemOpen
  \bibfield  {author} {\bibinfo {author} {\bibfnamefont {H.-Y.}\ \bibnamefont
  {{Huang}}}, \bibinfo {author} {\bibfnamefont {M.}~\bibnamefont
  {{Broughton}}}, \bibinfo {author} {\bibfnamefont {J.}~\bibnamefont
  {{Cotler}}}, \bibinfo {author} {\bibfnamefont {S.}~\bibnamefont {{Chen}}},
  \bibinfo {author} {\bibfnamefont {J.}~\bibnamefont {{Li}}}, \bibinfo {author}
  {\bibfnamefont {M.}~\bibnamefont {{Mohseni}}}, \bibinfo {author}
  {\bibfnamefont {H.}~\bibnamefont {{Neven}}}, \bibinfo {author} {\bibfnamefont
  {R.}~\bibnamefont {{Babbush}}}, \bibinfo {author} {\bibfnamefont
  {R.}~\bibnamefont {{Kueng}}}, \bibinfo {author} {\bibfnamefont
  {J.}~\bibnamefont {{Preskill}}},\ and\ \bibinfo {author} {\bibfnamefont
  {J.~R.}\ \bibnamefont {{McClean}}},\ }\bibfield  {title} {\bibinfo {title}
  {{Quantum advantage in learning from experiments}},\ }\href@noop {}
  {\bibfield  {journal} {\bibinfo  {journal} {arXiv e-prints}\ ,\ \bibinfo
  {eid} {arXiv:2112.00778}} (\bibinfo {year} {2021}{\natexlab{a}})},\ \Eprint
  {https://arxiv.org/abs/2112.00778} {arXiv:2112.00778 [quant-ph]} \BibitemShut
  {NoStop}%
\bibitem [{\citenamefont {{Chen}}\ \emph {et~al.}(2021)\citenamefont {{Chen}},
  \citenamefont {{Cotler}}, \citenamefont {{Huang}},\ and\ \citenamefont
  {{Li}}}]{chen2021exponential}%
  \BibitemOpen
  \bibfield  {author} {\bibinfo {author} {\bibfnamefont {S.}~\bibnamefont
  {{Chen}}}, \bibinfo {author} {\bibfnamefont {J.}~\bibnamefont {{Cotler}}},
  \bibinfo {author} {\bibfnamefont {H.-Y.}\ \bibnamefont {{Huang}}},\ and\
  \bibinfo {author} {\bibfnamefont {J.}~\bibnamefont {{Li}}},\ }\bibfield
  {title} {\bibinfo {title} {{Exponential separations between learning with and
  without quantum memory}},\ }\href@noop {} {\bibfield  {journal} {\bibinfo
  {journal} {arXiv e-prints}\ ,\ \bibinfo {eid} {arXiv:2111.05881}} (\bibinfo
  {year} {2021})},\ \Eprint {https://arxiv.org/abs/2111.05881}
  {arXiv:2111.05881 [quant-ph]} \BibitemShut {NoStop}%
\bibitem [{\citenamefont {Arrasmith}\ \emph {et~al.}(2021)\citenamefont
  {Arrasmith}, \citenamefont {Cerezo}, \citenamefont {Czarnik}, \citenamefont
  {Cincio},\ and\ \citenamefont {Coles}}]{arrasmith2021effect}%
  \BibitemOpen
  \bibfield  {author} {\bibinfo {author} {\bibfnamefont {A.}~\bibnamefont
  {Arrasmith}}, \bibinfo {author} {\bibfnamefont {M.}~\bibnamefont {Cerezo}},
  \bibinfo {author} {\bibfnamefont {P.}~\bibnamefont {Czarnik}}, \bibinfo
  {author} {\bibfnamefont {L.}~\bibnamefont {Cincio}},\ and\ \bibinfo {author}
  {\bibfnamefont {P.~J.}\ \bibnamefont {Coles}},\ }\bibfield  {title} {\bibinfo
  {title} {Effect of barren plateaus on gradient-free optimization},\ }\href
  {https://doi.org/10.22331/q-2021-10-05-558} {\bibfield  {journal} {\bibinfo
  {journal} {{Quantum}}\ }\textbf {\bibinfo {volume} {5}},\ \bibinfo {pages}
  {558} (\bibinfo {year} {2021})}\BibitemShut {NoStop}%
\bibitem [{\citenamefont {{Cirac}}\ \emph {et~al.}(2020)\citenamefont
  {{Cirac}}, \citenamefont {{Perez-Garcia}}, \citenamefont {{Schuch}},\ and\
  \citenamefont {{Verstraete}}}]{cirac2020matrix}%
  \BibitemOpen
  \bibfield  {author} {\bibinfo {author} {\bibfnamefont {I.}~\bibnamefont
  {{Cirac}}}, \bibinfo {author} {\bibfnamefont {D.}~\bibnamefont
  {{Perez-Garcia}}}, \bibinfo {author} {\bibfnamefont {N.}~\bibnamefont
  {{Schuch}}},\ and\ \bibinfo {author} {\bibfnamefont {F.}~\bibnamefont
  {{Verstraete}}},\ }\bibfield  {title} {\bibinfo {title} {{Matrix Product
  States and Projected Entangled Pair States: Concepts, Symmetries, and
  Theorems}},\ }\href@noop {} {\bibfield  {journal} {\bibinfo  {journal} {arXiv
  e-prints}\ ,\ \bibinfo {eid} {arXiv:2011.12127}} (\bibinfo {year} {2020})},\
  \Eprint {https://arxiv.org/abs/2011.12127} {arXiv:2011.12127 [quant-ph]}
  \BibitemShut {NoStop}%
\bibitem [{\citenamefont {{Verdon}}\ \emph {et~al.}(2019)\citenamefont
  {{Verdon}}, \citenamefont {{Broughton}}, \citenamefont {{McClean}},
  \citenamefont {{Sung}}, \citenamefont {{Babbush}}, \citenamefont {{Jiang}},
  \citenamefont {{Neven}},\ and\ \citenamefont
  {{Mohseni}}}]{verdon2019learning}%
  \BibitemOpen
  \bibfield  {author} {\bibinfo {author} {\bibfnamefont {G.}~\bibnamefont
  {{Verdon}}}, \bibinfo {author} {\bibfnamefont {M.}~\bibnamefont
  {{Broughton}}}, \bibinfo {author} {\bibfnamefont {J.~R.}\ \bibnamefont
  {{McClean}}}, \bibinfo {author} {\bibfnamefont {K.~J.}\ \bibnamefont
  {{Sung}}}, \bibinfo {author} {\bibfnamefont {R.}~\bibnamefont {{Babbush}}},
  \bibinfo {author} {\bibfnamefont {Z.}~\bibnamefont {{Jiang}}}, \bibinfo
  {author} {\bibfnamefont {H.}~\bibnamefont {{Neven}}},\ and\ \bibinfo {author}
  {\bibfnamefont {M.}~\bibnamefont {{Mohseni}}},\ }\bibfield  {title} {\bibinfo
  {title} {{Learning to learn with quantum neural networks via classical neural
  networks}},\ }\href@noop {} {\bibfield  {journal} {\bibinfo  {journal} {arXiv
  e-prints}\ ,\ \bibinfo {eid} {arXiv:1907.05415}} (\bibinfo {year} {2019})},\
  \Eprint {https://arxiv.org/abs/1907.05415} {arXiv:1907.05415 [quant-ph]}
  \BibitemShut {NoStop}%
\bibitem [{\citenamefont {Volkoff}\ and\ \citenamefont
  {Coles}(2021)}]{volkoff2021largegradients}%
  \BibitemOpen
  \bibfield  {author} {\bibinfo {author} {\bibfnamefont {T.}~\bibnamefont
  {Volkoff}}\ and\ \bibinfo {author} {\bibfnamefont {P.~J.}\ \bibnamefont
  {Coles}},\ }\bibfield  {title} {\bibinfo {title} {Large gradients via
  correlation in random parameterized quantum circuits},\ }\href
  {https://doi.org/10.1088/2058-9565/abd891} {\bibfield  {journal} {\bibinfo
  {journal} {Quantum Science and Technology}\ }\textbf {\bibinfo {volume}
  {6}},\ \bibinfo {pages} {025008} (\bibinfo {year} {2021})}\BibitemShut
  {NoStop}%
\bibitem [{\citenamefont {{Mitarai}}\ \emph {et~al.}(2018)\citenamefont
  {{Mitarai}}, \citenamefont {{Negoro}}, \citenamefont {{Kitagawa}},\ and\
  \citenamefont {{Fujii}}}]{mitarai2018quantum}%
  \BibitemOpen
  \bibfield  {author} {\bibinfo {author} {\bibfnamefont {K.}~\bibnamefont
  {{Mitarai}}}, \bibinfo {author} {\bibfnamefont {M.}~\bibnamefont {{Negoro}}},
  \bibinfo {author} {\bibfnamefont {M.}~\bibnamefont {{Kitagawa}}},\ and\
  \bibinfo {author} {\bibfnamefont {K.}~\bibnamefont {{Fujii}}},\ }\bibfield
  {title} {\bibinfo {title} {{Quantum circuit learning}},\ }\href
  {https://doi.org/10.1103/PhysRevA.98.032309} {\bibfield  {journal} {\bibinfo
  {journal} {\pra}\ }\textbf {\bibinfo {volume} {98}},\ \bibinfo {eid} {032309}
  (\bibinfo {year} {2018})},\ \Eprint {https://arxiv.org/abs/1803.00745}
  {arXiv:1803.00745 [quant-ph]} \BibitemShut {NoStop}%
\bibitem [{\citenamefont {{Schuld}}\ \emph {et~al.}(2019)\citenamefont
  {{Schuld}}, \citenamefont {{Bergholm}}, \citenamefont {{Gogolin}},
  \citenamefont {{Izaac}},\ and\ \citenamefont
  {{Killoran}}}]{schuld2019evaluating}%
  \BibitemOpen
  \bibfield  {author} {\bibinfo {author} {\bibfnamefont {M.}~\bibnamefont
  {{Schuld}}}, \bibinfo {author} {\bibfnamefont {V.}~\bibnamefont
  {{Bergholm}}}, \bibinfo {author} {\bibfnamefont {C.}~\bibnamefont
  {{Gogolin}}}, \bibinfo {author} {\bibfnamefont {J.}~\bibnamefont {{Izaac}}},\
  and\ \bibinfo {author} {\bibfnamefont {N.}~\bibnamefont {{Killoran}}},\
  }\bibfield  {title} {\bibinfo {title} {{Evaluating analytic gradients on
  quantum hardware}},\ }\href {https://doi.org/10.1103/PhysRevA.99.032331}
  {\bibfield  {journal} {\bibinfo  {journal} {\pra}\ }\textbf {\bibinfo
  {volume} {99}},\ \bibinfo {eid} {032331} (\bibinfo {year} {2019})},\ \Eprint
  {https://arxiv.org/abs/1811.11184} {arXiv:1811.11184 [quant-ph]} \BibitemShut
  {NoStop}%
\bibitem [{\citenamefont {{Popescu}}\ \emph {et~al.}(2006)\citenamefont
  {{Popescu}}, \citenamefont {{Short}},\ and\ \citenamefont
  {{Winter}}}]{popescu2006entanglement}%
  \BibitemOpen
  \bibfield  {author} {\bibinfo {author} {\bibfnamefont {S.}~\bibnamefont
  {{Popescu}}}, \bibinfo {author} {\bibfnamefont {A.~J.}\ \bibnamefont
  {{Short}}},\ and\ \bibinfo {author} {\bibfnamefont {A.}~\bibnamefont
  {{Winter}}},\ }\bibfield  {title} {\bibinfo {title} {{Entanglement and the
  foundations of statistical mechanics}},\ }\href
  {https://doi.org/10.1038/nphys444} {\bibfield  {journal} {\bibinfo  {journal}
  {Nature Physics}\ }\textbf {\bibinfo {volume} {2}},\ \bibinfo {pages} {754}
  (\bibinfo {year} {2006})},\ \Eprint {https://arxiv.org/abs/quant-ph/0511225}
  {arXiv:quant-ph/0511225 [quant-ph]} \BibitemShut {NoStop}%
\bibitem [{\citenamefont {Oliveira}\ \emph {et~al.}(2007)\citenamefont
  {Oliveira}, \citenamefont {Dahlsten},\ and\ \citenamefont
  {Plenio}}]{oliveira2007}%
  \BibitemOpen
  \bibfield  {author} {\bibinfo {author} {\bibfnamefont {R.}~\bibnamefont
  {Oliveira}}, \bibinfo {author} {\bibfnamefont {O.~C.~O.}\ \bibnamefont
  {Dahlsten}},\ and\ \bibinfo {author} {\bibfnamefont {M.~B.}\ \bibnamefont
  {Plenio}},\ }\bibfield  {title} {\bibinfo {title} {Generic entanglement can
  be generated efficiently},\ }\bibfield  {journal} {\bibinfo  {journal}
  {Physical Review Letters}\ }\textbf {\bibinfo {volume} {98}},\ \href
  {https://doi.org/10.1103/physrevlett.98.130502}
  {10.1103/physrevlett.98.130502} (\bibinfo {year} {2007})\BibitemShut
  {NoStop}%
\bibitem [{\citenamefont {Dahlsten}\ \emph {et~al.}(2007)\citenamefont
  {Dahlsten}, \citenamefont {Oliveira},\ and\ \citenamefont
  {Plenio}}]{dahlsten2007typicalentanglement}%
  \BibitemOpen
  \bibfield  {author} {\bibinfo {author} {\bibfnamefont {O.~C.~O.}\
  \bibnamefont {Dahlsten}}, \bibinfo {author} {\bibfnamefont {R.}~\bibnamefont
  {Oliveira}},\ and\ \bibinfo {author} {\bibfnamefont {M.~B.}\ \bibnamefont
  {Plenio}},\ }\bibfield  {title} {\bibinfo {title} {The emergence of typical
  entanglement in two-party random processes},\ }\href
  {https://doi.org/10.1088/1751-8113/40/28/s16} {\bibfield  {journal} {\bibinfo
   {journal} {Journal of Physics A: Mathematical and Theoretical}\ }\textbf
  {\bibinfo {volume} {40}},\ \bibinfo {pages} {8081} (\bibinfo {year}
  {2007})}\BibitemShut {NoStop}%
\bibitem [{\citenamefont {Chen}\ \emph {et~al.}(2016)\citenamefont {Chen},
  \citenamefont {Ma}, \citenamefont {Nikoufar},\ and\ \citenamefont
  {Fei}}]{renyis_continuitybound}%
  \BibitemOpen
  \bibfield  {author} {\bibinfo {author} {\bibfnamefont {Z.}~\bibnamefont
  {Chen}}, \bibinfo {author} {\bibfnamefont {Z.}~\bibnamefont {Ma}}, \bibinfo
  {author} {\bibfnamefont {I.}~\bibnamefont {Nikoufar}},\ and\ \bibinfo
  {author} {\bibfnamefont {S.-M.}\ \bibnamefont {Fei}},\ }\bibfield  {title}
  {\bibinfo {title} {Sharp continuity bounds for entropy and conditional
  entropy},\ }\href {https://doi.org/10.1007/s11433-016-0367-x} {\bibfield
  {journal} {\bibinfo  {journal} {Science China Physics, Mechanics \&
  Astronomy}\ }\textbf {\bibinfo {volume} {60}},\ \bibinfo {pages} {020321}
  (\bibinfo {year} {2016})}\BibitemShut {NoStop}%
\bibitem [{\citenamefont {Meyer}(2021)}]{meyer2021fisher}%
  \BibitemOpen
  \bibfield  {author} {\bibinfo {author} {\bibfnamefont {J.~J.}\ \bibnamefont
  {Meyer}},\ }\bibfield  {title} {\bibinfo {title} {Fisher {I}nformation in
  {N}oisy {I}ntermediate-{S}cale {Q}uantum {A}pplications},\ }\href
  {https://doi.org/10.22331/q-2021-09-09-539} {\bibfield  {journal} {\bibinfo
  {journal} {{Quantum}}\ }\textbf {\bibinfo {volume} {5}},\ \bibinfo {pages}
  {539} (\bibinfo {year} {2021})}\BibitemShut {NoStop}%
\bibitem [{\citenamefont {{Rath}}\ \emph {et~al.}(2021)\citenamefont {{Rath}},
  \citenamefont {{Branciard}}, \citenamefont {{Minguzzi}},\ and\ \citenamefont
  {{Vermersch}}}]{rath2021quantum}%
  \BibitemOpen
  \bibfield  {author} {\bibinfo {author} {\bibfnamefont {A.}~\bibnamefont
  {{Rath}}}, \bibinfo {author} {\bibfnamefont {C.}~\bibnamefont {{Branciard}}},
  \bibinfo {author} {\bibfnamefont {A.}~\bibnamefont {{Minguzzi}}},\ and\
  \bibinfo {author} {\bibfnamefont {B.}~\bibnamefont {{Vermersch}}},\
  }\bibfield  {title} {\bibinfo {title} {{Quantum Fisher Information from
  Randomized Measurements}},\ }\href
  {https://doi.org/10.1103/PhysRevLett.127.260501} {\bibfield  {journal}
  {\bibinfo  {journal} {\prl}\ }\textbf {\bibinfo {volume} {127}},\ \bibinfo
  {eid} {260501} (\bibinfo {year} {2021})},\ \Eprint
  {https://arxiv.org/abs/2105.13164} {arXiv:2105.13164 [quant-ph]} \BibitemShut
  {NoStop}%
\bibitem [{\citenamefont {Vidal}(2003)}]{vidal2003efficientclassical}%
  \BibitemOpen
  \bibfield  {author} {\bibinfo {author} {\bibfnamefont {G.}~\bibnamefont
  {Vidal}},\ }\bibfield  {title} {\bibinfo {title} {Efficient classical
  simulation of slightly entangled quantum computations},\ }\href
  {https://doi.org/10.1103/PhysRevLett.91.147902} {\bibfield  {journal}
  {\bibinfo  {journal} {Phys. Rev. Lett.}\ }\textbf {\bibinfo {volume} {91}},\
  \bibinfo {pages} {147902} (\bibinfo {year} {2003})}\BibitemShut {NoStop}%
\bibitem [{\citenamefont {Van~den Nest}\ \emph {et~al.}(2007)\citenamefont
  {Van~den Nest}, \citenamefont {D\"ur}, \citenamefont {Vidal},\ and\
  \citenamefont {Briegel}}]{nest2007classicalsimulation}%
  \BibitemOpen
  \bibfield  {author} {\bibinfo {author} {\bibfnamefont {M.}~\bibnamefont
  {Van~den Nest}}, \bibinfo {author} {\bibfnamefont {W.}~\bibnamefont {D\"ur}},
  \bibinfo {author} {\bibfnamefont {G.}~\bibnamefont {Vidal}},\ and\ \bibinfo
  {author} {\bibfnamefont {H.~J.}\ \bibnamefont {Briegel}},\ }\bibfield
  {title} {\bibinfo {title} {Classical simulation versus universality in
  measurement-based quantum computation},\ }\href
  {https://doi.org/10.1103/PhysRevA.75.012337} {\bibfield  {journal} {\bibinfo
  {journal} {Phys. Rev. A}\ }\textbf {\bibinfo {volume} {75}},\ \bibinfo
  {pages} {012337} (\bibinfo {year} {2007})}\BibitemShut {NoStop}%
\bibitem [{\citenamefont {Brand{\~a}o}\ and\ \citenamefont
  {Horodecki}(2013)}]{brandao2013arealaw}%
  \BibitemOpen
  \bibfield  {author} {\bibinfo {author} {\bibfnamefont {F.~G. S.~L.}\
  \bibnamefont {Brand{\~a}o}}\ and\ \bibinfo {author} {\bibfnamefont
  {M.}~\bibnamefont {Horodecki}},\ }\bibfield  {title} {\bibinfo {title} {An
  area law for entanglement from exponential decay of correlations},\
  }\href@noop {} {\bibfield  {journal} {\bibinfo  {journal} {Nature Physics}\
  }\textbf {\bibinfo {volume} {9}},\ \bibinfo {pages} {721} (\bibinfo {year}
  {2013})}\BibitemShut {NoStop}%
\bibitem [{\citenamefont {{Eisert}}\ \emph {et~al.}(2010)\citenamefont
  {{Eisert}}, \citenamefont {{Cramer}},\ and\ \citenamefont
  {{Plenio}}}]{eisert2010colloquium}%
  \BibitemOpen
  \bibfield  {author} {\bibinfo {author} {\bibfnamefont {J.}~\bibnamefont
  {{Eisert}}}, \bibinfo {author} {\bibfnamefont {M.}~\bibnamefont {{Cramer}}},\
  and\ \bibinfo {author} {\bibfnamefont {M.~B.}\ \bibnamefont {{Plenio}}},\
  }\bibfield  {title} {\bibinfo {title} {{Colloquium: Area laws for the
  entanglement entropy}},\ }\href {https://doi.org/10.1103/RevModPhys.82.277}
  {\bibfield  {journal} {\bibinfo  {journal} {Reviews of Modern Physics}\
  }\textbf {\bibinfo {volume} {82}},\ \bibinfo {pages} {277} (\bibinfo {year}
  {2010})},\ \Eprint {https://arxiv.org/abs/0808.3773} {arXiv:0808.3773
  [quant-ph]} \BibitemShut {NoStop}%
\bibitem [{\citenamefont {Schollw{\"o}ck}(2011)}]{DMRG_in_the_age_of_MPS}%
  \BibitemOpen
  \bibfield  {author} {\bibinfo {author} {\bibfnamefont {U.}~\bibnamefont
  {Schollw{\"o}ck}},\ }\bibfield  {title} {\bibinfo {title} {The density-matrix
  renormalization group in the age of matrix product states},\ }\href
  {https://doi.org/https://doi.org/10.1016/j.aop.2010.09.012} {\bibfield
  {journal} {\bibinfo  {journal} {Annals of Physics}\ }\textbf {\bibinfo
  {volume} {326}},\ \bibinfo {pages} {96 } (\bibinfo {year} {2011})},\ \bibinfo
  {note} {january 2011 Special Issue}\BibitemShut {NoStop}%
\bibitem [{\citenamefont {Kitaev}(2015)}]{kitaevTalks}%
  \BibitemOpen
  \bibfield  {author} {\bibinfo {author} {\bibfnamefont {A.}~\bibnamefont
  {Kitaev}},\ }\bibfield  {title} {\bibinfo {title} {A simple model of quantum
  holography}} (\bibinfo {year} {2015}),\ \bibinfo {note} {talks at KITP, April
  7, 2015 and May 27, 2015.}\BibitemShut {Stop}%
\bibitem [{\citenamefont {Huang}\ and\ \citenamefont
  {Gu}(2019)}]{yichen2019sykentanglement}%
  \BibitemOpen
  \bibfield  {author} {\bibinfo {author} {\bibfnamefont {Y.}~\bibnamefont
  {Huang}}\ and\ \bibinfo {author} {\bibfnamefont {Y.}~\bibnamefont {Gu}},\
  }\bibfield  {title} {\bibinfo {title} {Eigenstate entanglement in the
  sachdev-ye-kitaev model},\ }\href
  {https://doi.org/10.1103/PhysRevD.100.041901} {\bibfield  {journal} {\bibinfo
   {journal} {Phys. Rev. D}\ }\textbf {\bibinfo {volume} {100}},\ \bibinfo
  {pages} {041901} (\bibinfo {year} {2019})}\BibitemShut {NoStop}%
\bibitem [{\citenamefont {{Benedetti}}\ \emph {et~al.}(2019)\citenamefont
  {{Benedetti}}, \citenamefont {{Lloyd}}, \citenamefont {{Sack}},\ and\
  \citenamefont {{Fiorentini}}}]{benedetti2019parameterized}%
  \BibitemOpen
  \bibfield  {author} {\bibinfo {author} {\bibfnamefont {M.}~\bibnamefont
  {{Benedetti}}}, \bibinfo {author} {\bibfnamefont {E.}~\bibnamefont
  {{Lloyd}}}, \bibinfo {author} {\bibfnamefont {S.}~\bibnamefont {{Sack}}},\
  and\ \bibinfo {author} {\bibfnamefont {M.}~\bibnamefont {{Fiorentini}}},\
  }\bibfield  {title} {\bibinfo {title} {{Parameterized quantum circuits as
  machine learning models}},\ }\href {https://doi.org/10.1088/2058-9565/ab4eb5}
  {\bibfield  {journal} {\bibinfo  {journal} {Quantum Science and Technology}\
  }\textbf {\bibinfo {volume} {4}},\ \bibinfo {eid} {043001} (\bibinfo {year}
  {2019})},\ \Eprint {https://arxiv.org/abs/1906.07682} {arXiv:1906.07682
  [quant-ph]} \BibitemShut {NoStop}%
\bibitem [{\citenamefont {{Schuld}}\ \emph {et~al.}(2020)\citenamefont
  {{Schuld}}, \citenamefont {{Bocharov}}, \citenamefont {{Svore}},\ and\
  \citenamefont {{Wiebe}}}]{schuld2020circuit-centric}%
  \BibitemOpen
  \bibfield  {author} {\bibinfo {author} {\bibfnamefont {M.}~\bibnamefont
  {{Schuld}}}, \bibinfo {author} {\bibfnamefont {A.}~\bibnamefont
  {{Bocharov}}}, \bibinfo {author} {\bibfnamefont {K.~M.}\ \bibnamefont
  {{Svore}}},\ and\ \bibinfo {author} {\bibfnamefont {N.}~\bibnamefont
  {{Wiebe}}},\ }\bibfield  {title} {\bibinfo {title} {{Circuit-centric quantum
  classifiers}},\ }\href {https://doi.org/10.1103/PhysRevA.101.032308}
  {\bibfield  {journal} {\bibinfo  {journal} {\pra}\ }\textbf {\bibinfo
  {volume} {101}},\ \bibinfo {eid} {032308} (\bibinfo {year}
  {2020})}\BibitemShut {NoStop}%
\bibitem [{\citenamefont {{Farhi}}\ \emph {et~al.}(2014)\citenamefont
  {{Farhi}}, \citenamefont {{Goldstone}},\ and\ \citenamefont
  {{Gutmann}}}]{farhi2014quantum}%
  \BibitemOpen
  \bibfield  {author} {\bibinfo {author} {\bibfnamefont {E.}~\bibnamefont
  {{Farhi}}}, \bibinfo {author} {\bibfnamefont {J.}~\bibnamefont
  {{Goldstone}}},\ and\ \bibinfo {author} {\bibfnamefont {S.}~\bibnamefont
  {{Gutmann}}},\ }\bibfield  {title} {\bibinfo {title} {{A Quantum Approximate
  Optimization Algorithm}},\ }\href@noop {} {\bibfield  {journal} {\bibinfo
  {journal} {arXiv e-prints}\ ,\ \bibinfo {eid} {arXiv:1411.4028}} (\bibinfo
  {year} {2014})},\ \Eprint {https://arxiv.org/abs/1411.4028} {arXiv:1411.4028
  [quant-ph]} \BibitemShut {NoStop}%
\bibitem [{\citenamefont {{Sack}}\ and\ \citenamefont
  {{Serbyn}}(2021)}]{sack2021quantumannealing}%
  \BibitemOpen
  \bibfield  {author} {\bibinfo {author} {\bibfnamefont {S.~H.}\ \bibnamefont
  {{Sack}}}\ and\ \bibinfo {author} {\bibfnamefont {M.}~\bibnamefont
  {{Serbyn}}},\ }\bibfield  {title} {\bibinfo {title} {{Quantum annealing
  initialization of the quantum approximate optimization algorithm}},\ }\href
  {https://doi.org/10.22331/q-2021-07-01-491} {\bibfield  {journal} {\bibinfo
  {journal} {Quantum}\ }\textbf {\bibinfo {volume} {5}},\ \bibinfo {pages}
  {491} (\bibinfo {year} {2021})},\ \Eprint {https://arxiv.org/abs/2101.05742}
  {arXiv:2101.05742 [quant-ph]} \BibitemShut {NoStop}%
\bibitem [{\citenamefont {Barison}\ \emph {et~al.}(2021)\citenamefont
  {Barison}, \citenamefont {Vicentini},\ and\ \citenamefont
  {Carleo}}]{barison2021efficientquantum}%
  \BibitemOpen
  \bibfield  {author} {\bibinfo {author} {\bibfnamefont {S.}~\bibnamefont
  {Barison}}, \bibinfo {author} {\bibfnamefont {F.}~\bibnamefont {Vicentini}},\
  and\ \bibinfo {author} {\bibfnamefont {G.}~\bibnamefont {Carleo}},\
  }\bibfield  {title} {\bibinfo {title} {An efficient quantum algorithm for the
  time evolution of parameterized circuits},\ }\href
  {https://doi.org/10.22331/q-2021-07-28-512} {\bibfield  {journal} {\bibinfo
  {journal} {{Quantum}}\ }\textbf {\bibinfo {volume} {5}},\ \bibinfo {pages}
  {512} (\bibinfo {year} {2021})}\BibitemShut {NoStop}%
\bibitem [{\citenamefont {Lin}\ \emph {et~al.}(2021)\citenamefont {Lin},
  \citenamefont {Dilip}, \citenamefont {Green}, \citenamefont {Smith},\ and\
  \citenamefont {Pollmann}}]{lin2021real}%
  \BibitemOpen
  \bibfield  {author} {\bibinfo {author} {\bibfnamefont {S.-H.}\ \bibnamefont
  {Lin}}, \bibinfo {author} {\bibfnamefont {R.}~\bibnamefont {Dilip}}, \bibinfo
  {author} {\bibfnamefont {A.~G.}\ \bibnamefont {Green}}, \bibinfo {author}
  {\bibfnamefont {A.}~\bibnamefont {Smith}},\ and\ \bibinfo {author}
  {\bibfnamefont {F.}~\bibnamefont {Pollmann}},\ }\bibfield  {title} {\bibinfo
  {title} {Real- and imaginary-time evolution with compressed quantum
  circuits},\ }\href {https://doi.org/10.1103/PRXQuantum.2.010342} {\bibfield
  {journal} {\bibinfo  {journal} {PRX Quantum}\ }\textbf {\bibinfo {volume}
  {2}},\ \bibinfo {pages} {010342} (\bibinfo {year} {2021})}\BibitemShut
  {NoStop}%
\bibitem [{\citenamefont {Larocca}\ \emph
  {et~al.}(2021{\natexlab{b}})\citenamefont {Larocca}, \citenamefont {Ju},
  \citenamefont {García-Martín}, \citenamefont {Coles},\ and\ \citenamefont
  {Cerezo}}]{larocca2021theory}%
  \BibitemOpen
  \bibfield  {author} {\bibinfo {author} {\bibfnamefont {M.}~\bibnamefont
  {Larocca}}, \bibinfo {author} {\bibfnamefont {N.}~\bibnamefont {Ju}},
  \bibinfo {author} {\bibfnamefont {D.}~\bibnamefont {García-Martín}},
  \bibinfo {author} {\bibfnamefont {P.~J.}\ \bibnamefont {Coles}},\ and\
  \bibinfo {author} {\bibfnamefont {M.}~\bibnamefont {Cerezo}},\ }\href@noop {}
  {\bibinfo {title} {Theory of overparametrization in quantum neural networks}}
  (\bibinfo {year} {2021}{\natexlab{b}}),\ \Eprint
  {https://arxiv.org/abs/2109.11676} {arXiv:2109.11676 [quant-ph]} \BibitemShut
  {NoStop}%
\bibitem [{\citenamefont {Chen}\ \emph {et~al.}(2021)\citenamefont {Chen},
  \citenamefont {Yu}, \citenamefont {Zeng},\ and\ \citenamefont
  {Flammia}}]{chen2021robust}%
  \BibitemOpen
  \bibfield  {author} {\bibinfo {author} {\bibfnamefont {S.}~\bibnamefont
  {Chen}}, \bibinfo {author} {\bibfnamefont {W.}~\bibnamefont {Yu}}, \bibinfo
  {author} {\bibfnamefont {P.}~\bibnamefont {Zeng}},\ and\ \bibinfo {author}
  {\bibfnamefont {S.~T.}\ \bibnamefont {Flammia}},\ }\bibfield  {title}
  {\bibinfo {title} {Robust shadow estimation},\ }\href
  {https://doi.org/10.1103/PRXQuantum.2.030348} {\bibfield  {journal} {\bibinfo
   {journal} {PRX Quantum}\ }\textbf {\bibinfo {volume} {2}},\ \bibinfo {pages}
  {030348} (\bibinfo {year} {2021})}\BibitemShut {NoStop}%
\bibitem [{\citenamefont {{Enshan Koh}}\ and\ \citenamefont
  {{Grewal}}(2020)}]{koh2020classical}%
  \BibitemOpen
  \bibfield  {author} {\bibinfo {author} {\bibfnamefont {D.}~\bibnamefont
  {{Enshan Koh}}}\ and\ \bibinfo {author} {\bibfnamefont {S.}~\bibnamefont
  {{Grewal}}},\ }\bibfield  {title} {\bibinfo {title} {{Classical Shadows with
  Noise}},\ }\href@noop {} {\bibfield  {journal} {\bibinfo  {journal} {arXiv
  e-prints}\ ,\ \bibinfo {eid} {arXiv:2011.11580}} (\bibinfo {year} {2020})},\
  \Eprint {https://arxiv.org/abs/2011.11580} {arXiv:2011.11580 [quant-ph]}
  \BibitemShut {NoStop}%
\bibitem [{\citenamefont {{Mi \emph{et al.}}}(2021)}]{mi2021information}%
  \BibitemOpen
  \bibfield  {author} {\bibinfo {author} {\bibfnamefont {X.}~\bibnamefont {{Mi
  \emph{et al.}}}},\ }\bibfield  {title} {\bibinfo {title} {{Information
  Scrambling in Computationally Complex Quantum Circuits}},\ }\href@noop {}
  {\bibfield  {journal} {\bibinfo  {journal} {arXiv e-prints}\ ,\ \bibinfo
  {eid} {arXiv:2101.08870}} (\bibinfo {year} {2021})},\ \Eprint
  {https://arxiv.org/abs/2101.08870} {arXiv:2101.08870 [quant-ph]} \BibitemShut
  {NoStop}%
\bibitem [{\citenamefont {Bezanson}\ \emph {et~al.}(2017)\citenamefont
  {Bezanson}, \citenamefont {Edelman}, \citenamefont {Karpinski},\ and\
  \citenamefont {Shah}}]{julia}%
  \BibitemOpen
  \bibfield  {author} {\bibinfo {author} {\bibfnamefont {J.}~\bibnamefont
  {Bezanson}}, \bibinfo {author} {\bibfnamefont {A.}~\bibnamefont {Edelman}},
  \bibinfo {author} {\bibfnamefont {S.}~\bibnamefont {Karpinski}},\ and\
  \bibinfo {author} {\bibfnamefont {V.~B.}\ \bibnamefont {Shah}},\ }\bibfield
  {title} {\bibinfo {title} {Julia: A fresh approach to numerical computing},\
  }\href {https://doi.org/10.1137/141000671} {\bibfield  {journal} {\bibinfo
  {journal} {SIAM review}\ }\textbf {\bibinfo {volume} {59}},\ \bibinfo {pages}
  {65} (\bibinfo {year} {2017})}\BibitemShut {NoStop}%
\bibitem [{\citenamefont {Luo}\ \emph {et~al.}(2020)\citenamefont {Luo},
  \citenamefont {Liu}, \citenamefont {Zhang},\ and\ \citenamefont
  {Wang}}]{yao}%
  \BibitemOpen
  \bibfield  {author} {\bibinfo {author} {\bibfnamefont {X.-Z.}\ \bibnamefont
  {Luo}}, \bibinfo {author} {\bibfnamefont {J.-G.}\ \bibnamefont {Liu}},
  \bibinfo {author} {\bibfnamefont {P.}~\bibnamefont {Zhang}},\ and\ \bibinfo
  {author} {\bibfnamefont {L.}~\bibnamefont {Wang}},\ }\bibfield  {title}
  {\bibinfo {title} {Yao.jl: {E}xtensible, {E}fficient {F}ramework for
  {Q}uantum {A}lgorithm {D}esign},\ }\href
  {https://doi.org/10.22331/q-2020-10-11-341} {\bibfield  {journal} {\bibinfo
  {journal} {{Quantum}}\ }\textbf {\bibinfo {volume} {4}},\ \bibinfo {pages}
  {341} (\bibinfo {year} {2020})}\BibitemShut {NoStop}%
\bibitem [{\citenamefont {{Aaronson}}(2017)}]{aaronson2018shadow}%
  \BibitemOpen
  \bibfield  {author} {\bibinfo {author} {\bibfnamefont {S.}~\bibnamefont
  {{Aaronson}}},\ }\bibfield  {title} {\bibinfo {title} {{Shadow Tomography of
  Quantum States}},\ }\href@noop {} {\bibfield  {journal} {\bibinfo  {journal}
  {arXiv e-prints}\ ,\ \bibinfo {eid} {arXiv:1711.01053}} (\bibinfo {year}
  {2017})},\ \Eprint {https://arxiv.org/abs/1711.01053} {arXiv:1711.01053
  [quant-ph]} \BibitemShut {NoStop}%
\bibitem [{\citenamefont {{Aaronson}}\ and\ \citenamefont
  {{Rothblum}}(2019)}]{aaronson2019gentle}%
  \BibitemOpen
  \bibfield  {author} {\bibinfo {author} {\bibfnamefont {S.}~\bibnamefont
  {{Aaronson}}}\ and\ \bibinfo {author} {\bibfnamefont {G.~N.}\ \bibnamefont
  {{Rothblum}}},\ }\bibfield  {title} {\bibinfo {title} {{Gentle Measurement of
  Quantum States and Differential Privacy}},\ }\href@noop {} {\bibfield
  {journal} {\bibinfo  {journal} {arXiv e-prints}\ ,\ \bibinfo {eid}
  {arXiv:1904.08747}} (\bibinfo {year} {2019})},\ \Eprint
  {https://arxiv.org/abs/1904.08747} {arXiv:1904.08747 [quant-ph]} \BibitemShut
  {NoStop}%
\bibitem [{\citenamefont {{Paini}}\ and\ \citenamefont
  {{Kalev}}(2019)}]{paini2019approximate}%
  \BibitemOpen
  \bibfield  {author} {\bibinfo {author} {\bibfnamefont {M.}~\bibnamefont
  {{Paini}}}\ and\ \bibinfo {author} {\bibfnamefont {A.}~\bibnamefont
  {{Kalev}}},\ }\bibfield  {title} {\bibinfo {title} {{An approximate
  description of quantum states}},\ }\href@noop {} {\bibfield  {journal}
  {\bibinfo  {journal} {arXiv e-prints}\ ,\ \bibinfo {eid} {arXiv:1910.10543}}
  (\bibinfo {year} {2019})},\ \Eprint {https://arxiv.org/abs/1910.10543}
  {arXiv:1910.10543 [quant-ph]} \BibitemShut {NoStop}%
\bibitem [{\citenamefont {{Morris}}\ and\ \citenamefont
  {{Daki{\'c}}}(2019)}]{morris2019selective}%
  \BibitemOpen
  \bibfield  {author} {\bibinfo {author} {\bibfnamefont {J.}~\bibnamefont
  {{Morris}}}\ and\ \bibinfo {author} {\bibfnamefont {B.}~\bibnamefont
  {{Daki{\'c}}}},\ }\bibfield  {title} {\bibinfo {title} {{Selective Quantum
  State Tomography}},\ }\href@noop {} {\bibfield  {journal} {\bibinfo
  {journal} {arXiv e-prints}\ ,\ \bibinfo {eid} {arXiv:1909.05880}} (\bibinfo
  {year} {2019})},\ \Eprint {https://arxiv.org/abs/1909.05880}
  {arXiv:1909.05880 [quant-ph]} \BibitemShut {NoStop}%
\bibitem [{\citenamefont {{Elben}}\ \emph {et~al.}(2020)\citenamefont
  {{Elben}}, \citenamefont {{Kueng}}, \citenamefont {{Huang}}, \citenamefont
  {{van Bijnen}}, \citenamefont {{Kokail}}, \citenamefont {{Dalmonte}},
  \citenamefont {{Calabrese}}, \citenamefont {{Kraus}}, \citenamefont
  {{Preskill}}, \citenamefont {{Zoller}},\ and\ \citenamefont
  {{Vermersch}}}]{elben2020mixed-state}%
  \BibitemOpen
  \bibfield  {author} {\bibinfo {author} {\bibfnamefont {A.}~\bibnamefont
  {{Elben}}}, \bibinfo {author} {\bibfnamefont {R.}~\bibnamefont {{Kueng}}},
  \bibinfo {author} {\bibfnamefont {H.-Y.~R.}\ \bibnamefont {{Huang}}},
  \bibinfo {author} {\bibfnamefont {R.}~\bibnamefont {{van Bijnen}}}, \bibinfo
  {author} {\bibfnamefont {C.}~\bibnamefont {{Kokail}}}, \bibinfo {author}
  {\bibfnamefont {M.}~\bibnamefont {{Dalmonte}}}, \bibinfo {author}
  {\bibfnamefont {P.}~\bibnamefont {{Calabrese}}}, \bibinfo {author}
  {\bibfnamefont {B.}~\bibnamefont {{Kraus}}}, \bibinfo {author} {\bibfnamefont
  {J.}~\bibnamefont {{Preskill}}}, \bibinfo {author} {\bibfnamefont
  {P.}~\bibnamefont {{Zoller}}},\ and\ \bibinfo {author} {\bibfnamefont
  {B.}~\bibnamefont {{Vermersch}}},\ }\bibfield  {title} {\bibinfo {title}
  {{Mixed-State Entanglement from Local Randomized Measurements}},\ }\href
  {https://doi.org/10.1103/PhysRevLett.125.200501} {\bibfield  {journal}
  {\bibinfo  {journal} {\prl}\ }\textbf {\bibinfo {volume} {125}},\ \bibinfo
  {eid} {200501} (\bibinfo {year} {2020})},\ \Eprint
  {https://arxiv.org/abs/2007.06305} {arXiv:2007.06305 [quant-ph]} \BibitemShut
  {NoStop}%
\bibitem [{\citenamefont {{Huang}}\ \emph
  {et~al.}(2021{\natexlab{b}})\citenamefont {{Huang}}, \citenamefont {{Kueng}},
  \citenamefont {{Torlai}}, \citenamefont {{Albert}},\ and\ \citenamefont
  {{Preskill}}}]{huang2021provably}%
  \BibitemOpen
  \bibfield  {author} {\bibinfo {author} {\bibfnamefont {H.-Y.}\ \bibnamefont
  {{Huang}}}, \bibinfo {author} {\bibfnamefont {R.}~\bibnamefont {{Kueng}}},
  \bibinfo {author} {\bibfnamefont {G.}~\bibnamefont {{Torlai}}}, \bibinfo
  {author} {\bibfnamefont {V.~V.}\ \bibnamefont {{Albert}}},\ and\ \bibinfo
  {author} {\bibfnamefont {J.}~\bibnamefont {{Preskill}}},\ }\bibfield  {title}
  {\bibinfo {title} {{Provably efficient machine learning for quantum many-body
  problems}},\ }\href@noop {} {\bibfield  {journal} {\bibinfo  {journal} {arXiv
  e-prints}\ ,\ \bibinfo {eid} {arXiv:2106.12627}} (\bibinfo {year}
  {2021}{\natexlab{b}})},\ \Eprint {https://arxiv.org/abs/2106.12627}
  {arXiv:2106.12627 [quant-ph]} \BibitemShut {NoStop}%
\bibitem [{\citenamefont {{Evans}}\ \emph {et~al.}(2019)\citenamefont
  {{Evans}}, \citenamefont {{Harper}},\ and\ \citenamefont
  {{Flammia}}}]{evans2019scalable}%
  \BibitemOpen
  \bibfield  {author} {\bibinfo {author} {\bibfnamefont {T.~J.}\ \bibnamefont
  {{Evans}}}, \bibinfo {author} {\bibfnamefont {R.}~\bibnamefont {{Harper}}},\
  and\ \bibinfo {author} {\bibfnamefont {S.~T.}\ \bibnamefont {{Flammia}}},\
  }\bibfield  {title} {\bibinfo {title} {{Scalable Bayesian Hamiltonian
  learning}},\ }\href@noop {} {\bibfield  {journal} {\bibinfo  {journal} {arXiv
  e-prints}\ ,\ \bibinfo {eid} {arXiv:1912.07636}} (\bibinfo {year} {2019})},\
  \Eprint {https://arxiv.org/abs/1912.07636} {arXiv:1912.07636 [quant-ph]}
  \BibitemShut {NoStop}%
\bibitem [{\citenamefont {Huang}\ \emph {et~al.}(2021)\citenamefont {Huang},
  \citenamefont {Kueng},\ and\ \citenamefont
  {Preskill}}]{huang2021derandomization}%
  \BibitemOpen
  \bibfield  {author} {\bibinfo {author} {\bibfnamefont {H.-Y.}\ \bibnamefont
  {Huang}}, \bibinfo {author} {\bibfnamefont {R.}~\bibnamefont {Kueng}},\ and\
  \bibinfo {author} {\bibfnamefont {J.}~\bibnamefont {Preskill}},\ }\bibfield
  {title} {\bibinfo {title} {Efficient estimation of pauli observables by
  derandomization},\ }\href {https://doi.org/10.1103/PhysRevLett.127.030503}
  {\bibfield  {journal} {\bibinfo  {journal} {Phys. Rev. Lett.}\ }\textbf
  {\bibinfo {volume} {127}},\ \bibinfo {pages} {030503} (\bibinfo {year}
  {2021})}\BibitemShut {NoStop}%
\bibitem [{\citenamefont {{Acharya}}\ \emph {et~al.}(2021)\citenamefont
  {{Acharya}}, \citenamefont {{Saha}},\ and\ \citenamefont
  {{Sengupta}}}]{acharya2021informationally}%
  \BibitemOpen
  \bibfield  {author} {\bibinfo {author} {\bibfnamefont {A.}~\bibnamefont
  {{Acharya}}}, \bibinfo {author} {\bibfnamefont {S.}~\bibnamefont {{Saha}}},\
  and\ \bibinfo {author} {\bibfnamefont {A.~M.}\ \bibnamefont {{Sengupta}}},\
  }\bibfield  {title} {\bibinfo {title} {{Informationally complete POVM-based
  shadow tomography}},\ }\href@noop {} {\bibfield  {journal} {\bibinfo
  {journal} {arXiv e-prints}\ ,\ \bibinfo {eid} {arXiv:2105.05992}} (\bibinfo
  {year} {2021})},\ \Eprint {https://arxiv.org/abs/2105.05992}
  {arXiv:2105.05992 [quant-ph]} \BibitemShut {NoStop}%
\bibitem [{\citenamefont {Neven}\ \emph {et~al.}(2021)\citenamefont {Neven},
  \citenamefont {Carrasco}, \citenamefont {Vitale}, \citenamefont {Kokail},
  \citenamefont {Elben}, \citenamefont {Dalmonte}, \citenamefont {Calabrese},
  \citenamefont {Zoller}, \citenamefont {Vermersch}, \citenamefont {Kueng},\
  and\ \citenamefont {Kraus}}]{neven2021symmetry}%
  \BibitemOpen
  \bibfield  {author} {\bibinfo {author} {\bibfnamefont {A.}~\bibnamefont
  {Neven}}, \bibinfo {author} {\bibfnamefont {J.}~\bibnamefont {Carrasco}},
  \bibinfo {author} {\bibfnamefont {V.}~\bibnamefont {Vitale}}, \bibinfo
  {author} {\bibfnamefont {C.}~\bibnamefont {Kokail}}, \bibinfo {author}
  {\bibfnamefont {A.}~\bibnamefont {Elben}}, \bibinfo {author} {\bibfnamefont
  {M.}~\bibnamefont {Dalmonte}}, \bibinfo {author} {\bibfnamefont
  {P.}~\bibnamefont {Calabrese}}, \bibinfo {author} {\bibfnamefont
  {P.}~\bibnamefont {Zoller}}, \bibinfo {author} {\bibfnamefont
  {B.}~\bibnamefont {Vermersch}}, \bibinfo {author} {\bibfnamefont
  {R.}~\bibnamefont {Kueng}},\ and\ \bibinfo {author} {\bibfnamefont
  {B.}~\bibnamefont {Kraus}},\ }\bibfield  {title} {\bibinfo {title}
  {Symmetry-resolved entanglement detection using partial transpose moments},\
  }\href {https://doi.org/10.1038/s41534-021-00487-y} {\bibfield  {journal}
  {\bibinfo  {journal} {npj Quantum Information}\ }\textbf {\bibinfo {volume}
  {7}},\ \bibinfo {pages} {152} (\bibinfo {year} {2021})}\BibitemShut {NoStop}%
\bibitem [{\citenamefont {Elben}\ \emph {et~al.}(2020)\citenamefont {Elben},
  \citenamefont {Kueng}, \citenamefont {Huang}, \citenamefont {van Bijnen},
  \citenamefont {Kokail}, \citenamefont {Dalmonte}, \citenamefont {Calabrese},
  \citenamefont {Kraus}, \citenamefont {Preskill}, \citenamefont {Zoller},\
  and\ \citenamefont {Vermersch}}]{elben2020mixed}%
  \BibitemOpen
  \bibfield  {author} {\bibinfo {author} {\bibfnamefont {A.}~\bibnamefont
  {Elben}}, \bibinfo {author} {\bibfnamefont {R.}~\bibnamefont {Kueng}},
  \bibinfo {author} {\bibfnamefont {H.-Y.~R.}\ \bibnamefont {Huang}}, \bibinfo
  {author} {\bibfnamefont {R.}~\bibnamefont {van Bijnen}}, \bibinfo {author}
  {\bibfnamefont {C.}~\bibnamefont {Kokail}}, \bibinfo {author} {\bibfnamefont
  {M.}~\bibnamefont {Dalmonte}}, \bibinfo {author} {\bibfnamefont
  {P.}~\bibnamefont {Calabrese}}, \bibinfo {author} {\bibfnamefont
  {B.}~\bibnamefont {Kraus}}, \bibinfo {author} {\bibfnamefont
  {J.}~\bibnamefont {Preskill}}, \bibinfo {author} {\bibfnamefont
  {P.}~\bibnamefont {Zoller}},\ and\ \bibinfo {author} {\bibfnamefont
  {B.}~\bibnamefont {Vermersch}},\ }\bibfield  {title} {\bibinfo {title}
  {Mixed-state entanglement from local randomized measurements},\ }\href
  {https://doi.org/10.1103/PhysRevLett.125.200501} {\bibfield  {journal}
  {\bibinfo  {journal} {Phys. Rev. Lett.}\ }\textbf {\bibinfo {volume} {125}},\
  \bibinfo {pages} {200501} (\bibinfo {year} {2020})}\BibitemShut {NoStop}%
\bibitem [{\citenamefont {Foucart}\ and\ \citenamefont
  {Rauhut}(2013)}]{rauhut2013book}%
  \BibitemOpen
  \bibfield  {author} {\bibinfo {author} {\bibfnamefont {S.}~\bibnamefont
  {Foucart}}\ and\ \bibinfo {author} {\bibfnamefont {H.}~\bibnamefont
  {Rauhut}},\ }\href {https://doi.org/10.1007/978-0-8176-4948-7} {\emph
  {\bibinfo {title} {A mathematical introduction to compressive sensing}}},\
  Applied and Numerical Harmonic Analysis\ (\bibinfo  {publisher}
  {Birkh\"{a}user/Springer, New York},\ \bibinfo {year} {2013})\ pp.\ \bibinfo
  {pages} {xviii+625}\BibitemShut {NoStop}%
\bibitem [{\citenamefont {Dankert}\ \emph
  {et~al.}(2009{\natexlab{a}})\citenamefont {Dankert}, \citenamefont {Cleve},
  \citenamefont {Emerson},\ and\ \citenamefont {Livine}}]{dankert09unitary}%
  \BibitemOpen
  \bibfield  {author} {\bibinfo {author} {\bibfnamefont {C.}~\bibnamefont
  {Dankert}}, \bibinfo {author} {\bibfnamefont {R.}~\bibnamefont {Cleve}},
  \bibinfo {author} {\bibfnamefont {J.}~\bibnamefont {Emerson}},\ and\ \bibinfo
  {author} {\bibfnamefont {E.}~\bibnamefont {Livine}},\ }\bibfield  {title}
  {\bibinfo {title} {Exact and approximate unitary 2-designs and their
  application to fidelity estimation},\ }\href
  {https://doi.org/10.1103/PhysRevA.80.012304} {\bibfield  {journal} {\bibinfo
  {journal} {Phys. Rev. A}\ }\textbf {\bibinfo {volume} {80}},\ \bibinfo
  {pages} {012304} (\bibinfo {year} {2009}{\natexlab{a}})}\BibitemShut
  {NoStop}%
\bibitem [{\citenamefont {Gross}\ \emph {et~al.}(2007)\citenamefont {Gross},
  \citenamefont {Audenaert},\ and\ \citenamefont {Eisert}}]{gross07evenly}%
  \BibitemOpen
  \bibfield  {author} {\bibinfo {author} {\bibfnamefont {D.}~\bibnamefont
  {Gross}}, \bibinfo {author} {\bibfnamefont {K.}~\bibnamefont {Audenaert}},\
  and\ \bibinfo {author} {\bibfnamefont {J.}~\bibnamefont {Eisert}},\
  }\bibfield  {title} {\bibinfo {title} {Evenly distributed unitaries: on the
  structure of unitary designs},\ }\href {https://doi.org/10.1063/1.2716992}
  {\bibfield  {journal} {\bibinfo  {journal} {J. Math. Phys.}\ }\textbf
  {\bibinfo {volume} {48}},\ \bibinfo {pages} {052104, 22} (\bibinfo {year}
  {2007})}\BibitemShut {NoStop}%
\bibitem [{\citenamefont {Dankert}\ \emph
  {et~al.}(2009{\natexlab{b}})\citenamefont {Dankert}, \citenamefont {Cleve},
  \citenamefont {Emerson},\ and\ \citenamefont {Livine}}]{dankert2009exact}%
  \BibitemOpen
  \bibfield  {author} {\bibinfo {author} {\bibfnamefont {C.}~\bibnamefont
  {Dankert}}, \bibinfo {author} {\bibfnamefont {R.}~\bibnamefont {Cleve}},
  \bibinfo {author} {\bibfnamefont {J.}~\bibnamefont {Emerson}},\ and\ \bibinfo
  {author} {\bibfnamefont {E.}~\bibnamefont {Livine}},\ }\bibfield  {title}
  {\bibinfo {title} {Exact and approximate unitary 2-designs and their
  application to fidelity estimation},\ }\href
  {https://doi.org/10.1103/PhysRevA.80.012304} {\bibfield  {journal} {\bibinfo
  {journal} {Phys. Rev. A}\ }\textbf {\bibinfo {volume} {80}},\ \bibinfo
  {pages} {012304} (\bibinfo {year} {2009}{\natexlab{b}})}\BibitemShut
  {NoStop}%
\bibitem [{\citenamefont {Cleve}\ \emph {et~al.}(2016)\citenamefont {Cleve},
  \citenamefont {Leung}, \citenamefont {Liu},\ and\ \citenamefont
  {Wang}}]{cleve2016nearlinear}%
  \BibitemOpen
  \bibfield  {author} {\bibinfo {author} {\bibfnamefont {R.}~\bibnamefont
  {Cleve}}, \bibinfo {author} {\bibfnamefont {D.}~\bibnamefont {Leung}},
  \bibinfo {author} {\bibfnamefont {L.}~\bibnamefont {Liu}},\ and\ \bibinfo
  {author} {\bibfnamefont {C.}~\bibnamefont {Wang}},\ }\href@noop {} {\bibinfo
  {title} {Near-linear constructions of exact unitary 2-designs}} (\bibinfo
  {year} {2016}),\ \Eprint {https://arxiv.org/abs/1501.04592} {arXiv:1501.04592
  [quant-ph]} \BibitemShut {NoStop}%
\bibitem [{\citenamefont {{Kueng}}\ and\ \citenamefont
  {{Gross}}(2015)}]{kueng2015qubit}%
  \BibitemOpen
  \bibfield  {author} {\bibinfo {author} {\bibfnamefont {R.}~\bibnamefont
  {{Kueng}}}\ and\ \bibinfo {author} {\bibfnamefont {D.}~\bibnamefont
  {{Gross}}},\ }\bibfield  {title} {\bibinfo {title} {{Qubit stabilizer states
  are complex projective 3-designs}},\ }\href@noop {} {\bibfield  {journal}
  {\bibinfo  {journal} {arXiv e-prints}\ ,\ \bibinfo {eid} {arXiv:1510.02767}}
  (\bibinfo {year} {2015})},\ \Eprint {https://arxiv.org/abs/1510.02767}
  {arXiv:1510.02767 [quant-ph]} \BibitemShut {NoStop}%
\bibitem [{\citenamefont {{Webb}}(2015)}]{webb2016clifford}%
  \BibitemOpen
  \bibfield  {author} {\bibinfo {author} {\bibfnamefont {Z.}~\bibnamefont
  {{Webb}}},\ }\bibfield  {title} {\bibinfo {title} {{The Clifford group forms
  a unitary 3-design}},\ }\href@noop {} {\bibfield  {journal} {\bibinfo
  {journal} {arXiv e-prints}\ ,\ \bibinfo {eid} {arXiv:1510.02769}} (\bibinfo
  {year} {2015})},\ \Eprint {https://arxiv.org/abs/1510.02769}
  {arXiv:1510.02769 [quant-ph]} \BibitemShut {NoStop}%
\bibitem [{\citenamefont {Zhu}(2017)}]{zhu17clifford}%
  \BibitemOpen
  \bibfield  {author} {\bibinfo {author} {\bibfnamefont {H.}~\bibnamefont
  {Zhu}},\ }\bibfield  {title} {\bibinfo {title} {Multiqubit clifford groups
  are unitary 3-designs},\ }\href {https://doi.org/10.1103/PhysRevA.96.062336}
  {\bibfield  {journal} {\bibinfo  {journal} {Phys. Rev. A}\ }\textbf {\bibinfo
  {volume} {96}},\ \bibinfo {pages} {062336} (\bibinfo {year}
  {2017})}\BibitemShut {NoStop}%
\bibitem [{\citenamefont {Watrous}(2018)}]{watrous_2018}%
  \BibitemOpen
  \bibfield  {author} {\bibinfo {author} {\bibfnamefont {J.}~\bibnamefont
  {Watrous}},\ }\href {https://doi.org/10.1017/9781316848142} {\emph {\bibinfo
  {title} {The Theory of Quantum Information}}}\ (\bibinfo  {publisher}
  {Cambridge University Press},\ \bibinfo {year} {2018})\BibitemShut {NoStop}%
\bibitem [{\citenamefont {Page}(1993)}]{page1993}%
  \BibitemOpen
  \bibfield  {author} {\bibinfo {author} {\bibfnamefont {D.~N.}\ \bibnamefont
  {Page}},\ }\bibfield  {title} {\bibinfo {title} {Average entropy of a
  subsystem},\ }\href {https://doi.org/10.1103/PhysRevLett.71.1291} {\bibfield
  {journal} {\bibinfo  {journal} {Phys. Rev. Lett.}\ }\textbf {\bibinfo
  {volume} {71}},\ \bibinfo {pages} {1291} (\bibinfo {year}
  {1993})}\BibitemShut {NoStop}%
\bibitem [{\citenamefont {Hayden}\ and\ \citenamefont
  {Preskill}(2007)}]{hayden2007blackholes}%
  \BibitemOpen
  \bibfield  {author} {\bibinfo {author} {\bibfnamefont {P.}~\bibnamefont
  {Hayden}}\ and\ \bibinfo {author} {\bibfnamefont {J.}~\bibnamefont
  {Preskill}},\ }\bibfield  {title} {\bibinfo {title} {Black holes as mirrors:
  quantum information in random subsystems},\ }\href
  {https://doi.org/10.1088/1126-6708/2007/09/120} {\bibfield  {journal}
  {\bibinfo  {journal} {Journal of High Energy Physics}\ }\textbf {\bibinfo
  {volume} {2007}},\ \bibinfo {pages} {120} (\bibinfo {year}
  {2007})}\BibitemShut {NoStop}%
\bibitem [{\citenamefont {Sekino}\ and\ \citenamefont
  {Susskind}(2008)}]{sekino2008fastscramblers}%
  \BibitemOpen
  \bibfield  {author} {\bibinfo {author} {\bibfnamefont {Y.}~\bibnamefont
  {Sekino}}\ and\ \bibinfo {author} {\bibfnamefont {L.}~\bibnamefont
  {Susskind}},\ }\bibfield  {title} {\bibinfo {title} {Fast scramblers},\
  }\href {https://doi.org/10.1088/1126-6708/2008/10/065} {\bibfield  {journal}
  {\bibinfo  {journal} {Journal of High Energy Physics}\ }\textbf {\bibinfo
  {volume} {2008}},\ \bibinfo {pages} {065} (\bibinfo {year}
  {2008})}\BibitemShut {NoStop}%
\bibitem [{\citenamefont {Nahum}\ \emph {et~al.}(2017)\citenamefont {Nahum},
  \citenamefont {Ruhman}, \citenamefont {Vijay},\ and\ \citenamefont
  {Haah}}]{nahum_2017_RUC}%
  \BibitemOpen
  \bibfield  {author} {\bibinfo {author} {\bibfnamefont {A.}~\bibnamefont
  {Nahum}}, \bibinfo {author} {\bibfnamefont {J.}~\bibnamefont {Ruhman}},
  \bibinfo {author} {\bibfnamefont {S.}~\bibnamefont {Vijay}},\ and\ \bibinfo
  {author} {\bibfnamefont {J.}~\bibnamefont {Haah}},\ }\bibfield  {title}
  {\bibinfo {title} {Quantum entanglement growth under random unitary
  dynamics},\ }\href {https://doi.org/10.1103/PhysRevX.7.031016} {\bibfield
  {journal} {\bibinfo  {journal} {Phys. Rev. X}\ }\textbf {\bibinfo {volume}
  {7}},\ \bibinfo {pages} {031016} (\bibinfo {year} {2017})}\BibitemShut
  {NoStop}%
\bibitem [{\citenamefont {Nahum}\ \emph {et~al.}(2018)\citenamefont {Nahum},
  \citenamefont {Vijay},\ and\ \citenamefont {Haah}}]{nahum_2018_RUC}%
  \BibitemOpen
  \bibfield  {author} {\bibinfo {author} {\bibfnamefont {A.}~\bibnamefont
  {Nahum}}, \bibinfo {author} {\bibfnamefont {S.}~\bibnamefont {Vijay}},\ and\
  \bibinfo {author} {\bibfnamefont {J.}~\bibnamefont {Haah}},\ }\bibfield
  {title} {\bibinfo {title} {Operator spreading in random unitary circuits},\
  }\href {https://doi.org/10.1103/PhysRevX.8.021014} {\bibfield  {journal}
  {\bibinfo  {journal} {Phys. Rev. X}\ }\textbf {\bibinfo {volume} {8}},\
  \bibinfo {pages} {021014} (\bibinfo {year} {2018})}\BibitemShut {NoStop}%
\bibitem [{\citenamefont {Hosur}\ \emph {et~al.}(2016)\citenamefont {Hosur},
  \citenamefont {Qi}, \citenamefont {Roberts},\ and\ \citenamefont
  {Yoshida}}]{hosur2016}%
  \BibitemOpen
  \bibfield  {author} {\bibinfo {author} {\bibfnamefont {P.}~\bibnamefont
  {Hosur}}, \bibinfo {author} {\bibfnamefont {X.-L.}\ \bibnamefont {Qi}},
  \bibinfo {author} {\bibfnamefont {D.~A.}\ \bibnamefont {Roberts}},\ and\
  \bibinfo {author} {\bibfnamefont {B.}~\bibnamefont {Yoshida}},\ }\bibfield
  {title} {\bibinfo {title} {Chaos in quantum channels},\ }\bibfield  {journal}
  {\bibinfo  {journal} {Journal of High Energy Physics}\ }\textbf {\bibinfo
  {volume} {2016}},\ \href {https://doi.org/10.1007/jhep02(2016)004}
  {10.1007/jhep02(2016)004} (\bibinfo {year} {2016})\BibitemShut {NoStop}%
\bibitem [{\citenamefont {von Keyserlingk}\ \emph {et~al.}(2018)\citenamefont
  {von Keyserlingk}, \citenamefont {Rakovszky}, \citenamefont {Pollmann},\ and\
  \citenamefont {Sondhi}}]{pollman2018otocs}%
  \BibitemOpen
  \bibfield  {author} {\bibinfo {author} {\bibfnamefont {C.}~\bibnamefont {von
  Keyserlingk}}, \bibinfo {author} {\bibfnamefont {T.}~\bibnamefont
  {Rakovszky}}, \bibinfo {author} {\bibfnamefont {F.}~\bibnamefont
  {Pollmann}},\ and\ \bibinfo {author} {\bibfnamefont {S.}~\bibnamefont
  {Sondhi}},\ }\bibfield  {title} {\bibinfo {title} {Operator hydrodynamics,
  otocs, and entanglement growth in systems without conservation laws},\
  }\bibfield  {journal} {\bibinfo  {journal} {Physical Review X}\ }\textbf
  {\bibinfo {volume} {8}},\ \href {https://doi.org/10.1103/physrevx.8.021013}
  {10.1103/physrevx.8.021013} (\bibinfo {year} {2018})\BibitemShut {NoStop}%
\bibitem [{\citenamefont {Foong}\ and\ \citenamefont
  {Kanno}(1994)}]{proofPage_1994}%
  \BibitemOpen
  \bibfield  {author} {\bibinfo {author} {\bibfnamefont {S.~K.}\ \bibnamefont
  {Foong}}\ and\ \bibinfo {author} {\bibfnamefont {S.}~\bibnamefont {Kanno}},\
  }\bibfield  {title} {\bibinfo {title} {Proof of page's conjecture on the
  average entropy of a subsystem},\ }\href
  {https://doi.org/10.1103/PhysRevLett.72.1148} {\bibfield  {journal} {\bibinfo
   {journal} {Phys. Rev. Lett.}\ }\textbf {\bibinfo {volume} {72}},\ \bibinfo
  {pages} {1148} (\bibinfo {year} {1994})}\BibitemShut {NoStop}%
\bibitem [{\citenamefont {Sachdev}\ and\ \citenamefont
  {Ye}(1993)}]{sachdevandye_1993}%
  \BibitemOpen
  \bibfield  {author} {\bibinfo {author} {\bibfnamefont {S.}~\bibnamefont
  {Sachdev}}\ and\ \bibinfo {author} {\bibfnamefont {J.}~\bibnamefont {Ye}},\
  }\bibfield  {title} {\bibinfo {title} {Gapless spin-fluid ground state in a
  random quantum heisenberg magnet},\ }\href
  {https://doi.org/10.1103/PhysRevLett.70.3339} {\bibfield  {journal} {\bibinfo
   {journal} {Phys. Rev. Lett.}\ }\textbf {\bibinfo {volume} {70}},\ \bibinfo
  {pages} {3339} (\bibinfo {year} {1993})}\BibitemShut {NoStop}%
\end{thebibliography}
\end{document}